\documentclass[a4paper,UKenglish]{lipics-v2016}
\usepackage{amsmath, amssymb, stmaryrd, mathpartir, commath, url, longtable,hyperref, paralist}
\usepackage{microtype}
\usepackage[dvipsnames]{xcolor}
\usepackage{csquotes}
\usepackage{cite}
\hypersetup{colorlinks=true,citecolor=blue}
\RequirePackage{subcaption}
\usepackage{graphicx}
\usepackage{xspace}
\usepackage{dirtytalk}
\usepackage{mathtools}
\newcommand{\lch}{\lambda{}_{\text{ch}}\xspace}

\newcommand{\lact}{\lambda{}_{\text{act}}\xspace}
\newenvironment{syntax}
{\small\[\begin{array}{@{}lr@{\hspace{2px}}c@{\hspace{2px}}l}\ignorespaces}
{\end{array}\]\ignorespacesafterend}

\newenvironment{fake}[1]{\par\vspace{3pt}\noindent\textbf{#1}\itshape}{\normalfont\ignorespacesafterend\vspace{3pt}\par}

\usepackage{float}
\floatstyle{boxed}
\restylefloat{figure}
\usepackage{verbatim}

\newcommand{\jarg}[1]{\{ #1 \} \; \, }

\newcommand{\defeq}{\triangleq}
\newcommand\doubleplus{+\kern-1.3ex+\kern0.8ex}
\newcommand{\coalesceinner}[1]{\underline{#1}}
\newcommand{\coalesce}[2]{\abs{#1}^{#2}}
\newcommand{\cct}[1]{\mkwd{cct}(#1)}
\newcommand{\compat}{\asymp}
\newcommand{\app}{\,}

\newcommand{\termcol}{Blue}
\newcommand{\mkwd}[1]{\ensuremath{\mathbf{\mathsf{#1}}}}
\newcommand{\termconst}[1]{{\color{\termcol}{\mathsf{#1}}}}
\newcommand{\chan}[1]{\textsf{ChanRef}(#1)}
\newcommand{\one}{\mathbf{1}}

\newcommand{\letunit}[2]{#1 ; #2}
\newcommand{\const}[1]{\mkwd{#1}}
\newcommand{\rec}[3]{\termconst{rec}\app #1(#2) \app . \app #3}
\newcommand{\rectwo}[2]{\termconst{rec}\app #1(#2) \app . }
\newcommand{\inl}[1]{\termconst{inl} \app #1}
\newcommand{\inr}[1]{\termconst{inr} \app #1}
\newcommand{\caseof}[2]{\termconst{case}\:#1\:\{#2\}}
\newcommand{\caseofone}[1]{\termconst{case} \:#1 \: \{ }
\newcommand{\emptylist}{\texttt{[~]}}
\newcommand{\listty}[1]{\mkwd{List}(#1)}
\newcommand{\listterm}[1]{\textbf{[} #1 \textbf{]}}
\newcommand{\listcons}[2]{#1 :: #2}
\newcommand{\letin}[3]{\termconst{let} \app #1 = #2 \app \termconst{in} \app #3 }
\newcommand{\letintwo}[2]{\termconst{let} \app #1 = #2 \app \termconst{in} }
\newcommand{\listappend}[2]{#1 \doubleplus #2}
\newcommand{\variant}[1]{\langle #1 \rangle}

\newcommand{\epush}[1]{\termconst{Push}(#1)}
\newcommand{\epop}[1]{\termconst{Pop}(#1)}
\newcommand{\esome}[1]{\termconst{Some}(#1)}
\newcommand{\enone}{\termconst{None}}
\newcommand{\optty}[1]{\termconst{Option}(\textsf{#1})}
\newcommand{\operationty}[1]{\termconst{Operation}(\textsf{#1})}
\newcommand{\efflet}[3]{\termconst{let} \; #1 \Leftarrow #2 \; \termconst{in} \; #3}
\newcommand{\efflettwo}[2]{\termconst{let} \; #1 \Leftarrow #2 \; \termconst{in}}
\newcommand{\efflettwonoin}[2]{\termconst{let} \; #1 \Leftarrow #2}

\newcommand{\effreturn}[1]{\termconst{return} \app #1}

\newcommand{\lchcol}{violet}
\newcommand{\lchwd}[1]{{\color{\lchcol}{\mkwd{#1}}}}
\newcommand{\lchbuf}[2]{#1(#2)}
\newcommand{\lchfork}[1]{{\color{\lchcol}{\const{fork}}} \app #1}
\newcommand{\lchgive}[2]{{\color{\lchcol}{\const{give}}} \app #1 \app #2}
\newcommand{\lchtake}[1]{{\color{\lchcol}{\const{take}}} \app #1}
\newcommand{\lchnewch}{{\color{\lchcol}{\const{newCh}}}}
\newcommand{\lchchoose}[2]{\lchwd{choose} \app #1 \app #2}
\newcommand{\chanzero}{\mkwd{Chan}}

\newcommand{\lactcol}{RedOrange}
\newcommand{\lactwd}[1]{{\color{\lactcol}{\mkwd{#1}}}}
\newcommand{\lactto}[1]{\to^{#1}}
\newcommand{\lacttotwo}[2]{\to^{#1, #2}}
\newcommand{\lactsend}[2]{{\color{\lactcol}{\mkwd{send}}} \app #1 \app #2}
\newcommand{\lactrecv}{{\color{\lactcol}{\mkwd{receive}}}}
\newcommand{\lactspawn}[1]{{\color{\lactcol}{\mkwd{spawn}}} \app #1}
\newcommand{\lactself}{{\color{\lactcol}{\mkwd{self}}}}
\newcommand{\pid}[1]{\mkwd{ActorRef}(#1)}
\newcommand{\actor}[3]{\langle #1, #2, #3 \rangle}

\newcommand{\pidtwo}[2]{\mkwd{ActorRef}(#1, #2)}
\newcommand{\lactwait}[1]{\lactwd{wait} \app #1}

\newcommand{\config}[1]{\mathcal{#1}}
\newcommand{\Ex}{E}
\newcommand{\Cx}{G}

\newcommand{\teval}{\longrightarrow_{\textsf{M}}}
\newcommand{\ceval}{\longrightarrow}

\newcommand{\translateactty}[1]{\llbracket \app #1 \app \rrbracket}
\newcommand{\translateactval}[1]{\llbracket \app #1 \app \rrbracket}
\newcommand{\translateactterm}[2]{\llbracket \app #1 \app
\rrbracket \; #2 }
\newcommand{\translateacttermzero}[1]{\llbracket \app #1 \app
\rrbracket }
\newcommand{\translateactconfig}[1]{\llbracket \app #1 \app \rrbracket}

\newcommand{\translatechty}[1]{\llparenthesis \app #1 \app \rrparenthesis}
\newcommand{\translatechtysync}[1]{\llparenthesis \app #1 \app \rrparenthesis}
\newcommand{\translatechval}[1]{\llparenthesis \app #1 \app \rrparenthesis}
\newcommand{\translatechterm}[1]{\llparenthesis \app #1 \app \rrparenthesis }
\newcommand{\translatechconfig}[1]{\llparenthesis \app #1 \app \rrparenthesis}

\newcommand{\metadef}[1]{\textbf{\texttt{#1}}}

\newcommand{\seq}[1]{\overrightarrow{#1}}
\newcommand{\fv}{\mathsf{fv}}

\newenvironment{proofcase}[1]{\noindent\textbf{Case #1} \\
  \begin{longtable}[l]{lll}
  }
  {\end{longtable}}

\mprset {sep=1em}
\def \TirNameStyle #1{\scriptsize \textsc{#1}}
\newenvironment{qedproof}{\begin{proof}}{\end{proof}}
\usepackage{tensor}

\newcommand{\selrecv}[1]{\lactwd{receive} \; \{ #1 \}}

\newcommand{\when}{\; \termconst{when} \; }

\newcommand{\matches}{\textsf{matches}}
\newcommand{\matchesany}{\textsf{matchesAny}}
\newcommand{\boolty}{\textsf{Bool}\xspace}
\newcommand{\ttrue}{\textsf{true}\xspace}
\newcommand{\ffalse}{\textsf{false}\xspace}

\DeclarePairedDelimiter\floor{\lfloor}{\rfloor}

\newcommand{\seltrans}[1]{\floor{#1}}
\newcommand{\seltransterm}[2]{\floor{#1} \app #2}

\newcommand{\find}[2]{\metadef{find}(#1, #2)}
\newcommand{\transloop}[2]{\metadef{loop}(#1, #2)}

\newcommand{\roll}[1]{\termconst{roll} \: #1}
\newcommand{\unroll}[1]{\termconst{unroll} \: #1}

\newcommand{\of}{\, {:} \,\xspace}
\newcommand{\nodups}[1]{\metadef{noDups}(#1)}
 
\usepackage{tikz}
\usetikzlibrary{matrix,fit}

\title{Mixing Metaphors: Actors as Channels and Channels as Actors (Extended Version)}
\titlerunning{Mixing Metaphors}
\author{Simon Fowler}
\author{Sam Lindley}
\author{Philip Wadler}
\affil{University of Edinburgh, Edinburgh, United Kingdom \\
\texttt{simon.fowler@ed.ac.uk, sam.lindley@ed.ac.uk, wadler@inf.ed.ac.uk}}
\authorrunning{Fowler, Lindley, and Wadler} 
\Copyright{S.\ Fowler, S.\ Lindley, and P.\ Wadler}
\subjclass{D.1.3 Concurrent Programming}
\keywords{Actors, Channels, Communication-centric Programming Languages}

\EventEditors{Peter M\"uller}
\EventNoEds{1}
\EventLongTitle{31st European Conference on Object-Oriented Programming
(ECOOP 2017)}
\EventShortTitle{ECOOP 2017}
\EventAcronym{ECOOP}
\EventYear{2017}
\EventDate{June 18--23, 2017}
\EventLocation{Barcelona, Spain}
\EventLogo{}
\SeriesVolume{74}
\ArticleNo{90}

\begin{document}
\maketitle              
\begin{abstract}
Channel- and actor-based programming languages are both used in
practice, but the two are often confused. Languages such as Go
provide anonymous processes which communicate using buffers or rendezvous
points---known as channels---while languages such as Erlang provide
addressable processes---known as actors---each with a single incoming message
queue.
The lack of a common representation makes it difficult to reason about
translations that exist in the folklore. We define a calculus
$\lch$ for typed asynchronous channels, and a calculus $\lact$ for
typed actors. We define translations from $\lact$ into $\lch$ and
$\lch$ into $\lact$ and prove that both are type- and
semantics-preserving.
We show that our approach accounts for synchronisation and selective
receive in actor systems and discuss future extensions to support guarded
choice and behavioural types.
\end{abstract}
\section{Introduction}
When comparing channels (as used by Go) and actors (as used by Erlang), one
runs into an immediate mixing of metaphors. The words themselves do not refer
to comparable entities!

In languages such as Go, anonymous processes pass messages via named channels,
whereas in languages such as Erlang, named processes accept messages from an
associated mailbox. A channel is either a named rendezvous point or buffer,
whereas an actor is a process. We should really be comparing named
processes (actors) with anonymous processes, and buffers tied to a
particular process (mailboxes) with buffers that can link any process to
any process (channels). Nonetheless, we will stick with the popular names,
even if it is as inapposite as comparing TV channels with TV actors.

\begin{figure}[b]
  \centering
  \begin{subfigure}[t]{0.4\textwidth}
    \centering
    \includegraphics[width=0.7\textwidth]{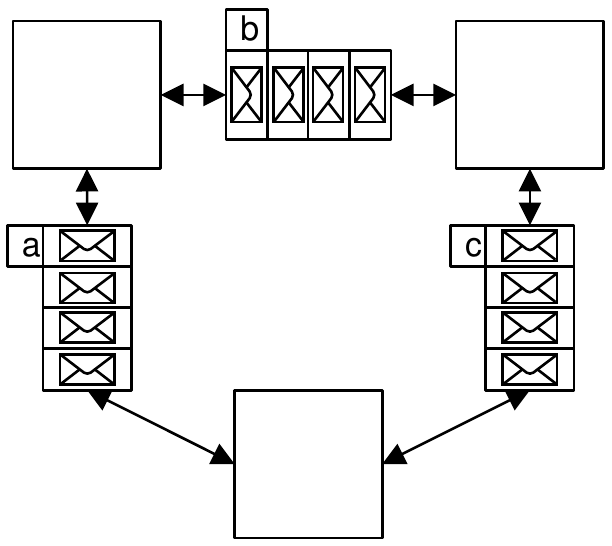}
    \caption{Asynchronous Channels}
    \label{fig:intro-chans}
  \end{subfigure}
  ~
  \begin{subfigure}[t]{0.4\textwidth}
    \centering
    \includegraphics[width=0.75\textwidth]{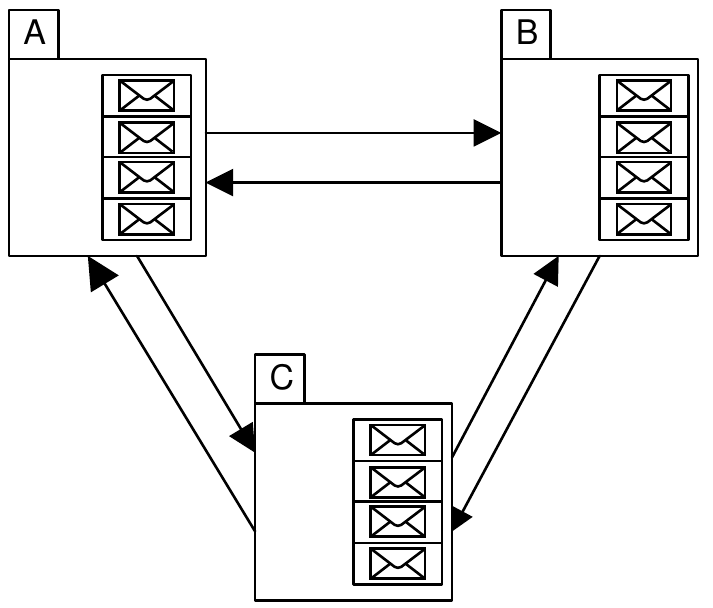}
    \caption{Actors}
    \label{fig:intro-actors}
  \end{subfigure}
  \qquad
\caption{Channels and Actors}\label{fig:chan-actor-diag}
\end{figure}

\noindent
Figure~\ref{fig:chan-actor-diag} compares asynchronous channels with
actors. On the left, three anonymous processes communicate via channels
named $a, b, c$. On the right, three processes named $A,
B, C$ send messages to each others' associated mailboxes.
Actors are necessarily asynchronous, allowing non-blocking sends and
buffering of received values, whereas channels can either be asynchronous
or synchronous (rendezvous-based). Indeed, Go provides both synchronous
\emph{and} asynchronous channels, and libraries such as
\texttt{core.async}~\cite{core-async} provide library support for asynchronous
channels.
However, this is not the only
difference: each actor has a single buffer which only it can read---its
\emph{mailbox}---whereas asynchronous channels are free-floating buffers that
can be read by any process with a reference to the channel.

Channel-based languages such as Go enjoy a firm basis in process calculi such
as CSP~\cite{hoare:csp} and the $\pi$-calculus~\cite{milner:picalc}. It is easy
to type channels, either with simple types (see~\cite{sangiorgi:picalc-book},
p.\ 231) or more
complex systems such as session
types~\cite{honda:dyadic, honda:primitives, gay:linearsessions}.
Actor-based languages such as Erlang are seen by many as the "gold standard"
for distributed computing due to their support for fault tolerance
through supervision hierarchies~\cite{armstrong:thesis, cesarini:scalable-erlang}.

Both models are popular with developers, with channel-based languages and
frameworks such as Go, Concurrent ML~\cite{reppy:cml}, and Hopac~\cite{hopac};
and actor-based languages and frameworks such as Erlang, Elixir, and Akka.

\subsection{Motivation}
This paper provides a formal account of actors and channels as implemented in
programming languages.
Our motivation for a formal account is threefold:
it helps clear up confusion; it clarifies results
that have been described informally by putting practice into theory; and it
provides a foundation for future research.

\subparagraph{Confusion.}
There is often confusion over the differences between channels and
actors. For example, the following questions appear on StackOverflow
and Quora respectively:
\begin{quote}
``If I wanted to port a Go library that uses Goroutines, would Scala be a
  good choice because its inbox/[A]kka framework is similar in nature to
	coroutines?''~\cite{scala-actors-similar-go}, and
\end{quote}
\begin{quote}
``I don't know anything about [the] actor pattern however I do know
	goroutines and channels in Go. How are [the] two related to each
	other?''~\cite{how-akka-different}
\end{quote}
In academic circles, the term \emph{actor} is often used
imprecisely.
For instance, Albert et al.~\cite{albert:clp-testing} refer to Go as
an actor language.
Similarly, Harvey~\cite{harvey:thesis} refers to his language Ensemble
as actor-based. Ensemble is a language specialised for writing
distributed applications running on heterogeneous platforms. It is
actor-based to the extent that it has lightweight, addressable,
single-threaded processes, and forbids co-ordination via shared
memory. However, Ensemble communicates using channels as opposed to
mailboxes so we would argue that it is channel-based (with actor-like
features) rather than actor-based.

\subparagraph{Putting practice into theory.}
The success of actor-based languages is largely due to their support for
\emph{supervision}.
A popular pattern for writing
actor-based applications is to arrange processes in \emph{supervision
hierarchies}~\cite{armstrong:thesis}, where \emph{supervisor} processes restart
child processes should they fail.
Projects such as Proto.Actor~\cite{proto-actor} emulate
actor-style programming in a channel-based language in an attempt to
gain some of the benefits, by associating queues with processes.
Hopac~\cite{hopac} is a channel-based library for F\#, based on
Concurrent ML~\cite{reppy:cml}. The documentation~\cite{hopac-actors} contains a
comparison with actors, including an implementation of a simple actor-based
communication model using Hopac-style channels, as well as an implementation of
Hopac-style channels using an actor-based communication model.
By comparing the two, this paper provides a formal model for the underlying
techniques, and studies properties arising from the translations.

\subparagraph{A foundation for future research.}
Traditionally, actor-based languages have had untyped mailboxes. More recent
advancements such as TAkka~\cite{he:typecasting-actors}, Akka
Typed~\cite{akka-typed}, and Typed Actors~\cite{typed-actors} have added types
to mailboxes in order to gain additional safety guarantees. Our formal model
provides a foundation for these innovations, characterises why na\"ively adding
types to mailboxes is problematic, and provides a core language for future
experimentation.

\subsection{Our approach}
We define two concurrent \emph{$\lambda$-calculi}, describing
\emph{asynchronous} channels and type-parameterised actors, define translations
between them, and then discuss
various extensions.

\subparagraph{Why the $\lambda$ calculus?}
Our common framework is that of a simply-typed concurrent $\lambda$-calculus: that
is, a $\lambda$-calculus equipping a term language with
primitives for communication and concurrency, as well as a language of
\emph{configurations} to model concurrent behaviour.
We work with the $\lambda$-calculus rather than a process calculus for two
reasons: firstly, the simply-typed $\lambda$-calculus has a well-behaved core
with a strong metatheory (for example, confluent reduction and strong
normalisation), as well as a direct propositions-as-types correspondence with
logic. We can therefore modularly extend the language, knowing which properties
remain; typed process calculi typically do not have such a well-behaved core.

Secondly, we are ultimately interested in functional programming
languages; the $\lambda$ calculus is the canonical choice for studying such
extensions.

\subparagraph{Why asynchronous channels?}
While actor-based languages must be asynchronous by design, channels may be
either synchronous (requiring a rendezvous between sender and receiver) or
asynchronous (where sending happens immediately).  In this paper, we
consider asynchronous channels since actors must be asynchronous,
and it is possible to emulate asynchronous channels using synchronous
channels~\cite{reppy:cml}.
We could adopt synchronous channels, use these to encode asynchronous
channels, and then do the translations. We elect not to since it complicates
the translations, and we argue that the distinction between synchronous and
asynchronous communication is not \emph{the} defining difference between the
two models.

\begin{figure}
\centering
\begin{subfigure}[t]{0.4\textwidth}
  \centering
  \begin{tikzpicture}
    \matrix (m) [matrix of math nodes,row sep=2em,column sep=7em,minimum width=1em]
    {
       P_1 & P_1 \\
       P_2 & P_2 \\
       P_3 & P_3 \\
    };
    \path
      (m-1-1) edge (m-1-2)
              edge (m-2-2)
              edge (m-3-2)
      (m-2-1) edge (m-1-2)
              edge (m-2-2)
              edge (m-3-2)
      (m-3-1) edge (m-1-2)
              edge (m-2-2)
              edge (m-3-2);

    \node (sender) [fit = (m-1-1) (m-3-1)] {};
    \node at (sender.north) [above] {\textsf{sender}};

    \node (receiver) [fit = (m-1-2) (m-3-2)] {};
    \node at (receiver.north) [above] {\textsf{receiver}};

  \end{tikzpicture}
  \caption{Channel}
  \label{fig:channel-endpoints}
\end{subfigure}
~
\begin{subfigure}[t]{0.4\textwidth}
  \centering
  \begin{tikzpicture}
    \matrix (m) [matrix of math nodes,row sep=2em,column sep=7em,minimum width=1em]
    {
       P_1 & P_1 \\
       P_2 & P_2 \\
       P_3 & P_3 \\
    };
    \path
      (m-1-1) edge (m-2-2)
      (m-2-1) edge (m-2-2)
      (m-3-1) edge (m-2-2);

    \node (sender) [fit = (m-1-1) (m-3-1)] {};
    \node at (sender.north) [above] {\textsf{sender}};

    \node (receiver) [fit = (m-1-2) (m-3-2)] {};
    \node at (receiver.north) [above] {\textsf{receiver}};

  \end{tikzpicture}
  \caption{Mailbox}
  \label{fig:mailbox-endpoints}
\end{subfigure}

\caption{Mailboxes as pinned channels}
\label{fig:mailboxes-as-pinned-channels}
\end{figure}
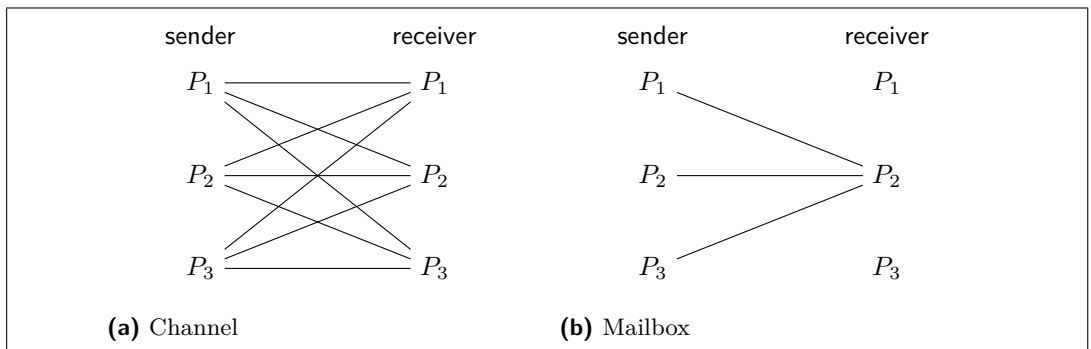

\subsection{Summary of results}

We identify four key differences between the models, which are exemplified by
the formalisms and the translations: process addressability, the restrictiveness
of communication patterns, the granularity of typing, and the ability to
control the order in which messages are processed.

\subparagraph{Process addressability.}
In channel-based systems, processes are \emph{anonymous}, whereas channels are
named. In contrast, in actor-based systems, processes are named.

\subparagraph{Restrictiveness of communication patterns.}
Communication over full-duplex channels is more liberal than communication via
mailboxes, as shown in Figure~\ref{fig:mailboxes-as-pinned-channels}.
Figure~\ref{fig:channel-endpoints} shows the communication patterns
allowed by a single channel: each process $P_i$ can use
the channel to communicate with every other process. Conversely,
Figure~\ref{fig:mailbox-endpoints} shows the communication patterns
allowed by a mailbox associated with process $P_2$: while any process
can send to the mailbox, only $P_2$ can read from it.
Viewed this way, it is apparent that the restrictions imposed on the
communication behaviour of actors are exactly those captured by Merro and
Sangiorgi's localised $\pi$-calculus~\cite{MerroS04}.

Readers familiar with actor-based programming may be wondering whether such a
characterisation is too crude, as it does not account for processing messages
out-of-order. Fear not---we show in \textsection{}\ref{sec:extensions} that our
minimal actor calculus can simulate this functionality.

Restrictiveness of communication patterns is not necessarily a bad thing; while
it is easy to distribute actors, \emph{delegation} of asynchronous channels is
more involved, requiring a distributed
algorithm~\cite{hu:sessionjava}. Associating mailboxes with addressable
processes also helps with structuring applications for
reliability~\cite{cesarini:scalable-erlang}.

\subparagraph{Granularity of typing.}
As a result of the fact that each process has a single incoming message queue,
mailbox types tend to be less precise; in particular, they are most commonly
variant types detailing all of the messages that can be received.
Na\"ively implemented, this gives rise to the \emph{type pollution problem},
which we describe further in \textsection{}\ref{sec:example:type-pollution}.

\subparagraph{Message ordering.}
Channels and mailboxes are ordered message queues, but there is no
inherent ordering between messages on two different channels. Channel-based
languages allow a user to specify from which channel a message should be received,
whereas processing messages out-of-order can be achieved in actor languages
using selective receive.

The remainder of the paper captures these differences both in the design of the
formalisms, and the techniques used in the encodings and extensions.

\subsection{Contributions and paper outline}

This paper makes five main contributions:

\begin{enumerate}
  \item A calculus $\lch$ with typed asynchronous channels
    (\textsection{}\ref{sec:lch}), and a calculus $\lact$ with
    type-parameterised actors (\textsection{}\ref{sec:lact}), based on the
    $\lambda$-calculus extended with communication primitives specialised to
    each model. We give a type system and operational semantics for each
    calculus, and precisely characterise the notion of progress that each
    calculus enjoys.
  \item A simple translation from $\lact$ into $\lch$
    (\textsection{}\ref{sec:lact-lch}), and a more involved
    translation from $\lch$ into $\lact$ (\textsection{}\ref{sec:lch-lact}),
    with proofs that both translations are type- and semantics-preserving.
    While the former translation is straightforward, it is
    \emph{global}, in the sense of Felleisen~\cite{felleisen:expressiveness}.
    While the latter is more involved, it is in fact \emph{local}. Our
    initial translation from $\lch$ to $\lact$ sidesteps type pollution by
    assigning the same type to each channel in the system.
  \item An extension of $\lact$ to support synchronous calls, showing how this
    can alleviate type pollution and simplify the translation from
    $\lch$ into $\lact$ (\textsection{}\ref{sec:extensions-sync}).
  \item An extension of $\lact$ to support Erlang-style selective receive, a
    translation from $\lact$ with selective receive into plain $\lact$, and
    proofs that the translation is type- and semantics-preserving
    (\textsection{}\ref{sec:extensions-selrecv}).
  \item An extension of $\lch$ with input-guarded choice
    (\textsection{}\ref{sec:extensions-choice}) and an outline of how
    $\lact$ might be extended with behavioural types
    (\textsection{}\ref{sec:extensions-behavioural-types}).
\end{enumerate}

\noindent
The rest of the paper is organised as follows:
\textsection{}\ref{sec:example} displays side-by-side two
implementations of a concurrent stack, one using channels and the
other using actors; \textsection{}\ref{sec:lch}--\ref{sec:extensions}
presents the main technical content;
\textsection{}\ref{sec:related-work} discusses related work; and
\textsection{}\ref{sec:conclusion} concludes.

 \section{Channels and actors side-by-side}
\label{sec:example}
\begin{figure}[t]
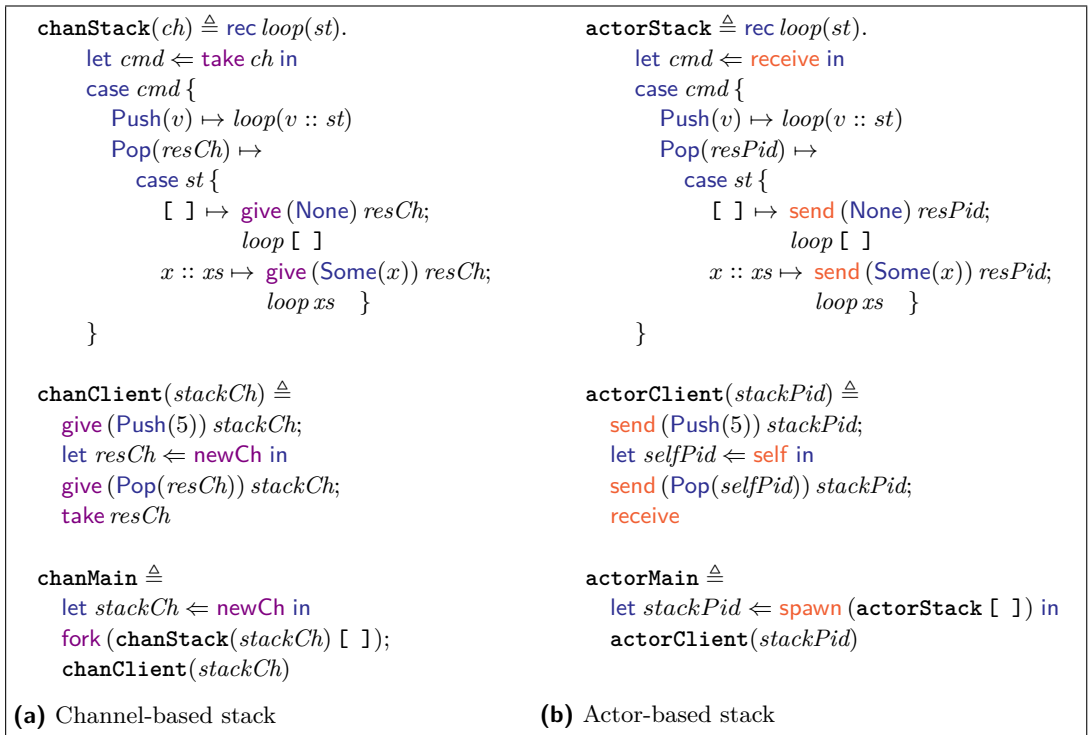

  \vspace{-1em}
  \begin{subfigure}[t]{0.47\textwidth}
    \small
    \begin{mathpar}
      \begin{array}{l}
        \metadef{chanStack}(\textit{ch}) \defeq \termconst{rec} \app \textit{loop}(\textit{st}) . \\
        \qquad \efflettwo{\textit{cmd}}{\lchtake{\textit{ch}}} \\
        \qquad \caseofone{\textit{cmd}} \\
        \qquad \quad \epush{v} \mapsto \textit{loop}(\listcons{v}{\textit{st}}) \\
        \qquad \quad \epop{\textit{resCh}} \mapsto \\
        \qquad \quad \quad \caseofone{\textit{st}} \\
        \qquad \quad \qquad \emptylist{} \mapsto
          \begin{aligned}[t]
            & \lchgive{(\enone{})}{\textit{resCh}}; \\
            & \textit{loop} \app \emptylist{}
          \end{aligned} \\
        \qquad \quad \qquad \listcons{x}{\textit{xs}} \mapsto
          \begin{aligned}[t]
            & \lchgive{(\esome{x})}{\textit{resCh}}; \\
            & \textit{loop} \app \textit{xs} \quad \}
          \end{aligned} \\
        \qquad  \} \\
           \\
        \metadef{chanClient}(\textit{stackCh}) \defeq \\
        \quad \lchgive{(\epush{5})}{\textit{stackCh}}; \\
        \quad \efflettwo{\textit{resCh}}{\lchnewch{}} \\
        \quad \lchgive{(\epop{\textit{resCh}})}{\textit{stackCh}}; \\
        \quad \lchtake{\textit{resCh}} \\
        \\
        \metadef{chanMain} \defeq \\
        \quad \efflettwo{\textit{stackCh}}{\lchnewch{}} \\
        \quad \lchfork{(\metadef{chanStack}(\textit{stackCh}) \app \emptylist{} ) }; \\
        \quad \metadef{chanClient}(\textit{stackCh})
      \end{array}
    \end{mathpar}
    \vspace{-1em}
    \caption{Channel-based stack}
    \label{fig:chan-stack-example}
  \end{subfigure}
  ~
  \begin{subfigure}[t]{0.5\textwidth}
    \small
    \[
      \begin{array}{l}
        \metadef{actorStack} \defeq \termconst{rec} \app \textit{loop}(st) . \\
        \qquad \efflettwo{\textit{cmd}}{\lactrecv} \\
        \qquad \caseofone{\textit{cmd}} \\
        \qquad \quad \epush{v} \mapsto \textit{loop}(\listcons{v}{st}) \\
        \qquad \quad \epop{\textit{\textit{resPid}}} \mapsto \\
        \qquad \quad \quad \caseofone{\textit{st}} \\
        \qquad \quad \qquad \emptylist{}  \mapsto
          \begin{aligned}[t]
            & \lactsend{(\enone{})}{\textit{resPid}}; \\
            & \textit{loop} \app \emptylist{}
          \end{aligned} \\
        \qquad \quad \qquad \listcons{x}{\textit{xs}} \mapsto
          \begin{aligned}[t]
            & \lactsend{(\esome{x})}{\textit{resPid}}; \\
          & \textit{loop} \app \textit{xs} \quad \}
        \end{aligned} \\
      \qquad \} \\
        \\
        \metadef{actorClient}(\textit{stackPid}) \defeq \\
        \quad \lactsend{(\epush{5})}{\textit{stackPid}}; \\
        \quad \efflettwo{\textit{selfPid}}{\lactself{}} \\
        \quad \lactsend{(\epop{\textit{selfPid}})}{\textit{stackPid}}; \\
        \quad \lactrecv{} \\
        \\
        \metadef{actorMain} \defeq \\
        \quad \efflettwo{stackPid}{\lactspawn{(\metadef{actorStack} \app \emptylist{}) } } \\
        \quad \metadef{actorClient}(\textit{stackPid})
      \end{array}
    \]
    \vspace{0.2em}
    \caption{Actor-based stack}
    \label{fig:actor-stack-example}
  \end{subfigure}
  \caption{Concurrent stacks using channels and actors}
  \label{fig:conc-stack-example}
\end{figure}

Let us consider the example of a concurrent stack. A concurrent stack carrying
values of type $A$ can receive a command to push a value onto the top of the
stack, or to pop a value and return it to the process making the
request.
Assuming a standard encoding of algebraic datatypes, we define
a type
$\operationty{A} = \termconst{Push}(A) \mid \termconst{Pop}(B)$ (where $B =
\chan{A}$ for channels, and $\pid{A}$ for actors) to describe
operations on the stack, and
$\optty{A} = \termconst{Some}(A) \mid \termconst{None}$ to handle
the possibility of popping from an empty stack.

Figure~\ref{fig:conc-stack-example} shows the stack implemented using
channels (Figure~\ref{fig:chan-stack-example}) and using actors
(Figure~\ref{fig:actor-stack-example}). Each implementation uses a
common core language based on the simply-typed $\lambda$-calculus
extended with recursion, lists, and sums.

At first glance, the two stack implementations seem remarkably similar. Each:
\begin{enumerate}
  \item Waits for a command
  \item Case splits on the command, and either:
    \begin{compactitem}
      \item Pushes a value onto the top of the stack, or;
      \item Takes the value from the head of the stack and returns it in a
        response message
      \end{compactitem}
    \item Loops with an updated state.
\end{enumerate}
The main difference is that \metadef{chanStack} is
parameterised over a channel \textit{ch}, and retrieves a value from the channel using
$\lchtake{\textit{ch}}$. Conversely, \metadef{actorStack} retrieves a
value from its mailbox using the nullary primitive $\lactrecv{}$.

Let us now consider functions which interact with the stacks.  The
\metadef{chanClient} function sends commands over the
\textit{stackCh} channel, and begins by pushing $5$ onto the stack.
Next, it creates a channel \textit{resCh} to be used to
receive the result and sends this in a request, before retrieving the
result from the result channel using $\lchwd{take}$.
In contrast, \metadef{actorClient} performs a similar set of steps,
but sends its process ID (retrieved using $\lactself{}$) in the request
instead of creating a new channel; the result is then retrieved from the
mailbox using $\lactrecv{}$.

\begin{figure}[t]
  \vspace{-1em}
  \begin{subfigure}[t]{0.48\textwidth}
    \small
    \[
    \begin{array}{l}
      \metadef{chanClient2}(\textit{intStackCh}, \textit{stringStackCh}) \defeq \\
      \quad \efflettwo{\textit{intResCh}}{\lchnewch{}} \\
      \quad \efflettwo{\textit{strResCh}}{\lchnewch{}} \\
      \quad \lchgive{(\epop{\textit{intResCh}})}{\textit{intStackCh}}; \\
      \quad \efflettwo{\textit{res1}}{\lchtake{\textit{intResCh}}} \\
      \quad \lchgive{(\epop{\textit{strResCh}})}{\textit{stringStackCh}}; \\
      \quad \efflettwo{\textit{res2}}{\lchtake{\textit{strResCh}}} \\
      \quad (\textit{res1}, \textit{res2})
    \end{array}
    \]
  \end{subfigure}
  ~
  \begin{subfigure}[t]{0.48\textwidth}
    \small
    \[
    \begin{array}{l}
      \metadef{actorClient2}(\textit{intStackPid}, \textit{stringStackPid}) \defeq \\
      \quad \efflettwo{\textit{selfPid}}{\lactself{}} \\
      \quad \lactsend{(\epop{\textit{selfPid}})}{\textit{intStackPid}}; \\
      \quad \efflettwo{\textit{res1}}{\lactrecv{}} \\
      \quad \lactsend{(\epop{\textit{selfPid}})}{\textit{stringStackPid}}; \\
      \quad \efflettwo{\textit{res2}}{\lactrecv{}} \\
      \quad (\textit{res1}, \textit{res2})
    \end{array}
    \]  \end{subfigure}
    \caption{Clients interacting with multiple stacks}
    \label{fig:client-example-2}
\end{figure}
\subparagraph{Type pollution.}\label{sec:example:type-pollution}
The differences become more prominent when considering clients which interact
with multiple stacks of different types, as shown in
Figure~\ref{fig:client-example-2}. Here, \metadef{chanClient2} creates new result channels
for integers and strings, sends requests for the results, and creates a pair of
type $(\optty{Int} \times \optty{String})$.
The \metadef{actorClient2} function attempts to do something similar, but
cannot create separate result channels. Consequently, the actor must be able to
handle messages either of type $\optty{Int}$ \emph{or} type $\optty{String}$,
meaning that the final pair has type
$(\optty{Int} + \optty{String}) \times (\optty{Int} + \optty{String})$.

Additionally, it is necessary to modify \metadef{actorStack} to use
the correct injection into the actor type when sending the result; for
example an integer stack would have to send a value
$\inl{(\esome{5})}$ instead of simply $\esome{5}$.  This \emph{type
  pollution} problem can be addressed through the use of
subtyping~\cite{he:typecasting-actors}, or synchronisation
abstractions such as futures~\cite{deboer:complete-future}.

 \section{$\lch$: A concurrent $\lambda$-calculus for channels}\label{sec:lch}

In this section we introduce $\lch$, a concurrent $\lambda$-calculus
extended with asynchronous channels. To concentrate on the core differences
between channel- and actor-style communication, we begin with minimal calculi;
note that these do not contain all features (such as lists, sums, and recursion)
needed to express the examples in \textsection{}\ref{sec:example}.

\subsection{Syntax and typing of terms}

\begin{figure}[t]
~Syntax
\begin{syntax}
  \text{Types} & A,B & ::= &
    \one \mid A \to B \mid \chan{A} \\
  \text{Variables and names} & \alpha & ::= & x \mid a \\
  \text{Values} & V, W & ::= & \alpha \mid \lambda x . M \mid () \\
  \text{Computations}  & L,M,N & ::=  & V\,W \\
                &       & \mid & \efflet{x}{M}{N} \mid \effreturn{V} \\
    & & \mid & \lchfork{M} \mid \lchgive{V}{W} \mid \lchtake{V} \mid \lchnewch{} \\
\end{syntax}
~Value typing rules
\hfill \framebox{$\Gamma \vdash V : A$} \\
{\footnotesize
\begin{mathpar}
\inferrule
[Var]
  { \alpha : A \in \Gamma}
  { \Gamma \vdash \alpha : A}

\inferrule
[Abs]
  {\Gamma,x:A \vdash M:B}
  {\Gamma \vdash \lambda x.M : A \to B}

\inferrule
[Unit]
  { }
  {\Gamma \vdash (): \one}
\end{mathpar}
}

~Computation typing rules
\hfill \framebox{$\Gamma \vdash M : A$} \\
{\footnotesize
\begin{mathpar}
\inferrule
[App]
  {\Gamma \vdash V : A \to B \\
   \Gamma \vdash W: A}
  {\Gamma \vdash V\,W: B}

\inferrule
[EffLet]
  { \Gamma \vdash M : A \\ \Gamma, x : A \vdash N : B }
  { \Gamma \vdash \efflet{x}{M}{N} : B }

\inferrule
[Return]
  { \Gamma \vdash V : A }
  { \Gamma \vdash \effreturn{V} : A }

\inferrule
[Give]
  { \Gamma \vdash V : A \\\\ \Gamma \vdash W : \chan{A} }
  { \Gamma \vdash \lchgive{V}{W} : \one }

\inferrule
[Take]
  {\Gamma \vdash V : \chan{A}}
  {\Gamma \vdash \lchtake{V} : A }

  \inferrule
[Fork]
  {\Gamma \vdash M : \one }
  {\Gamma \vdash \lchfork{M} : \one}

\inferrule
[NewCh]
  { }
  { \Gamma \vdash \lchnewch{} : \chan{A} }
\end{mathpar}
}
\caption{Syntax and typing rules for $\lch$ terms and values}
\label{fig:lch-syntax-typing}
\end{figure}

Figure~\ref{fig:lch-syntax-typing} gives the syntax and typing rules
of $\lch$, a $\lambda$-calculus based on fine-grain
call-by-value~\cite{levy:fine-grain-cbv}: terms are partitioned into
values and computations.
Key to this formulation are two constructs:
$\effreturn{V}$ represents a computation that has completed, whereas
$\efflet{x}{M}{N}$ evaluates $M$ to $\effreturn{V}$, substituting $V$ for $x$
in $M$.
Fine-grain call-by-value is convenient since it makes evaluation order
explicit and, unlike A-normal form~\cite{FlanaganSDF93}, is closed
under reduction.

Types consist of the unit type $\one$, function types $A \to B$, and
channel reference types $\chan{A}$ which can be used to communicate
along a channel of type $A$.
We let $\alpha$ range over variables $x$ and runtime names $a$.
We write $\letin{x}{V}{M}$ for $(\lambda x . M) \app V$ and
$\letunit{M}{N}$ for $\efflet{x}{M}{N}$, where $x$ is fresh.

\subparagraph{Communication and concurrency for channels.}
The $\lchgive{V}{W}$ operation sends value $V$ along channel $W$,
while $\lchtake{V}$ retrieves a value from a channel $V$. Assuming an
extension of the language with integers and arithmetic operators, we
can define a function $\metadef{neg}(c)$ which receives a number $n$
along channel $c$ and replies with the negation of $n$ as follows:
\[
\metadef{neg}(c) \defeq
     \efflet{n}{\lchtake{c}}{\efflet{\textit{negN}}{(-n)}{\lchgive{\textit{negN}}{\app c}}}
\]
The $\lchfork{M}$ operation spawns a new process to evaluate term
$M$. The operation returns the unit value, and therefore it is not possible to
interact with the process directly.
The $\lchnewch{}$ operation creates a new channel. Note that channel creation is
decoupled from process creation, meaning that a process can have access to
multiple channels.

\subsection{Operational semantics}
\subparagraph{Configurations.}
The concurrent behaviour of $\lch$ is given by a nondeterministic
reduction relation on \emph{configurations}~(Figure~\ref{fig:lch-eval-syntax}).
Configurations consist of parallel composition ($\config{C} \parallel
\config{D}$), restrictions ($(\nu a) \config{C}$), computations ($M$),
and buffers ($\lchbuf{a}{\seq{V}}$, where $\seq{V} = V_1 \cdot \ldots \cdot V_n$).

\subparagraph{Evaluation contexts.}
Reduction is defined in terms of evaluation contexts $E$, which are simplified due to
fine-grain call-by-value. We also define configuration contexts, allowing reduction
modulo parallel composition and name restriction.

\begin{figure}[t]
~Syntax of evaluation contexts and configurations
\begin{syntax}
\text{Evaluation contexts} & E & ::= & [~] \mid \efflet{x}{E}{M} \\
  \text{Configurations} & \config{C}, \config{D}, \config{E} & ::= & \config{C} \parallel \config{D}
  \mid (\nu a) \config{C} \mid \lchbuf{a}{\seq{V}} \mid M \\
\text{Configuration contexts} & G & ::= & [~] \mid G \parallel \config{C} \mid (\nu a) G \\
\end{syntax}
~Typing rules for configurations
\hfill \framebox{$\Gamma ; \Delta \vdash \config{C}$} \\
{\footnotesize
\begin{mathpar}
\inferrule
  [Par]
  { \Gamma ; \Delta_1 \vdash \config{C}_1 \\ \Gamma ; \Delta_2 \vdash \config{C}_2 }
  { \Gamma ; \Delta_1, \Delta_2 \vdash \config{C}_1 \parallel \config{C}_2 }

\inferrule
  [Chan]
  { \Gamma, a : \chan{A} ; \Delta, a {:} A \vdash \config{C} }
  { \Gamma ; \Delta \vdash (\nu a) \config{C} }

\inferrule
  [Buf]
  { (\Gamma \vdash V_i : A)_i }
  { \Gamma ; a : A \vdash \lchbuf{a}{\seq{V}} }

\inferrule
  [Term]
  { \Gamma \vdash M : \one }
  { \Gamma ; \cdot \vdash M }
\end{mathpar}
}
\caption{$\lch$ configurations and evaluation contexts}
\label{fig:lch-eval-syntax}
\end{figure}

\subparagraph{Reduction.}
Figure~\ref{fig:lch-semantics} shows the reduction rules for $\lch$.
Reduction is defined as a deterministic reduction on terms ($\teval$) and a
nondeterministic reduction relation on configurations ($\ceval$).
Reduction on configurations is defined modulo structural congruence
rules which capture scope extrusion and the commutativity and associativity of
parallel composition.
\subparagraph{Typing of configurations.}
To ensure that buffers are well-scoped and contain values of the correct type,
we define typing rules on configurations (Figure~\ref{fig:lch-eval-syntax}).
The judgement $\Gamma ; \Delta \vdash
\config{C}$ states that under environments $\Gamma$ and $\Delta$, $\config{C}$
is well-typed. $\Gamma$ is a typing environment for terms, whereas $\Delta$ is a
linear typing environment for configurations, mapping names $a$ to channel types
$A$. Linearity in $\Delta$ ensures that a configuration $\config{C}$ under a
name restriction $(\nu a) \config{C}$ contains exactly one buffer with name $a$.
Note that \textsc{Chan} extends both $\Gamma$ \emph{and} $\Delta$, adding an
(unrestricted) \emph{reference} into $\Gamma$ and the \emph{capability} to type a buffer into
$\Delta$.  \textsc{Par} states that $\config{C}_1 \parallel \config{C}_2$ is
typeable if $\config{C}_1$ and $\config{C}_2$ are typeable under disjoint linear
environments, and \textsc{Buf} states that under a term environment $\Gamma$ and
a singleton linear environment $a {:} A$, it is possible to type a buffer
$\lchbuf{a}{\seq{V}}$ if $\Gamma \vdash V_i{:}A$ for all $V_i \in \seq{V}$.  As
an example, $(\nu a) (\lchbuf{a}{\seq{V}})$ is well-typed, but $(\nu a)
(\lchbuf{a}{\seq{V}} \parallel \lchbuf{a}{\seq{W}})$ and $(\nu a)
(\effreturn{()})$ are not.
\subparagraph{Relation notation.}
Given a relation $R$, we write $R^+$ for its transitive closure, and $R^*$ for
its reflexive, transitive closure.
\subparagraph{Properties of the term language.}
Reduction on terms preserves typing, and pure terms enjoy progress.
We omit most proofs in the body of the paper which are mainly straightforward
inductions; selected full proofs can be found in the extended version~\cite{mm-extended}.
\begin{lemma}[Preservation ($\lch$ terms)]\label{lem:lch-term-pres}
If $\Gamma \vdash M : A$ and $M \teval M'$, then $\Gamma \vdash M' : A$.
\end{lemma}

\begin{lemma}[Progress ($\lch$ terms)]\label{lem:lch-term-progress}
  Assume $\Gamma$ is empty or only contains channel references $a_i {:}
\chan{A_i}$.
If $\Gamma \vdash M {:} A$, then either:
\begin{enumerate}
\item $M = \effreturn{V}$ for some value $V$, or
\item $M$ can be written $E[M']$, where $M'$ is a communication or concurrency
  primitive (i.e.,\ $\lchgive{V}{W}, \lchtake{V}, \lchfork{M}$, or
  $\lchnewch{}$), or
\item There exists some $M'$ such that $M \teval M'$.
\end{enumerate}
\end{lemma}
\begin{figure}[t]
~Reduction on terms
{\small
\[
\begin{array}{ccc}
  (\lambda x . M) \app V \teval M \{ V / x \} \qquad \quad &
  \efflet{x}{\effreturn{V}}{M} \teval M \{ V / x \} \qquad \quad &
  E[M_1] \teval E[M_2] \\
  & & (\text{if } M_1 \teval M_2)
\end{array}
\]}
\vspace{-1.5em}

~Structural congruence
{\small
\begin{mathpar}
\config{C} \parallel \config{D} \equiv \config{D} \parallel C

\config{C} \parallel (\config{D} \parallel \config{E}) \equiv (\config{C} \parallel D) \parallel \config{E}

\config{C} \parallel (\nu a)\config{D} \equiv (\nu a)(\config{C} \parallel \config{D}) \text{ if $a \not\in \fv(\config{C})$}

G[\config{C}] \equiv G[\config{D}] \text{ if $\config{C} \equiv \config{D}$}
\end{mathpar}}
~Reduction on configurations
{\small
\[
\begin{array}{lrcl}
\textsc{Give} & E[\lchgive{W}{a}] \parallel \lchbuf{a}{\seq{V}} & \ceval &
  E[\effreturn{()}] \parallel \lchbuf{a}{\seq{V} \cdot W} \\
\textsc{Take} & E[\lchtake{a}] \parallel \lchbuf{a}{W \cdot \seq{V}} & \ceval &
  E[\effreturn{W}] \parallel \lchbuf{a}{\seq{V}} \\
\textsc{Fork} & E[\lchfork{M}] & \ceval & E[\effreturn{()}] \parallel M \\
\textsc{NewCh} & E[\lchnewch{}] & \ceval & (\nu a) (E[\effreturn{a}] \parallel
\lchbuf{a}{\epsilon}) \qquad (a \text{ is a fresh name}) \\
\textsc{LiftM} & G[M_1] & \ceval & G[M_2] \qquad (\text{if } M_1 \teval M_2) \\
\textsc{Lift} & G[\config{C}_1] & \ceval & G[\config{C}_2] \qquad (\text{if }
\config{C}_1 \ceval \config{C}_2) \\
\end{array}
\]}
\caption{Reduction on $\lch$ terms and configurations} \label{fig:lch-semantics}
\end{figure}
\subparagraph{Reduction on configurations.}
Concurrency and communication is captured by reduction on configurations.
Reduction is defined modulo structural congruence rules, which capture
the associativity and commutativity of parallel composition, as well as the
usual scope extrusion rule.
The \textsc{Give} rule reduces $\lchgive{W}{a}$ in parallel with a
buffer $\lchbuf{a}{\seq{V}}$ by adding the value $W$ onto the end of
the buffer.
The \textsc{Take} rule reduces $\lchtake{a}$ in parallel with a
non-empty buffer by returning the first value in the buffer.
The \textsc{Fork} rule reduces $\lchfork{M}$ by spawning a new thread
$M$ in parallel with the parent process.
The \textsc{NewCh} rule reduces $\lchnewch{}$ by creating an empty
buffer and returning a fresh name for that buffer.

Structural congruence and reduction preserve the typeability of configurations.

\begin{lemma}\label{lem:lch-equiv-typeability}
If $\Gamma ; \Delta \vdash \config{C}$ and $\config{C} \equiv \config{D}$ for some configuration $\config{D}$, then $\Gamma ; \Delta \vdash \config{D}$.
\end{lemma}

\begin{theorem}[Preservation ($\lch$ configurations)]\label{thm:lch-preservation}
If $\Gamma ; \Delta \vdash \config{C}_1$ and $\config{C}_1 \ceval \config{C}_2$ then $\Gamma ; \Delta \vdash \config{C}_2$.
\end{theorem}

\subsection{Progress and canonical forms}

While it is possible to prove deadlock-freedom in systems with more discerning
type systems based on linear logic~\cite{wadler:propositions-as-sessions,
lindley:semantics} or those using channel priorities~\cite{padovani:deadlock-free-programs},
more liberal calculi such as $\lch$ and $\lact$ allow deadlocked
configurations.
We thus define a form of progress which does not preclude deadlock;
to help with proving a progress result, it is useful to consider the
notion of a \emph{canonical form} in order to allow us to reason about the configuration
as a whole.

\begin{definition}[Canonical form ($\lch$)]
  A configuration $\config{C}$ is in \emph{canonical
  form} if it can be written
  $(\nu a_1) \ldots (\nu a_n)(M_1 \parallel
      \ldots \parallel M_m \parallel \lchbuf{a_1}{\seq{V_1}} \parallel \ldots
  \parallel \lchbuf{a_n}{\seq{V_n}})$.
\end{definition}
Well-typed open configurations can be written in a form similar to canonical
form, but without bindings for names already in the environment. An immediate
corollary is that well-typed closed configurations can always be written in a
canonical form.

\begin{lemma}\label{lem:lch-semi-canon-form}
  If $\Gamma ; \Delta \vdash \config{C}$ with $\Delta = a_1 : A_1, \ldots, a_k :
  A_k$, then there exists a $\config{C'} \equiv \config{C}$ such that
  $\config{C'} = (\nu a_{k+1}) \ldots (\nu a_n) (M_1 \parallel \ldots \parallel M_m
  \parallel \lchbuf{a_1}{\seq{V_1}} \parallel \ldots \parallel \lchbuf{a_n}{\seq{V_n}})$.
\end{lemma}

\begin{corollary}
  If $\cdot ; \cdot \vdash \config{C}$, then there exists some $\config{C'}
  \equiv \config{C}$ such that $\config{C'}$ is in canonical form.
\end{corollary}
 Armed with a canonical form, we can now state that
 the only situation in which a well-typed closed configuration
 cannot reduce further is if all threads are either blocked or fully evaluated.
 Let a \emph{leaf configuration} be a configuration without subconfigurations,
 i.e., a term or a buffer.

\begin{theorem}[Weak progress ($\lch$
  configurations)]\label{thm:lch-config-progress} \hfill \\
  Let $\cdot ; \cdot \vdash \config{C}$, $\config{C} \not\ceval$, and let
  $\config{C'} = (\nu a_1) \ldots (\nu a_n)(M_1
  \parallel \ldots \parallel M_m \parallel \lchbuf{a_1}{\seq{V_1}} \parallel
  \ldots \lchbuf{a_n}{\seq{V_n}})$ be a canonical form of $\config{C}$. Then every
  leaf of $\config{C}$ is either:

  \begin{enumerate}
    \item A buffer $\lchbuf{a_i}{\seq{V_i}}$;
    \item A fully-reduced term of the form $\effreturn{V}$, or;
    \item A term of the form $E[\lchtake{a_i}]$, where $\seq{V_i} = \epsilon$.
  \end{enumerate}
\end{theorem}

\begin{proof}
  By Lemma~\ref{lem:lch-term-progress}, we know each $M_i$ is either of the
  form $\effreturn{V}$, or can be written $E[M']$ where $M'$ is a communication
  or concurrency primitive. It cannot be the case that $M' = \lchfork{N}$ or $M'
  = \lchnewch{}$, since both can reduce.
  Let us now consider $\lchwd{give}$ and $\lchwd{take}$, blocked on a variable
  $\alpha$. As we are considering closed configurations, a blocked term must be blocked on a
  $\nu$-bound name $a_i$, and as per the canonical form, we have that there exists some buffer
  $\lchbuf{a_i}{\seq{V_i}}$. Consequently, $\lchgive{V}{a_i}$ can always reduce
  via \textsc{Give}.
  A term $\lchtake{a_i}$ can reduce by \textsc{Take} if $\seq{V_i} = W \cdot
  \seq{V_i'}$; the only remaining case is where $\seq{V_i} =
  \epsilon$, satisfying (3).
\end{proof}

\begin{figure}[t]
~Syntax
\begin{syntax}
\text{Types} & A,B,C & ::=  & \one \mid A \lactto{C} B \mid \pid{A} \\ 
\text{Variables and names} & \alpha & ::= & x \mid a \\
\text{Values} & V, W & ::= & \alpha \mid \lambda x . M \mid () \\
\text{Computations}  & L,M,N & ::=  & V\,W \\
                &       & \mid & \efflet{x}{M}{N} \mid \effreturn{V} \\
  & & \mid & \lactspawn{M} \mid \lactsend{V}{W} \mid \lactrecv \mid \lactself{} \\
\end{syntax}
\vspace{-0.5em}
~Value typing rules \hfill \framebox{$\Gamma \vdash V : A$}
{ \footnotesize
\begin{mathpar}
\inferrule
[Var]
  { \alpha : A \in \Gamma }
  { \Gamma \vdash \alpha : A}

\inferrule
[Abs]
  {\Gamma,x:A \mid C \vdash M:B}
  {\Gamma \vdash \lambda x.M : A \lactto{C} B}

\inferrule
[Unit]
  { }
  {\Gamma \vdash (): \one}
\end{mathpar}
}

~Computation typing rules \hfill \framebox{$\Gamma \mid B \vdash M : A$}
{\footnotesize
\begin{mathpar}
\inferrule
[App]
  {\Gamma \vdash V : A \lactto{C} B \\\\
   \Gamma \vdash W: A}
  {\Gamma \mid C \vdash V \app W: B}

\inferrule
[EffLet]
  { \Gamma \mid C \vdash M : A \\\\ \Gamma, x : A \mid C \vdash N : B}
  { \Gamma \mid C \vdash \efflet{x}{M}{N} : B}

\inferrule
[EffReturn]
  { \Gamma \vdash V : A }
  { \Gamma \mid C \vdash \effreturn{V} : A}

\inferrule
[Send]
  {\Gamma \vdash V : A \\\\ \Gamma \vdash W : \pid{A} }
  {\Gamma \mid C \vdash \lactsend{V}{W} : \one }

\inferrule
[Recv]
  { }
  {\Gamma \mid A \vdash \lactrecv{} : A }

\inferrule
[Spawn]
  {\Gamma \mid A \vdash M : \one }
  {\Gamma \mid C \vdash \lactspawn{M} : \pid{A}}

\inferrule
[Self]
  { }
  { \Gamma \mid A \vdash \lactself{} : \pid{A}}
\end{mathpar}
}
\caption{Syntax and typing rules for $\lact$}\label{fig:lact-syntax-typing}
\end{figure}

\section{$\lact$: A concurrent $\lambda$-calculus for actors}\label{sec:lact}

In this section, we introduce $\lact$, a core language describing actor-based
concurrency.
There are many variations of actor-based languages (by the taxonomy
of De Koster et al\.,~\cite{dekoster:actor-taxonomy}, $\lact$ is \emph{process-based}), but each have named
processes associated with a mailbox.

Typed channels are well-established, whereas typed actors are less so, partly
due to the type pollution problem.
Nonetheless, Akka Typed~\cite{akka-typed} aims to replace untyped Akka actors, so
studying a typed actor calculus is of practical relevance.

Following Erlang, we provide an explicit
$\lactrecv$ operation to allow an actor to retrieve a
message from its mailbox: unlike $\lchwd{take}$ in $\lch$, $\lactrecv{}$
takes no arguments, so it is necessary to
use a simple \emph{type-and-effect system}~\cite{gifford:effects}.
We treat mailboxes as a FIFO queues to keep $\lact$ as minimal as possible, as
opposed to considering behaviours or selective receive. This
is orthogonal to the core model of communication, as we show in
\textsection{}\ref{sec:extensions-selrecv}.

\subsection{Syntax and typing of terms}

Figure~\ref{fig:lact-syntax-typing} shows the syntax and typing rules for
$\lact$.
As with $\lch$, $\alpha$ ranges over variables and names. $\pid{A}$ is an
\emph{actor reference} or process ID, and allows messages to be sent to an
actor.  As for communication and concurrency primitives, $\lactspawn{M}$ spawns
a new actor to evaluate a computation $M$; $\lactsend{V}{W}$ sends a value $V$
to an actor referred to by reference $W$; $\lactrecv{}$ receives a value from
the actor's mailbox; and $\lactself{}$ returns an actor's own process ID.

Function arrows $A \lactto{C} B$
are annotated with a type $C$ which denotes the type of the mailbox of the actor
evaluating the term. As an example, consider a function which receives an
integer and converts it to a string (assuming a function
\metadef{intToString}):
\[
  \metadef{recvAndShow} \defeq \lambda () . \efflet{x}{\lactrecv{}}{\metadef{intToString}(x)}
\]
Such a function would have type $\one \lactto{\textit{Int}}
\textit{String}$, and as an example would not be typeable for an actor that could
only receive booleans.
Again, we work in the setting of fine-grain call-by-value; the distinction
between values and computations is helpful when reasoning about the metatheory.
We have two typing judgements: the standard judgement on values $\Gamma \vdash
V : A$, and a judgement $\Gamma \mid B \vdash M : A$ which states that a term
$M$ has type $A$ under typing context $\Gamma$, and can receive values of type
$B$. The typing of $\lactrecv{}$ and $\lactself{}$ depends on the
type of the actor's mailbox.

\subsection{Operational semantics}

\begin{figure}[t]
~Syntax of evaluation contexts and configurations
\begin{syntax}
  \text{Evaluation contexts} & E & ::= &  [~] \mid \efflet{x}{E}{M} \\
  \text{Configurations} & \config{C}, \config{D}, \config{E} & ::= &  \config{C} \parallel \config{D} \mid (\nu a) \config{C} \mid \actor{a}{M}{\seq{V}} \\
  \text{Configuration contexts} & G & ::= & [~] \mid G \parallel \config{C} \mid (\nu a) G \\
\end{syntax}

~Typing rules for configurations
\hfill \framebox{$\Gamma ; \Delta \vdash \config{C}$} \\
\footnotesize
\begin{mathpar}
\inferrule
  [Par]
  { \Gamma ; \Delta_1 \vdash \config{C}_1 \\ \Gamma ; \Delta_2 \vdash \config{C}_2}
  { \Gamma ; \Delta_1, \Delta_2  \vdash \config{C}_1 \parallel \config{C}_2 }

\inferrule
  [Pid]
  { \Gamma, a : \pid{A} ; \Delta, a : A \vdash \config{C} }
  { \Gamma ; \Delta \vdash (\nu a) \config{C} }

\inferrule
  [Actor]
  { \Gamma, a : \pid{A} \mid A \vdash M : \one \\\\ (\Gamma, a : \pid{A} \vdash V_i : A)_i }
  { \Gamma, a : \pid{A} ; a : A \vdash \actor{a}{M}{\seq{V}} }
\end{mathpar}

\caption{$\lact$ evaluation contexts and configurations}\label{fig:lact-eval-syntax}
\end{figure}

Figure~\ref{fig:lact-eval-syntax} shows the syntax of $\lact$
evaluation contexts, as well as the syntax and typing rules of $\lact$
configurations. Evaluation contexts for terms and configurations are similar to
$\lch$. The primary difference from $\lch$
is the actor configuration $\actor{a}{M}{\seq{V}}$, which can be read as ``an
actor with name $a$ evaluating term $M$, with a mailbox consisting of values
$\seq{V}$''. Whereas a term $M$ is itself a configuration in $\lch$, a term in
$\lact$ must be evaluated as part of an actor configuration in order
to support context-sensitive operations such as receiving from the mailbox.
We again stratify the reduction rules into functional reduction on terms, and
reduction on configurations.
The typing rules for $\lact$ configurations ensure that all
values contained in an actor mailbox are well-typed with respect to the
mailbox type, and that a configuration $\config{C}$ under a name restriction
$(\nu a) \config{C}$ contains an actor with name $a$.
Figure~\ref{fig:lact-reduction} shows the reduction rules for $\lact$.
Again, reduction on terms preserves
typing, and the functional fragment of $\lact$ enjoys progress.

\begin{lemma}[Preservation ($\lact$ terms)]\label{lem:lact-term-pres}
If $\Gamma \vdash M : A$ and $M \teval M'$, then $\Gamma \vdash M' : A$.
\end{lemma}

\begin{lemma}[Progress ($\lact$ terms)]\label{lem:lact-term-progress}
Assume $\Gamma$ is either empty or only contains entries of the form $a_i :
\pid{A_i}$.
If $\Gamma \mid B \vdash M : A$, then either:
\begin{enumerate}
\item $M = \effreturn{V}$ for some value $V$, or
\item $M$ can be written as $E[M']$, where $M'$ is a communication or
  concurrency primitive (i.e.\ $\lactspawn{N}$, $\lactsend{V}{W}$,
  $\lactrecv{}$, or $\lactself{}$), or
\item There exists some $M'$ such that $M \teval M'$.
\end{enumerate}
\end{lemma}

\begin{figure}[t]
\vspace{1ex}
~Reduction on terms
{\small
\[
\begin{array}{ccc}
  (\lambda x.M)\app V \teval M \{ V / x \} \quad \qquad &
  \efflet{x}{\effreturn{V}}{M} \teval M \{ V / x \} \quad \qquad &
  E[M] \teval E[M'] \\
  & & (\text{if } M \teval M')
\end{array}
\]}
\vspace{-1.5em}

~Structural congruence
{\small
\begin{mathpar}
\config{C} \parallel \config{D} \equiv \config{D} \parallel \config{C}

\config{C} \parallel (\config{D} \parallel \config{E}) \equiv (\config{C} \parallel \config{D}) \parallel \config{E}

\config{C} \parallel (\nu a)\config{D} \equiv (\nu a)(\config{C} \parallel \config{D}) \text{ if $a \not\in \fv(\config{C})$}

G[\config{C}] \equiv G[\config{D}] \text{ if $\config{C} \equiv \config{D}$}

\end{mathpar}}
~Reduction on configurations
{\small
\[
\begin{array}{rrcl}
  \textsc{Spawn} & \actor{a}{\Ex[\lactspawn{M}]}{\seq V} &\ceval&
  (\nu b)(\actor{a}{E[\effreturn{b}]}{\seq V} \parallel \actor{b}{M}{\epsilon}) \\
  & & & \quad (b \text{ is fresh}) \\
  \textsc{Send} & \actor{a}{\Ex[\lactsend{V'}{b}]}{\seq V} \parallel
        \actor{b}{M}{\seq W} &\ceval&
  \actor{a}{\Ex[\effreturn{()}]}{\seq V} \parallel
  \actor{b}{M}{\seq{W} \cdot V'} \\
  \textsc{SendSelf} & \actor{a}{\Ex[\lactsend{V'}{a}]}{\seq V} &\ceval&
    \actor{a}{\Ex[\effreturn{()}]}{\seq V \cdot V'} \\
  \textsc{Self} & \actor{a}{E[\lactself{}]}{\seq{V}} & \ceval & \actor{a}{E[\effreturn{a}]}{\seq{V}}\\
  \textsc{Receive} & \actor{a}{\Ex[\lactrecv{}]}{W \cdot \seq{V}} &\ceval&
    \actor{a}{\Ex[\effreturn{W}]}{\seq V} \\
  \textsc{Lift} & \Cx[\config{C}_1] & \ceval & \Cx[\config{C}_2] \qquad
    (\text{if } \config{C}_1 \ceval \config{C}_2) \\
    \textsc{LiftM} & \actor{a}{M_1}{\seq{V}} & \ceval & \actor{a}{M_2}{\seq{V}} \qquad
    (\text{if } M_1 \teval M_2) \\
\end{array}
\]}

\caption{Reduction on $\lact$ terms and configurations}\label{fig:lact-reduction}
\end{figure}

\subparagraph{Reduction on configurations.} While $\lch$ makes
use of separate constructs to create new processes and channels, $\lact$ uses
a single construct $\lactspawn{M}$ to spawn a new actor with an empty mailbox to
evaluate term $M$.  Communication happens directly between actors instead of
through an intermediate entity: as a result of evaluating $\lactsend{V}{a}$, the
value $V$ will be appended directly to the end of the mailbox of actor $a$.
\textsc{SendSelf} allows reflexive sending; an alternative would be to decouple
mailboxes from the definition of actors, but this complicates both the configuration
typing rules and the intuition.
\textsc{Self} returns the name of the current process, and
\textsc{Receive} retrieves the head value of a non-empty mailbox.

As before, typing is preserved modulo structural congruence and under reduction.

\begin{lemma}
  If $\Gamma ; \Delta \vdash \config{C}$ and $\config{C} \equiv \config{D}$ for
some $\config{D}$, then $\Gamma ; \Delta \vdash \config{D}$.
\end{lemma}

\begin{theorem}[Preservation ($\lact$ configurations)]\label{thm:lact-preservation}
  If $\Gamma ; \Delta \vdash \config{C}_1$ and $\config{C}_1 \ceval
  \config{C}_2$, then $\Gamma ; \Delta \vdash
\config{C}_2$.
\end{theorem}

\subsection{Progress and canonical forms}
Again, we cannot guarantee deadlock-freedom for $\lact$.
Instead, we proceed by defining a canonical form, and characterising the form of
progress that $\lact$ enjoys. The technical development follows that of $\lch$.

\begin{definition}[Canonical form ($\lact$)]
  A $\lact$ configuration $\config{C}$ is in \emph{canonical form} if
  $\config{C}$ can be written $(\nu a_1) \ldots (\nu
    a_n)(\actor{a_1}{M_1}{\seq{V_1}} \parallel \ldots \parallel
  \actor{a_n}{M_n}{\seq{V_n}})$.
\end{definition}

\begin{lemma}\label{lem:lact-semi-canon-form}
  If $\Gamma ; \Delta \vdash \config{C}$ and $\Delta = a_1 : A_1, \ldots a_k :
  A_k$, then there exists $\config{C'} \equiv \config{C}$ such that $\config{C'} = (\nu
  a_{k+1}) \ldots (\nu a_n) (\actor{a_1}{M_1}{\seq{V_1}} \parallel \ldots
\parallel \actor{a_n}{M_n}{\seq{V_n}})$.
\end{lemma}
As before, it follows as a corollary
of Lemma~\ref{lem:lact-semi-canon-form} that closed configurations can be written in
canonical form. We can therefore classify the notion of progress enjoyed by
$\lact$.

\begin{corollary}
  If $\cdot ; \cdot \vdash \config{C}$, then there exists some $\config{C'}
  \equiv \config{C}$ such that $\config{C'}$ is in canonical form.
\end{corollary}

\begin{theorem}[Weak progress ($\lact$ configurations)] \hfill \\
  Let $\cdot ; \cdot \vdash \config{C}$, $\config{C} \not\ceval$, and let
  $\config{C'} = (\nu a_1) \ldots (\nu a_n)(\actor{a_1}{M_1}{\seq{V_1}}
  \parallel \ldots \parallel \actor{a_n}{M_n}{\seq{V_n}})$ be a canonical form
  of $\config{C}$. Each actor with name $a_i$ is either of the form
  $\actor{a_i}{\effreturn{W}}{\seq{V_i}}$ for some value $W$, or
  $\actor{a_i}{E[\lactrecv{}]}{\epsilon}$.
\end{theorem}

 \section{From $\lact$ to $\lch$}\label{sec:lact-lch}

With both calculi in place, we can define the translation from
$\lact$ into $\lch$. 
The key idea is to emulate a mailbox using a channel, and
to pass the channel as an argument to each function. The translation on terms is
parameterised over the channel name, which is
used to implement context-dependent operations (i.e.,\ $\lactwd{receive}$ and
$\lactwd{self}$).
Consider again \metadef{recvAndShow}.
\[
  \metadef{recvAndShow} \defeq \lambda () . \efflet{x}{\lactrecv{}}{\metadef{intToString}(x)}
\]
A possible configuration would be an actor evaluating
$\metadef{recvAndShow} \app ()$, with some name $a$ and mailbox
with values $\seq{V}$, under a name restriction for $a$.
\[
  (\nu a) (\actor{a}{\metadef{recvAndShow} \app ()}{\seq{V}})
\]
The translation on terms takes a channel name \textit{ch} as a parameter.
As a result of the translation, we have that:
\[
  \translateactterm{\metadef{recvAndShow} \app ()}{ch} = \efflet{x}{\lchtake{ch}}{\metadef{intToString}(x)}
\]
with the corresponding configuration
  $(\nu a) (\lchbuf{a}{\translateactval{\seq{V}}} \parallel
  \translateactterm{\metadef{recvAndShow} \app ()}{a})$.
The values from the mailbox are translated pointwise and form the contents
of a buffer with name $a$. The translation of $\metadef{recvAndShow}$ is
provided with the name $a$ which is used to emulate $\lactrecv{}$.
\subsection{Translation ($\lact$ to $\lch$)}

\begin{figure}[t]
~Translation on types
{\small
\begin{mathpar}
\translateactty{\pid{A}} =  \chan{\translateactty{A}}

\translateactty{A \lactto{C} B}  =  \translateactty{A} \to \chan{\translateactty{C}} \to \translateactty{B}

\translateactty{\one} = \one
\end{mathpar}}
~Translation on values
{\small
\begin{mathpar}
  \translateactval{x} = x

  \translateactval{a} = a

  \translateactval{\lambda x.M} = \lambda x . \lambda ch . (\translateactterm{M}{ch})

  \translateactval{()} = ()
\end{mathpar}}
~Translation on computation terms
\vspace{-1.2em}
{\small
\begin{mathpar}
\translateactterm{\efflet{x}{M}{N}}{ch} = \efflet{x}{(\translateactterm{M}{ch})}{\translateactterm{N}{ch}}
\vspace{-1.2em}
\end{mathpar}
\begin{minipage}{\textwidth}
\begin{minipage}[t]{0.5\textwidth}
  \[
  \begin{array}{rcl}
  \translateactterm{V \app W}{ch} & = &
    \efflettwo{f}{(\translateactval{V} \app \translateactval{W})}{\: f \app \textit{ch}} \\
  \translateactterm{\effreturn{V}}{ch} & = & \effreturn{\translateactval{V}} \\
  \translateactterm{\lactself{}}{ch} & = & \effreturn{ch} \\
  \translateactterm{\lactrecv{}}{ch} & = & \lchtake{ch} \\
  \end{array}
  \]
\end{minipage}
~
\begin{minipage}[t]{0.45\textwidth}
\[
\begin{array}{rcl}
\translateactterm{\lactspawn{M}}{ch}
                                 & = & \efflettwo{\textit{chMb}}{\lchnewch{}} \\
                                 &   & \lchfork{(\translateactterm{M}{\textit{chMb}})}; \\
                                 &   & \effreturn{\textit{chMb}} \\
\translateactterm{\lactsend{V}{W}}{ch} & = & \lchgive{(\translateactval{V})}{(\translateactval{W})} \\
\end{array}
\]
\end{minipage}
\end{minipage}}
\vspace{0.5em}

~Translation on configurations
{\small
\begin{mathpar}
\translateactconfig{\config{C}_1 \parallel \config{C}_2} = \translateactconfig{\config{C}_1} \parallel \translateactconfig{\config{C}_2}

\translateactconfig{(\nu a)\config{C}} = (\nu a) \app \translateactconfig{\config{C}}

\translateactconfig{\actor{a}{M}{\seq{V}}} = \lchbuf{a}{\translateactval{\seq{V}}} \parallel (\translateactterm{M}{a})
\end{mathpar}}
\vspace{-1em}
\caption{Translation from $\lact$ into $\lch$}
\label{fig:lact-lch-trans}
\end{figure}

Figure~\ref{fig:lact-lch-trans} shows the formal translation from $\lact$ into
$\lch$. Of particular note is the translation on terms:
$\translateactterm{-}{ch}$ translates
a $\lact$ term into a $\lch$ term using a channel with name $ch$ to emulate
a mailbox.
An actor reference is represented as a channel reference in $\lch$; we
emulate sending a message to another actor by writing to the channel
emulating the recipient's mailbox. Key to translating $\lact$ into
$\lch$ is the translation of function arrows $A \lactto{C} B$; the
effect annotation $C$ is replaced by a second parameter $\chan{C}$,
which is used to emulate the mailbox of the actor.
Values translate to themselves, with the exception of $\lambda$
abstractions, whose translation takes an additional parameter denoting
the channel used to emulate operations on a mailbox. Given parameter
\textit{ch}, the translation function for terms emulates $\lactrecv{}$
by taking a value from \textit{ch}, and emulates $\lactself{}$ by
returning \textit{ch}.

Though the translation is straightforward, it is a \emph{global}
translation~\cite{felleisen:expressiveness}, as all functions must be
modified in order to take the mailbox channel as an additional parameter.

\subsection{Properties of the translation}
The translation on terms and values preserves typing.
We extend the translation function pointwise to typing environments:
$  \translateactty{\alpha_1 : A_1, \ldots, \alpha_n : A_n} =
  \alpha_1 : \translateactty{A_1}, \ldots, \alpha_n : \translateactty{A_n}
$.
\begin{lemma}[$\translateactval{-}$ preserves typing (terms and
  values)]\label{lem:lact-lch-term-typing} \hfill
\begin{enumerate}
\item If $\Gamma \vdash V : A$ in $\lact$, then $\translateactty{\Gamma}
  \vdash \translateactval{V} : \translateactty{A}$ in $\lch$.
\item If $\Gamma \mid B \vdash M : A$ in $\lact$, then
  $\translateactty{\Gamma}, \alpha : \chan{\translateactty{B}} \vdash
  \translateactterm{M}{\alpha} : \translateactty{A}$ in $\lch$.
\end{enumerate}
\end{lemma}

The proof is by simultaneous induction on the derivations of $\Gamma \vdash V {:} A$ and
$\Gamma \mid B \vdash M {:} A$.
To state a semantics preservation result, we also define a translation on
configurations; the translations on parallel composition and name restrictions
are homomorphic. An actor configuration $\actor{a}{M}{\seq{V}}$ is translated as
a buffer $\lchbuf{a}{\translateactval{\seq{V}}}$, (writing
$\translateactval{\seq{V}} = \translateactval{V_0} \cdot, \ldots, \cdot
\translateactval{V_n}$ for each $V_i \in \seq{V}$), composed in parallel with the
translation of $M$, using $a$ as the mailbox channel.
We can now see that the translation preserves typeability of
configurations.

\begin{theorem}[$\translateactconfig{-}$ preserves typeability (configurations)] \hfill \\ \label{thm:lact-lch-config-typing}
  If $\Gamma ; \Delta \vdash \config{C}$ in $\lact$, then
  $\translateactty{\Gamma} ; \translateactty{\Delta} \vdash \translateactconfig{\config{C}}$
  in $\lch$.
\end{theorem}

We describe semantics preservation in terms of a simulation theorem: should a
configuration $\config{C}_1$ reduce to a configuration $\config{C}_2$ in
$\lact$, then there exists some configuration $\config{D}$ in $\lch$ such
that $\translateactconfig{\config{C}_1}$ reduces in zero or more steps to $\config{D}$,
with $\config{D} \equiv \translateactconfig{\config{C}_2}$.
To establish the result, we begin by showing that $\lact$ term
reduction can be simulated in $\lch$.

\begin{lemma}[Simulation of $\lact$ term reduction in $\lch$] \label{lem:lact-lch-term-sim}\hfill \\
If $\Gamma \vdash M_1 : A$ and $M_1 \teval M_2$ in $\lact$, then given some
$\alpha$, $\translateactterm{M_1}{\alpha} \teval^* \translateactterm{M_2}{\alpha}$
in $\lch$.
\end{lemma}

Finally, we can see that the translation preserves structural congruences, and
that $\lch$ configurations can simulate reductions in $\lact$.

\begin{lemma} \label{lem:lact-lch-struct}
If $\Gamma ; \Delta \vdash \config{C}$ and $\config{C} \equiv \config{D}$, then
$\translateactconfig{\config{C}} \equiv \translateactconfig{\config{D}}$.
\end{lemma}

\begin{theorem}[Simulation of $\lact$ configurations in $\lch$]\label{thm:lact-lch-config-sim} \hfill \\
  If $\Gamma ; \Delta \vdash \config{C}_1$ and $\config{C}_1 \ceval \config{C}_2$, then
  there exists some $\config{D}$ such that
  $\translateactconfig{\config{C}_1}\ceval^{*} \config{D}$, with $\config{D} \equiv
  \translateactconfig{\config{C}_2}$.
\end{theorem}

 \section{From $\lch$ to $\lact$}\label{sec:lch-lact}

The translation from $\lact$ into $\lch$ emulates an actor mailbox
using a channel to implement operations which normally rely on the
context of the actor. Though global, the translation is
straightforward due to the limited form of communication supported by
mailboxes.
Translating from $\lch$ into $\lact$ is more challenging, as would
be expected from Figure~\ref{fig:mailboxes-as-pinned-channels}. Each
channel in a system may have a different type; each process may have access to
multiple channels; and (crucially) channels may be freely passed between
processes.

\begin{figure}[t]
~Syntax
\vspace{-0.6em}
\begin{syntax}
  \text{Types} & A,B,C & ::= & \ldots \mid A \times B \mid A + B \mid \listty{A} \mid \mu X . A \mid X
  \\ 
  \text{Values} & V, W & ::= & \ldots \mid \rec{f}{x}{M} \mid (V, W) \mid \inl{V} \mid \inr{W} \mid \roll{V} \\
  \text{Terms} & L,M,N & ::= & \ldots \mid \letin{(x,y)}{V}{M} \mid  \caseof{V}{\mkwd{inl} \; x \mapsto M; \mkwd{inr} \; y \mapsto N} \mid \unroll{V} \\
\end{syntax}
~Additional value typing rules
\hfill \framebox{$\Gamma \vdash V : A$}\\
{\footnotesize
\begin{mathpar}
\inferrule
[Rec]
  { \Gamma, x : A, f : A \to B  \vdash M : B }
  { \Gamma \vdash \rec{f}{x}{M} : A \to B }

\inferrule
[Pair]
  {\Gamma \vdash V: A \\
   \Gamma \vdash W: B}
  {\Gamma \vdash (V, W): A \times B}

\inferrule
[Inl]
  { \Gamma \vdash V : A}
  { \Gamma \vdash \inl{V} : A + B }

\inferrule
[Roll]
  { \Gamma \vdash V : A \{ \mu X . A / X \} }
  { \Gamma \vdash \roll{V} : \mu X . A }
\end{mathpar}
}
~Additional term typing rules
\hfill \framebox{$\Gamma \vdash M : A$}
{\footnotesize
\begin{mathpar}
\inferrule
[Let]
  {\Gamma \vdash V : A \times A' \\\\
   \Gamma, x:A,y:A' \vdash M : B}
  {\Gamma \vdash \letin{(x,y)}{V}{M} : B}

\inferrule
[Case]
  {\Gamma \vdash V : A + A' \\\\
   \Gamma, x : A  \vdash M : B \\ \Gamma, y : A'  \vdash N : B}
  {\Gamma \vdash \caseof{V}{\inl{x} \mapsto M ; \inr{y} \mapsto N} : B }

\inferrule
[Unroll]
  { \Gamma \vdash V : \mu X . A }
  { \Gamma \vdash \unroll{V} : A \{ \mu X . A / X \} }
\end{mathpar}
}
\vspace{-0.25em}
~Additional term reduction rules
\vspace{-1.35em}
\hfill \framebox{$M \teval M'$} \\
{\small
\begin{align*}
  (\rec{f}{x}{M}) \app V & \teval M \{ (\rec{f}{x}{M}) / f, V / x \} \\
  \letin{(x,y)}{(V,W)}{M} &\teval M \{ V / x, W / y \} \\
  \caseof{(\mkwd{inl} \app V)}{\inl{x} \mapsto M ; \inr{y} \mapsto N} &\teval M \{ V / x \} \\
  \unroll{(\roll{V})} &\teval \effreturn{V}
\end{align*}}
\vspace{-0.25em}
~Encoding of lists
{\small
\vspace{-0.6em}
\begin{mathpar}
    \listty{A} \defeq \mu X . 1 + (A \times X)

    \emptylist \defeq \roll{(\inl ())}

    V :: W  \defeq \roll{(\inr{(V, W)})}
  \end{mathpar}
  \vspace{-3em}
  \begin{align*}
    \caseof{V}{\emptylist \mapsto M; x :: y \mapsto N} &\defeq \efflet{z}{\unroll{V}}{\caseof{z}{\inl{()} \mapsto M; \inr{(x, y) \mapsto N}}}
  \end{align*}}
\vspace{-1.75em}
\caption{Extensions to core languages to allow translation from $\lch$ into $\lact$}
\label{fig:lch-lact-extensions}
\end{figure}

\subsection{Extensions to the core language}
We require several more language constructs: sums, products, recursive
functions, and iso-recursive types. Recursive functions are used to
implement an event loop, and recursive types to maintain a term-level
buffer. Products are used to record both a list of values in the
buffer and a list of pending requests.
Sum types allow the disambiguation of the two types of messages sent to an
actor: one to queue a value (emulating $\lchwd{give}$) and one to dequeue a
value (emulating $\lchwd{take}$).
Sums are also used to encode monomorphic
variant types; we write $\variant{\ell_1:A_1, \dots, \ell_n:A_n}$ for
variant types and $\variant{\ell_i = V}$ for variant values.

Figure~\ref{fig:lch-lact-extensions} shows the extensions to the core
term language and their reduction rules; we omit the symmetric rules
for $\termconst{inr}$. With products, sums, and recursive types, we
can encode lists. The typing rules are shown for $\lch$ but can be
easily adapted for $\lact$, and it is straightforward to verify that
the extended languages still enjoy progress and preservation.

\subsection{Translation strategy ($\lch$ into $\lact$)}
\begin{figure}[t]
  \centering
  \begin{subfigure}[t]{0.4\textwidth}
    \centering
    \includegraphics[width=0.75\textwidth]{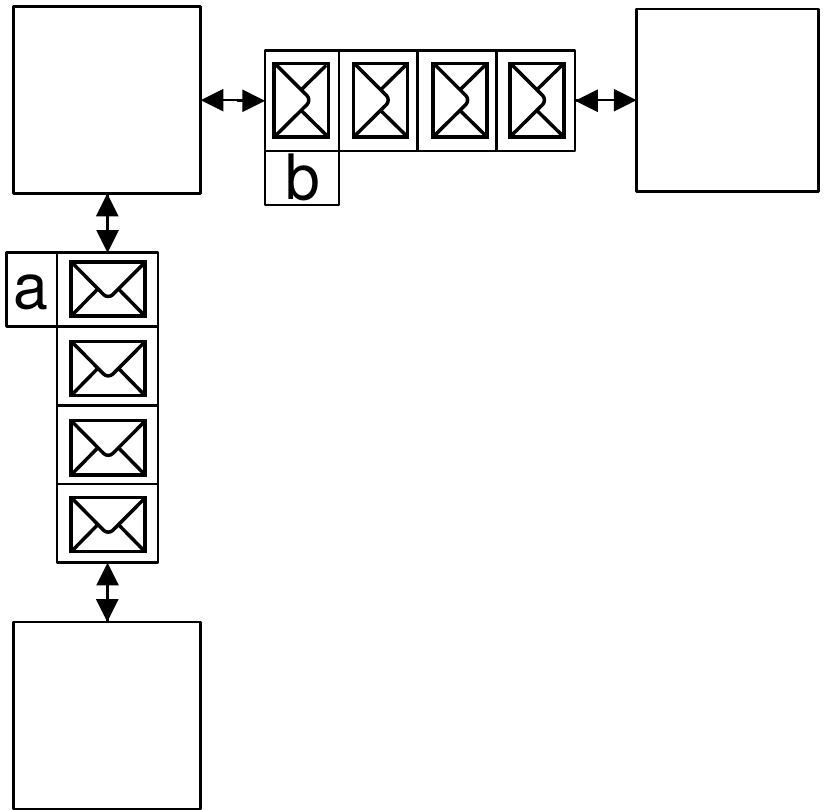}
    \caption{Before Translation}
  \end{subfigure}
  ~
  \begin{subfigure}[t]{0.4\textwidth}
    \centering
    \includegraphics[width=0.75\textwidth]{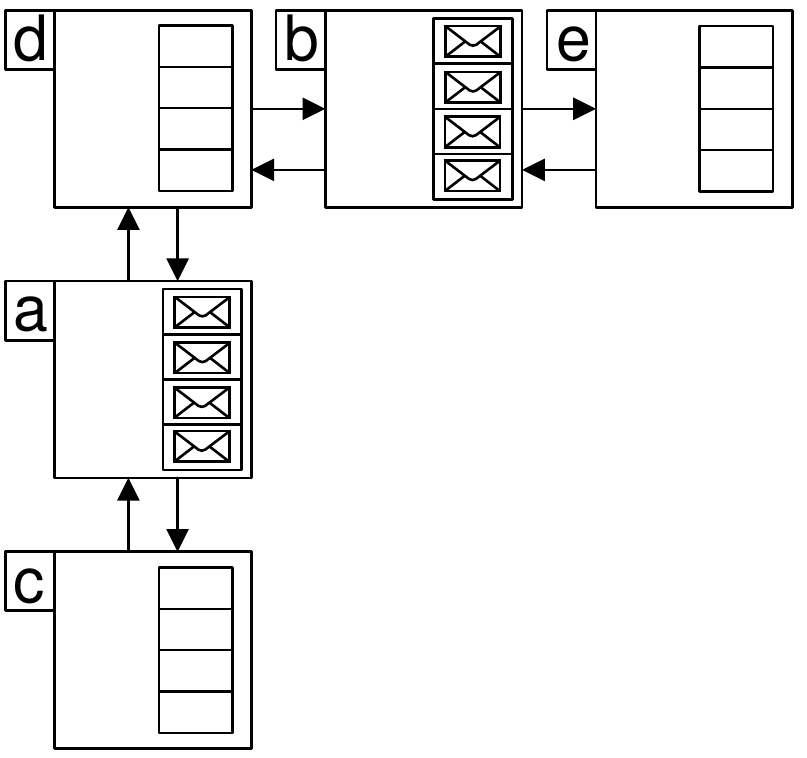}
    \caption{After Translation}
  \end{subfigure}
  \caption{Translation strategy: $\lch$ into $\lact$}
  \label{fig:lch-lact-translation-strategy}
\end{figure}

To translate typed actors into typed channels (shown in
Figure~\ref{fig:lch-lact-translation-strategy}), we emulate each
channel using an actor process, which is crucial in retaining the mobility of
channel endpoints.
Channel types describe the typing of a \emph{communication medium}
between communicating processes, where processes are unaware of the
identity of other communicating parties, and the types of messages
that another party may receive. Unfortunately, the same does not hold
for mailboxes.  Consequently, we require that before translating into
actors, \emph{every channel has the same type}. Although this may seem
restrictive, it is both possible and safe to transform a $\lch$
program with multiple channel types into a $\lch$ program with a
single channel type.

As an example, suppose we have a program which contains channels carrying values
of types $\mkwd{Int}$, $\mkwd{String}$, and $\chan{\mkwd{String}}$. It is possible to
construct a recursive variant type $\mu X . \langle \ell_1 : \mkwd{Int}, \ell_2 : \mkwd{String},
\ell_3 : \chan{X} \rangle$ which can be assigned to all channels in the
system. Then, supposing we wanted to send a $5$ along a channel which previously
had type $\chan{\mkwd{Int}}$, we would instead send a value
$\roll{\variant{\ell_1 = 5}}$ (where $\roll{V}$ is the introduction rule for an
iso-recursive type).
Appendix A~\cite{mm-extended} provides more details.

\subsection{Translation}
We write $\lch$ judgements of the form $\jarg{B} \Gamma \vdash M : A$ for a
term where all channels have type $B$, and similarly for value and configuration
typing judgements. Under such a judgement, we can write $\chanzero{}$ instead
of $\chan{B}$.

\subparagraph{Meta level definitions.}
The majority of the translation lies within the translation of $\lchnewch{}$,
which makes use of the meta-level definitions \metadef{body} and
\metadef{drain}.
The \metadef{body} function emulates a channel. Firstly, the
actor receives a message \textit{recvVal}, which is either
of the form $\inl{V}$ to store a message $V$, or $\inr{W}$ to
request that a value is dequeued and sent to the actor with ID
$W$. We assume a standard implementation of list concatenation
$(\doubleplus)$.
If the message is $\inl{V}$, then $V$ is appended to the tail of the list of
values stored in the channel, and the new state is passed as an argument to
\metadef{drain}. If the message is $\inr{W}$, then the process ID $W$ is
appended to the end of the list of processes waiting for a value.
The \metadef{drain} function satisfies all requests that can be satisfied,
returning an updated channel state. Note that \metadef{drain} does not need
to be recursive, since one of the lists will either be empty or a singleton.

\begin{figure}[t]
~Translation on types (wrt.\ a channel type $C$)
{\small
\begin{mathpar}
  \translatechty{\chanzero} = \pid{\translatechty{C} + \pid{\translatechty{C}}}

\translatechty{A \to B} = \translatechty{A} \lactto{\translatechty{C}} \translatechty{B}

\end{mathpar}}
~Translation on communication and concurrency primitives

{\small
\begin{minipage}{0.5\textwidth}
\[
\begin{array}{rcl}
  \translatechterm{\lchfork{M}} & = & \efflet{x}{\lactspawn{\translatechterm{M}}}{\effreturn{()}} \\
\translatechterm{\lchgive{V}{W}} & = & \lactsend{(\inl{\translatechval{V}})}{\translatechval{W}} \\
\translatechterm{\lchnewch} & = & \lactspawn{(\metadef{body} \, (\emptylist{}, \emptylist))} \\
\end{array}
\]
\end{minipage}
~
\begin{minipage}{0.5\textwidth}
\[
  \begin{array}{rcl}
\translatechterm{\lchtake{V}} & = &
 {
  \begin{aligned}[t]
    & \efflettwo{\textit{selfPid}}{\lactself{}} \\
    & \lactsend{(\mkwd{inr} \, \textit{selfPid})}{\translatechval{V}}; \\
    & \lactrecv{}
  \end{aligned}
  } \\
\end{array}
\]
\end{minipage}
}

~Translation on configurations
{\small
\[
  \begin{array}{ccc}
{\translatechconfig{\config{C}_{1} \parallel {\config{C}}_{2}}} =
  {\translatechconfig{\config{C}_{1}}} \parallel
  {\translatechconfig{\config{C}_{2}}} \qquad \quad &
  \translatechconfig{(\nu a) \config{C}} = (\nu a)
  \translatechconfig{\config{C}} \qquad \quad &
\translatechconfig{M} = (\nu a) (\actor{a}{\translatechterm{M}}{\epsilon})\quad
\\
& & a \text{ is a fresh name }
\end{array}
\]
\vspace{-1.5em}
\[
\begin{array}{rcl}
\translatechconfig{\lchbuf{a}{\seq{V}}} & = &
  \actor{a}{\metadef{body} \, (\translatechval{\seq{V}}, \emptylist)}{\epsilon}
  \qquad \text{where } \translatechval{\seq{V}} = \translatechval{V_0} :: \ldots :: \translatechval{V_n} :: \emptylist
\end{array}
\]}
~Meta level definitions
\vspace{-1.5em}

{\small
\begin{minipage}[t]{0.5\textwidth}
\[
\begin{array}{l}
      \metadef{body} \defeq{}
            \termconst{rec} \; g (\textit{state}) \; . \\
         \quad \efflettwo{\textit{recvVal}}{\lactrecv} \\
         \quad \letintwo{(\textit{vals}, \textit{pids})}{\textit{state}} \\
         \quad \caseofone{\textit{recvVal} } \\
         \qquad \quad \inl{v} \mapsto
         {
         \begin{aligned}[t]
           & \efflettwo{\textit{vals}'}{\listappend{vals}{\listterm{v}}} \\
           & \efflettwo{\textit{state}'}{\metadef{drain} \app
           (\textit{vals}', \textit{pids})} \\
           & g \app (\textit{state}') \\
         \end{aligned}
         } \\
        \qquad \quad \inr{\textit{pid}} \mapsto
         {
         \begin{aligned}[t]
           & \efflettwo{\textit{pids}'}{\listappend{\textit{pids}}{\listterm{pid}}} \\
           & \efflettwo{\textit{state}'}{\metadef{drain} \app
           (\textit{vals}, \textit{pids}')} \\
           & g \app (\textit{state}') \quad \} \\
         \end{aligned}
         }
\end{array}
\]
\end{minipage}
~
\begin{minipage}[t]{0.4\textwidth}
\[
\begin{array}{l}
    \metadef{drain} \defeq{} \lambda x . \\
        \quad \letintwo{(\textit{vals}, \textit{pids})}{x} \\
        \quad \caseofone{\textit{vals}} \\
        \qquad \emptylist \mapsto \effreturn{(\textit{vals}, \textit{pids})} \\
        \qquad \listcons{v}{vs} \mapsto \\
        \qquad \quad \caseofone{\textit{pids}} \\
        \qquad \quad \quad \quad \emptylist{} \mapsto \effreturn{(\textit{vals}, \textit{pids})}\\
        \qquad \quad \quad \quad \listcons{\textit{pid}}{\textit{pids}} \mapsto
        \begin{aligned}[t]
           & \lactsend{v}{\textit{pid}} ; \\
           & \effreturn{(\textit{vs}, \textit{pids})} \\
        \end{aligned} \\
        \qquad \quad \} \quad \}
\end{array}
\]
\end{minipage}
}
\caption{Translation from $\lch$ into $\lact$}
\label{fig:lch-lact-translation}
\end{figure}

\subparagraph{Translation on types.}
Figure~\ref{fig:lch-lact-translation} shows the translation from $\lch$
into $\lact$.
The translation function on types
$\translatechty{-}$ is defined with respect to the type of all channels $C$ and
is used to annotate function arrows and to assign a parameter to
$\mkwd{ActorRef}$ types. The (omitted) translations on sums,
products, and lists are homomorphic.
The translation of $\chanzero{}$ is $\pid{\translatechty{C} +
\pid{\translatechty{C}}}$, meaning an
actor which can receive a request to either store a value of type
$\translatechty{C}$, or to dequeue a value and send it to a process ID of
type $\pid{\translatechty{C}}$.

\subparagraph{Translation on communication and concurrency primitives.}
We omit the translation on values and functional terms, which are homomorphisms.
Processes in $\lch$ are anonymous, whereas all actors in $\lact$ are
addressable; to emulate $\lchwd{fork}$, we therefore discard the reference
returned by $\lactwd{spawn}$. The translation of
$\lchwd{give}$ wraps the translated value to be sent in the left injection of a
sum type, and sends to the translated channel name $\translatechval{W}$.
To emulate $\lchwd{take}$, the process ID (retrieved using $\lactself{}$) is
wrapped in the right injection and sent to the actor emulating the channel, and
the actor waits for the response message.
Finally, the translation of $\lchnewch{}$ spawns a new actor to execute
\metadef{body}.

\subparagraph{Translation on configurations.}
The translation function $\translatechconfig{-}$ is homomorphic on parallel
composition and name restriction. Unlike $\lch$, a term cannot exist outwith an
enclosing actor context in $\lact$, so the translation of a
process evaluating term $M$ is an actor evaluating $\translatechterm{M}$ with
some fresh name $a$ and an empty mailbox,
enclosed in a name restriction.
A buffer is translated to an actor with an empty mailbox, evaluating
\metadef{body} with a state containing the (term-level) list of values
previously stored in the buffer.

Although the translation from $\lch$ into $\lact$, is much more verbose than
the translation from $\lact$ to $\lch$, it is (once all channels have the same
type) a \emph{local transformation}~\cite{felleisen:expressiveness}.

\subsection{Properties of the translation}
Since all channels in the
source language of the translation have the same type, we can assume that each
entry in the codomain of $\Delta$ is the same type $B$.

\begin{definition}[Translation of typing environments wrt.\ a channel type $B$] \hfill
  \begin{enumerate}
    \item If $\Gamma = \alpha_1 {:} A_1, \ldots, \alpha_n : A_n$, define
        $\translatechty{\Gamma} = \alpha_1 : \translatechty{A_1}, \ldots, \alpha_n : \translatechty{A_n}$.
    \item Given a $\Delta = a_1 : B, \ldots, a_n : B$, define
      $\translatechty{\Delta} = \\ a_1 : (\translatechty{B} +
      \pid{\translatechty{B}}), \ldots,  a_n :
      (\translatechty{B} + \pid{\translatechty{B}})$.
  \end{enumerate}
\end{definition}

\noindent
The translation on terms preserves typing.
\begin{lemma}[$\translatechval{-}$ preserves typing (terms and values)]\label{lem:lch-lact-term-typing}\hfill \\ \vspace{-1.5em}
  \begin{enumerate}
    \item If $\jarg{B} \Gamma \vdash V {:} A$, then
      $\translatechty{\Gamma} \vdash \translatechval{V} {:}
      \translatechty{A}$.
    \item If $\jarg{B} \Gamma  \vdash M {:} A$, then
      $\translatechty{\Gamma} \mid \translatechty{B} \vdash
      \translatechterm{M} {:} \translatechty{A}$.
  \end{enumerate}
\end{lemma}

\noindent
The translation on configurations also preserves typeability.
We write $\Gamma \compat
\Delta$ if for each $a : A \in \Delta$, we have that $a : \chan{A} \in \Gamma$;
for closed configurations this is ensured by \textsc{Chan}. This is necessary since
the typing rules for $\lact$ require that the local actor name is present in the
term environment to ensure preservation in the presence of
$\lactself{}$, but there is no such restriction in $\lch$.

\begin{theorem}[$\translatechconfig{-}$ preserves typeability (configurations)]\hfill \\
  \label{thm:lch-lact-config-typing}
If $\jarg{A} \Gamma ; \Delta \vdash \config{C}$ with $\Gamma \compat \Delta$,
then $\translatechty{\Gamma} ; \translatechty{\Delta} \vdash
\translatechconfig{\config{C}}$.
\end{theorem}

\noindent
It is clear that reduction on translated $\lch$ terms can simulate reduction in
$\lact$.
\begin{lemma}\label{lem:lch-lact-term-sim}
If $\jarg{B} \Gamma \vdash M_1 : A $ and $M_1 \teval M_2$, then
$\translatechterm{M_1} \teval \translatechterm{M_2}$.
\end{lemma}
Finally, we show that $\lact$ can simulate $\lch$.
\begin{lemma}
  If $\Gamma ; \Delta \vdash \config{C}$ and $\config{C} \equiv \config{D}$,
  then $\translatechconfig{\config{C}} \equiv \translatechconfig{\config{D}}$.
\end{lemma}

\begin{theorem}[Simulation ($\lact$ configurations in $\lch$)]\label{thm:lch-lact-config-sim} \hfill \\
If $\jarg{A} \Gamma ; \Delta \vdash \config{C}_1$, and $\config{C}_1 \ceval \config{C}_2$,
then there exists some $\config{D}$ such that $\translatechconfig{\config{C}_1}
\ceval^* \config{D}$ with $\config{D} \equiv \translatechconfig{\config{C}_2}$.
\end{theorem}

\subparagraph{Remark.}
The translation from $\lch$ into $\lact$ is more involved than the translation
from $\lact$ into $\lch$ due to the asymmetry shown in
Figure~\ref{fig:mailboxes-as-pinned-channels}. Mailbox types are less precise;
generally taking the form of a large variant type.

Typical implementations of this translation use synchronisation mechanisms such
as futures or shared memory (see \textsection{}\ref{sec:extensions-sync}); the
implementation shown in the Hopac documentation uses ML
references~\cite{hopac-actors}.
Given the ubiquity of these abstractions, we were surprised to
discover that the additional expressive power of synchronisation is
not \emph{necessary}.
Our original attempt at a synchronisation-free translation was
type-directed. We were surprised to discover that the translation can
be described so succinctly after factoring out the coalescing step,
which precisely captures the type pollution problem.

 \section{Extensions}\label{sec:extensions}

In this section, we discuss common extensions to channel- and
actor-based languages.
Firstly, we discuss synchronisation, which is ubiquitous in practical
implementations of actor-inspired languages. Adding synchronisation
simplifies the translation from channels to actors, and relaxes the
restriction that all channels must have the same type.
Secondly, we consider an extension with Erlang-style selective
receive, and show how to encode it in $\lact$.
Thirdly, we discuss how to nondeterministically choose a message from
a collection of possible sources, and
finally, we discuss what the translations tell us about the nature of
behavioural typing disciplines for actors.
Establishing exactly how the latter two extensions fit into our
framework is the subject of ongoing and future work.

\begin{figure}[t]
~Additional types, terms, configuration reduction rule, and equivalence
{\small
\begin{mathpar}
\text{Types} ::= \pidtwo{A}{B} \mid \ldots

\text{Terms} ::= \lactwait{V} \mid \ldots
\end{mathpar}
\[
  \begin{array}{c}
  \actor{a}{E[\lactwait{b}]}{\seq{V}} \parallel \actor{b}{\effreturn{V'}}{\seq{W}} \ceval
  \actor{a}{E[\effreturn{V'}]}{\seq{V}} \parallel \actor{b}{\effreturn{V'}}{\seq{W}} \\
  (\nu a) (\actor{a}{\effreturn{V}}{\seq{V}}) \parallel \config{C} \equiv \config{C}
  \end{array}
\]
}
~Modified typing rules for terms
\hfill \framebox{$\Gamma \mid A, B \vdash M : A$}
{\footnotesize
\begin{mathpar}
  \inferrule
  [Sync-Spawn]
  { \Gamma \mid A, B  \vdash M : B }
  { \Gamma \mid C, C' \vdash \lactspawn{M} : \pidtwo{A}{B} }

  \inferrule
  [Sync-Wait]
  { \Gamma \vdash V : \pidtwo{A}{B} }
  { \Gamma \mid C, C' \vdash \lactwait{V} : B }

  \inferrule
  [Sync-Self]
  { }
  { \Gamma \mid A, B \vdash \lactself{} : \pidtwo{A}{B} }
\end{mathpar}
}
~Modified typing rules for configurations
\hfill \framebox{$\Gamma ; \Delta \vdash \config{C}$}
{ \footnotesize
\begin{mathpar}
  \inferrule
  [Sync-Actor]
  { \Gamma, a{:} \pidtwo{A}{B} \vdash M {:} B \\\\ (\Gamma, a{:} \pidtwo{A}{B} \vdash
  V_i {:} A)_i }
  { \Gamma, a: \pidtwo{A}{B} ; a{:} (A, B) \vdash \actor{a}{M}{\seq{V}} }

  \inferrule
  [Sync-Nu]
  { \Gamma, a : \pidtwo{A}{B} ; \Delta, a : (A, B) \vdash \config{C} }
  { \Gamma ; \Delta \vdash (\nu a) \config{C} }
\end{mathpar}
}

~Modified translation
\vspace{-1em}

{\small
  \begin{minipage}{0.5\textwidth}
\begin{align*}
  &\translatechtysync{\chan{A}} = \\
  & \qquad \pidtwo{\translatechtysync{A} +
  \pidtwo{\translatechtysync{A}}{\translatechtysync{A}}  }{\one} \\
  &\translatechtysync{A \to B} = \translatechtysync{A} \lacttotwo{C}{\one} \translatechtysync{B} \\
\end{align*}
\end{minipage}
~
\begin{minipage}{0.5\textwidth}
\begin{align*}
  &\translatechterm{\lchtake{V}} =
  {
  \begin{aligned}[t]
    & \efflettwonoin{\textit{requestorPid}}{\lactwd{spawn} \: (} \\
    & \qquad
        \begin{aligned}[t]
          & \efflettwo{\textit{selfPid}}{\lactself{}} \\
          & \lactsend{(\inr{\textit{selfPid}})}{\translatechval{V}}; \\
          & \lactrecv{} ) \: \termconst{in}
        \end{aligned}\\
    & \lactwait{\textit{requestorPid}}
  \end{aligned}
  }  
\end{align*}
\end{minipage}
}

\caption{Extensions to add synchronisation to $\lact$}
\label{fig:lact-sync-extension}
\end{figure}

\subsection{Synchronisation}\label{sec:extensions-sync}

Although communicating with an actor via asynchronous message passing
suffices for many purposes, implementing ``call-response'' style interactions
can become cumbersome.
Practical implementations such as Erlang and Akka implement some way
of synchronising on a result: Erlang achieves this by generating a
unique reference to send along with a request, \emph{selectively
receiving} from the mailbox to await a response tagged with the same
unique reference. Another method of synchronisation embraced by the
Active Object community~\cite{lavender:active-objects,
deboer:complete-future, johnsen:abs} and Akka
is to generate a \emph{future variable} which is populated with the
result of the call.

Figure~\ref{fig:lact-sync-extension} details an extension of $\lact$
with a synchronisation primitive, $\lactwd{wait}$, which encodes a
deliberately restrictive form of synchronisation capable of emulating
futures. The key idea behind $\lactwd{wait}$ is it allows some actor $a$ to
block until an actor $b$ evaluates to a value; this value is then returned
directly to $a$, bypassing the mailbox.
A variation of the $\lactwd{wait}$
primitive is implemented as part of the Links~\cite{cooper:links}
concurrency model. This is but one of multiple ways of allowing
synchronisation; first-class futures, shared reference cells, or selective
receive can achieve a similar result.  We discuss $\lactwd{wait}$ as it avoids
the need for new configurations.

We replace the unary type constructor for process IDs with a binary
type constructor $\pidtwo{A}{B}$, where $A$ is the type of messages
that the process can receive from its mailbox, and $B$ is the type of
value to which the process will eventually evaluate. We assume that
the remainder of the primitives are modified to take the additional
effect type into account.
We can now adapt the previous translation from $\lch$ to $\lact$, making
use of $\lactwd{wait}$ to avoid the need for the coalescing
transformation.
Channel references are translated into actor references which can either
receive a value of type $A$, or the PID of a process which can receive a value
of type $A$ and will eventually evaluate to a value of type $A$. Note
that the unbound annotation $C, 1$ on function arrows reflects that
the mailboxes can be of \emph{any} type, since the mailboxes are
unused in the actors emulating threads.

The key idea behind the modified translation is to spawn a fresh actor which
makes the request to the channel and blocks waiting for the response. Once the
spawned actor has received the result, the result can be
retrieved synchronously using $\lactwd{wait}$ \emph{without} reading from the mailbox.
The previous soundness theorems adapt to the new setting.
\begin{theorem}
  If $\Gamma ; \Delta \vdash \config{C}$ with $\Gamma \compat \Delta$, then
  $\translatechtysync{\Gamma} ; \translatechtysync{\Delta} \vdash
  \translatechconfig{\config{C}}$.
\end{theorem}

\begin{theorem}
  If $\Gamma ; \Delta \vdash \config{C}_1$ and $\config{C}_1 \ceval
  \config{C}_2$, then there exists some $\config{D}$ such that
  $\translatechconfig{\config{C}} \ceval^* \config{D}$ with $\config{D} \equiv
  \translatechconfig{\config{C}_2}$.
\end{theorem}
The translation in the other direction requires named threads and a $\lchwd{join}$
construct in $\lch$.

 \begin{figure}[t]
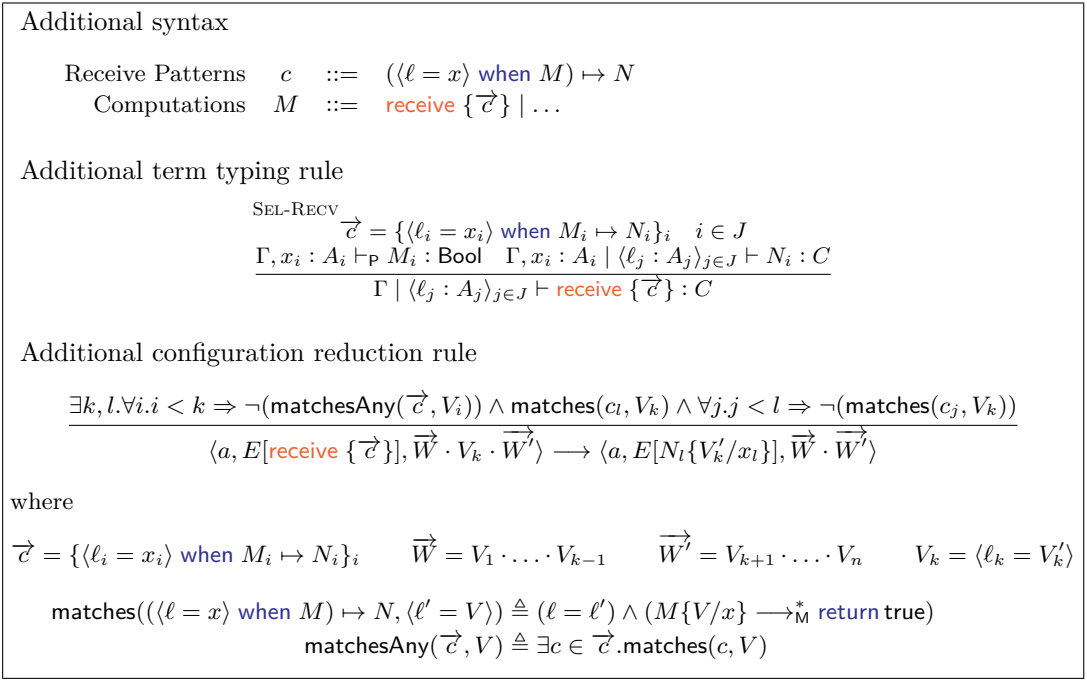

~Additional syntax
{\small
  \[
  \begin{array}{rccl}
    \text{Receive Patterns} & c & ::= & (\variant{\ell = x} \when M) \mapsto N \\
    \text{Computations} & M & ::= & \selrecv{\seq{c}} \mid \ldots \\
  \end{array}
\]
}

~Additional term typing rule
{\footnotesize
  \begin{mathpar}
    \inferrule
      [Sel-Recv]
      {  \seq{c} = \{ \variant{\ell_i = x_i} \when M_i \mapsto N_i \}_i \\ i \in J \\\\
        \Gamma, x_i : A_i \vdash_{\textsf{P}} M_i : \boolty \\
        \Gamma, x_i : A_i \mid \variant{\ell_j : A_j}_{j \in J} \vdash N_i : C}
        { \Gamma \mid \variant{\ell_j : A_j}_{j \in J} \vdash \selrecv{\seq{c}} : C }
  \end{mathpar}}

  ~Additional configuration reduction rule
{\small
\begin{mathpar}
  \inferrule
  { \exists k, l . \forall i. i < k \Rightarrow \neg (\matchesany(\seq{c}, V_i)) \wedge
    \matches(c_l, V_k) \wedge \forall j . j < l \Rightarrow \neg (\matches(c_j, V_k))
  }
  { \actor{a}{E[\selrecv{\seq{c}}]}{\seq{W} \cdot V_k \cdot \seq{W'}} \ceval \actor{a}{E[N_l\{ V'_k / x_l
      \} ]}{\seq{W}
  \cdot \seq{W'} } }
\end{mathpar}
where
\begin{mathpar}
    \seq{c} = \{ \variant{\ell_i = x_i} \when M_i \mapsto N_i \}_i

    \seq{W} = V_1 \cdot \ldots \cdot V_{k - 1}

    \seq{W'} = V_{k + 1} \cdot \ldots \cdot V_{n}

    V_k = \variant{\ell_k = V'_k}
\end{mathpar}
\vspace{-2em}
  \begin{align*}
    \matches((\variant{\ell = x} \when M) \mapsto N, \variant{\ell' = V}) &\defeq
  (\ell = \ell') \wedge (M \{ V / x \} \teval^* \effreturn{\ttrue}) \\
  \matchesany(\seq{c}, V) &\defeq \exists c \in \seq{c} . \matches(c, V)
\end{align*}
}
\vspace{-1.5em}
\caption{Additional syntax, typing rules, and reduction rules for $\lact$ with selective receive}
\label{fig:selrecv-extensions}
\end{figure}

 \subsection{Selective receive}\label{sec:extensions-selrecv}

The $\lactrecv{}$ construct in $\lact$ can only read the first message
in the queue, which is cumbersome as it often only makes sense for an
actor to handle a subset of messages at a given time.

In practice, Erlang provides a \emph{selective receive} construct,
matching messages in the mailbox against multiple pattern clauses.
Assume we have a mailbox containing values $V_1, \ldots V_n$ and
evaluate $\selrecv{c_1, \ldots, c_m}$. The
construct first tries to match value $V_1$ against clause
$c_1$---if it matches, then the body of
$c_1$ is evaluated, whereas if it fails, $V_1$ is tested against $c_2$
and so on. Should $V_1$ not match any pattern, then the
process is repeated until $V_n$ has been tested against $c_m$. At
this point, the process blocks until a matching message arrives.

More concretely, consider an actor with mailbox type
$C = \variant{\textsf{PriorityMessage} \of \textsf{Message},\\
\textsf{StandardMessage} \of \textsf{Message}, \textsf{Timeout} \of \one }$
which can receive both high- and low-priority messages.  
Let \textit{getPriority} be a function which extracts a priority from a message.

\noindent
Now consider the following actor:
{\small
\[
\begin{array}{l}
  \lactrecv \: \{  \\
    \quad
    \begin{array}{l@{\when}l@{~\mapsto~}l}
      \variant{\textsf{PriorityMessage} = \textit{msg}} &
        (\textit{getPriority} \app \textit{msg}) > 5 &
          \textit{handleMessage} \app \textit{msg} \\
       \variant{\textsf{Timeout} = \textit{msg}} & \ttrue & () \\
    \end{array} \\
  \}; \\
  \lactrecv \: \{  \\
    \quad
    \begin{array}{l@{\when}l@{~\mapsto~}l}
      \variant{\textsf{PriorityMessage} = \textit{msg}} & \ttrue &
        \textit{handleMessage} \app \textit{msg} \\
      \variant{\textsf{StandardMessage} = \textit{msg}} & \ttrue &
        \textit{handleMessage} \app \textit{msg} \\
      \variant{\textsf{Timeout} = \textit{msg}} & \ttrue & () \\
    \end{array} \\
  \} \\
\end{array}
\]}
This actor begins by handling a message only if it has a priority
greater than 5. After the timeout message is received, however, it
will handle any message---including lower-priority messages that were
received beforehand.

Figure~\ref{fig:selrecv-extensions} shows the additional syntax,
typing rule, and configuration reduction rule required to encode
selective receive; the type $\boolty$ and logical operators are
encoded using sums in the standard way. We write $\Gamma
\vdash_\textsf{P} M : A$ to mean that under context $\Gamma$, a term
$M$ which does not perform any communication or concurrency actions
has type $A$. Intuitively, this means that no subterm of $M$ is a
communication or concurrency construct.

The $\selrecv{\seq{c}}$ construct models an ordered sequence of receive pattern
clauses $c$ of the form $(\variant{\ell = x} \when M) \mapsto N$, which can be
read as ``If a message with body $x$ has label $\ell$ and satisfies predicate
$M$, then evaluate $N$''.
The typing rule for $\selrecv{\seq{c}}$ ensures that for each
pattern $\variant{\ell_i = x_i} \when M_i \mapsto N_i$ in $\seq{c}$, we have
that there exists some $\ell_i : A_i$ contained in the mailbox variant type;
and when $\Gamma$ is extended with $x_i : A_i$, that the guard $M_i$ has type
$\boolty$ and the body $N_i$ has the same type $C$ for each branch.

\begin{figure}[t]
~Translation on types
{\small
  \begin{mathpar}
    \seltrans{\pid{\variant{\ell_i : A_i}_i}} =
    \pid{\variant{\ell_i : \seltrans{A_i}}_i}

    \seltrans{A \times B} = \seltrans{A} \times \seltrans{B}

    \seltrans{A + B} = \seltrans{A} + \seltrans{B}

    \seltrans{\mu X . A} = \mu X . \seltrans{A}

    \seltrans{A \lactto{C} B} =
    \seltrans{A} \lactto{\seltrans{C}} \listty{\seltrans{C}} \lactto{\seltrans{C}} (\seltrans{B}
    \times \listty{\seltrans{C}})
  \end{mathpar}
  \quad where $C$ = $\variant{\ell_i : A'_i}_i$, and $\seltrans{C} = \variant{\ell_i :
  \seltrans{A'_i}}_i$
  }

~Translation on values
\vspace{-0.5em}
{\small
\begin{mathpar}
  \seltrans{\lambda x . M} = \lambda x . \lambda \textit{mb} . (\seltransterm{M}{\textit{mb}})

  \seltrans{\rec{f}{x}{M}} = \rec{f}{x}{\lambda \textit{mb} . (\seltransterm{M}{\textit{mb}})}
\end{mathpar}
\vspace{-2em}
}

~Translation on computation terms (wrt.\ a mailbox type $\variant{\ell_i : A_i}_i$)
\vspace{-1.25em}

{\small
\[
  \begin{array}{rcl}
    \seltransterm{V \app W}{\textit{mb}} & = &
      \efflet{\textit{f}}{(\seltrans{V} \app \seltrans{W})}{f \app \textit{mb}} \\
    \seltransterm{\effreturn{V}}{\textit{mb}} & = & \effreturn{(\seltrans{V}, \textit{mb})} \\
    \seltransterm{\efflet{x}{M}{N}}{\textit{mb}} & = &
      \efflettwo{\textit{resPair}}{\seltransterm{M}{\textit{mb}}}
        \:\letin{(\textit{x}, \textit{mb}')}{\textit{resPair }}{\seltransterm{N}{\textit{mb}'}} \\
      \seltransterm{\lactself}{\textit{mb}} & = &
      \efflet{\textit{selfPid}}{\lactself}{\effreturn{(\textit{selfPid}, \textit{mb})}} \\
      \seltransterm{\lactsend{V}{W}}{\textit{mb}} & = &
      \efflet{x}{\lactsend{(\seltrans{V})}{(\seltrans{W})}}{\effreturn{(x, \textit{mb})}} \\
      \seltransterm{\lactspawn{M}}{\textit{mb}} & = &
        \termconst{let} \: \textit{spawnRes} \Leftarrow \lactwd{spawn} (\seltrans{M}{\emptylist}) \:
                \termconst{in} \: \effreturn{(\textit{spawnRes}, \textit{mb})} \\
      \seltransterm{\selrecv{\seq{c}}}{\textit{mb}} & = & \find{\seq{c}}{\textit{mb}}
    \end{array}
\]
\vspace{-1em}
}

~Translation on configurations
\vspace{-0.75em}
{\small
  \[
  \begin{array}{rcl}
    \seltrans{(\nu a) \config{C}} & = & \{ (\nu a) \config{D} \mid
    \config{D} \in \seltrans{C} \} \\
    \seltrans{\config{C}_1 \parallel \config{C}_2 } & = &
    \{ \config{D}_1 \parallel \config{D}_2 \mid
      \config{D}_1 \in \seltrans{\config{C}_1} \wedge \config{D}_2 \in
    \seltrans{\config{C}_2} \}  \\
    \seltrans{\actor{a}{M}{\seq{V}}} & = &
      \{ \actor{a}{\seltransterm{M}{\emptylist}}{\seltransterm{\seq{V}}} \} \: \cup \\
      & & \{ \actor{a}{(\seltransterm{M}{\seq{W^1_i}})}{\seq{W^2_i}} \mid i \in 1 .. n \}
  \end{array}
  \begin{array}{l}
    \text{where} \\
     \quad \seq{W^1_i} = \seltrans{V_1} :: \ldots ::
    \seltrans{V_i} :: \emptylist \\
     \quad \seq{W^2_i} = \seltrans{V_{i + 1}} \cdot \ldots \cdot \seltrans{V_n}
  \end{array}
\]}
\vspace{-1em}

\caption{Translation from $\lact$ with selective receive into $\lact$}
\label{fig:lact-sel-recv-trans}
\end{figure}
 
The reduction rule for selective receive is inspired by that
of Fredlund~\cite{fredlund:thesis}. Assume that the
mailbox is of the form $V_1 \cdot \ldots \cdot V_k \cdot \ldots V_n$, with
$\seq{W} = V_1 \cdot \ldots \cdot V_{k - 1}$ and $\seq{W'} = V_{k + 1}
\cdot \ldots \cdot V_n$. The $\textsf{matches}(c, V)$ predicate holds if
the label matches, and the branch guard evaluates to true. The
$\matchesany(\seq{c}, V)$ predicate holds if $V$ matches any pattern in
$\seq{c}$. The key idea is that $V_k$ is the first value to satisfy a pattern.
The construct evaluates to the body of the matched
pattern, with the message payload $V'_k$ substituted for the pattern variable
$x_k$; the final mailbox is $\seq{W} \cdot \seq{W'}$
(that is, the original mailbox without $V_k$).

Reduction in the presence of selective receive preserves typing.

\begin{theorem}[Preservation ($\lact$ configurations with selective receive)]
  If $\Gamma ; \Delta \mid \variant{\ell_i : A_i}_i \vdash \config{C}_1$ and $\config{C}_1 \ceval \config{C}_2$, then
  $\Gamma ; \Delta \mid \variant{\ell_i : A_i}_i \vdash \config{C}_2$.
\end{theorem}
\subparagraph{Translation to $\lact$.}
Given the additional constructs used to translate $\lch$ into $\lact$, it is
possible to translate $\lact$ with selective receive into plain $\lact$.
Key to the translation is reasoning about values in the
mailbox at the term level; we maintain a term-level `save queue' of values that
have been received but not yet matched, and can loop through the list to find
the first matching value.
Our translation is similar in spirit to the ``stashing'' mechanism
described by Haller~\cite{haller:integration} to emulate selective receive in Akka,
where messages can be moved to an auxiliary queue for processing at a later
time.

Figure~\ref{fig:lact-sel-recv-trans} shows the translation formally. Except for
function types, the translation on types is homomorphic.
Similar to the translation from $\lact$ into $\lch$, we add an
additional parameter for the save queue.

\begin{figure}[t]
{\footnotesize
  \begin{minipage}[t]{0.4\textwidth}
  \begin{align*}
    &\find{\seq{c}}{\textit{mb}} \defeq \\
    & \begin{aligned}[t]
    & (\rectwo{\textit{findLoop}}{\textit{ms}} \\
    & \quad \letintwo{(mb_1, mb_2)}{ms} \\
    & \quad \caseofone{mb_2} \\
    & \qquad \emptylist{} \mapsto \transloop{\seq{c}}{\textit{mb}_1} \\
    & \qquad \listcons{x}{mb_2'} \mapsto \\
    & \qquad \quad \efflettwo{mb'}{mb_1 \doubleplus mb_2'} \\
    & \qquad \quad \caseofone{x} \metadef{branches}(\seq{c}, \textit{mb}', \\
    & \qquad \qquad
    \begin{aligned}[t]
      & \lambda y . (\efflettwo{\textit{mb}_1'}{\textit{mb}_1 \doubleplus \listterm{y}} \\
      & \qquad \textit{findLoop} \app (\textit{mb}_1', \textit{mb}_2')) ) \})
      \app (\emptylist{}, \textit{mb})
    \end{aligned}
  \end{aligned}
  \end{align*}
\begin{align*}
  & \metadef{label}(\variant{\ell = x} \when M \mapsto N) = \ell \\
  & \metadef{labels}(\seq{c}) = \nodups{[ \metadef{label}(c) \mid c \leftarrow \seq{c} ]} \\
  & \metadef{matching}(\ell, \seq{c}) = [ c \mid (c \leftarrow \seq{c}) \wedge \metadef{label}(c) = \ell ] \\
  & \metadef{unhandled}(\seq{c}) =
  [ \ell \mid (\variant{\ell : A} \leftarrow \variant{\ell_i : A_i}_i) \wedge \ell \not\in \metadef{labels}(\seq{c}) ] \hspace{-20em}
\end{align*}

\end{minipage}
~
\begin{minipage}[t]{0.5\textwidth}
\begin{align*}
  &\metadef{ifPats}(\textit{mb}, \ell, y, \epsilon, \textit{default}) = \textit{default} \app
  \variant{\ell = y} \\
  &\metadef{ifPats}(\textit{mb}, \ell, y, \\
  & \quad (\variant{\ell = x} \when M \mapsto N) \cdot pats,
\textit{default}) = \\
  &\qquad
  \begin{aligned}[t]
    & \efflettwo{\textit{resPair}}{(\seltransterm{M}{\textit{mb}}) \{ y / x \} } \\
    & \letintwo{(\textit{res}, \textit{mb}')}{\textit{resPair}} \\
    & \termconst{if} \: \textit{res} \: \termconst{then} \: (\seltrans{N}{\textit{mb}}) \{ y / x \} \\
    & \termconst{else} \: \metadef{ifPats}(\textit{mb}, \ell, y, \textit{pats}, \textit{default})
  \end{aligned}
\end{align*}
\[
  \begin{aligned}[t]
    & \transloop{\seq{c}}{\textit{mb}} \defeq \\ \quad
    & \quad (\rectwo{\textit{recvLoop}}{\textit{mb}} \\
    & \qquad \efflettwo{x}{\lactrecv} \\
    & \qquad \caseof{x}{\metadef{branches}(\seq{c}, \textit{mb}, \\
    & \quad \qquad \lambda y .
    \begin{aligned}[t]
      & \efflettwo{\textit{mb}'}{\textit{mb} \doubleplus \listterm{y}} \\
      & \textit{recvLoop} \app \textit{mb}')}) \app \textit{mb}
    \end{aligned}
  \end{aligned}
\]
\end{minipage}
\vspace{1em}
\begin{align*}
  & \metadef{branches}(\seq{c}, \textit{mb}, \textit{default}) =
    \metadef{patBranches}(\seq{c}, \textit{mb}, \textit{default}) \cdot
\metadef{defaultBranches}(\seq{c}, \textit{mb}, \textit{default}) \\
  &\metadef{patBranches}(\seq{c}, \textit{mb}, \textit{default}) = \\
  & \quad [ \variant{\ell = x} \mapsto
    \metadef{ifPats}(\textit{mb}, \ell, x, \seq{c_\ell}, \textit{default}) \mid
    (\ell \leftarrow \metadef{labels}(\seq{c})) \: \wedge
    \seq{c_\ell} = \metadef{matching}(\ell, \seq{c}) \wedge x \text{ fresh} ]  \\
    & \metadef{defaultBranches}(\seq{c}, \textit{mb}, \textit{default}) =
    [ \variant{\ell = x} \mapsto \textit{default} \app \variant{\ell = x} \mid (\ell \leftarrow \metadef{unhandled}(\seq{c})) \wedge x \text{ fresh} ] \\
\end{align*}
}
\vspace{-2em}
\caption{Meta level definitions for translation from $\lact$ with selective receive to $\lact$
(wrt.\ a mailbox type $\variant{\ell_i : A_i}_i$)}
\end{figure}

The translation on terms $\seltransterm{M}{\textit{mb}}$ takes a variable
\textit{mb} representing the save queue as its parameter, returning a pair of the
resulting term and the updated save queue. The majority of cases are standard,
except for $\selrecv{\seq{c}}$, which relies on the meta-level definition
\metadef{find}($\seq{c}$, \textit{mb}): $\seq{c}$ is a sequence of clauses,
and \textit{mb} is the save queue. The constituent \textit{findLoop} function takes
a pair of lists $(\textit{mb}_1, \textit{mb}_2)$, where $mb_1$ is the
list of processed values found not to match, and
$\textit{mb}_2$ is the list of values still to be processed.
The loop inspects the list until one either matches, or the end of the list is
reached.
Should no values in the term-level representation of the mailbox match, then the
\metadef{loop} function repeatedly receives from the mailbox, testing each new
message against the patterns.

Note that the $\termconst{case}$ construct in the core $\lact$ calculus is more
restrictive than selective receive: given a variant
$\variant{\ell_i : A_i}_i$, $\termconst{case}$ requires a
single branch for each label. Selective receive allows multiple
branches for each label, each containing a possibly-different
predicate, and does not require pattern matching to be exhaustive.

We therefore need to perform pattern matching elaboration; this is achieved by
the \metadef{branches} meta level definition. We make use of list comprehension notation:
for example, $[c \mid (c \leftarrow \seq{c}) \wedge \metadef{label}(c) = \ell]$ returns
the (ordered) list of clauses in a sequence $\seq{c}$ such that the label of the
receive clause matches a label $\ell$. We assume a meta level function $\metadef{noDups}$
which removes duplicates from a list.
Case branches are computed using the \metadef{branches} meta level
definition: \metadef{patBranches} creates a branch for each label present in
the selective receive, creating (via \metadef{ifPats}) a sequence of
if-then-else statements to check each predicate in turn;
\metadef{defaultBranches} creates a branch for each label that is present in
the mailbox type but not in any selective receive clauses.

\subparagraph{Properties of the translation.}
The translation preserves typing of terms and values.

\begin{lemma}[Translation preserves typing (values and terms)]\label{lem:seltrans-term-typing} \hfill
  \begin{enumerate}
    \item If $\Gamma \vdash V \of A$, then $\seltrans{\Gamma} \vdash
      \seltrans{V} \of \seltrans{A}$.
    \item If $\Gamma \mid \variant{\ell_i \of A_i}_i \vdash M \of B$, then \\ $\seltrans{\Gamma},
      \textit{mb} \of \listty{\variant{\ell_i \of \seltrans{A_i}}_i} \mid
    \variant{\ell_i \of \seltrans{A_i}}_i \vdash \seltrans{M}{\textit{mb}} \of (\seltrans{B} \times
      \listty{\variant{\ell_i \of \seltrans{A_i}}_i})$.
  \end{enumerate}
\end{lemma}

\noindent
Alas, a direct one-to-one translation on configurations is not possible, since a
message in a mailbox in the source language could be either in the mailbox or
the save queue in the target language. Consequently, we translate a
configuration into a set of possible configurations, depending on how many
messages have been processed.
We can show that all configurations in the resulting set are
type-correct, and can simulate the original reduction.

\begin{theorem}[Translation preserves typing]\label{thm:seltrans-term-typing}
  If $\Gamma ; \Delta \vdash \config{C}$, then $\forall \config{D} \in
  \seltrans{\config{C}}$, it is the case that $\seltrans{\Gamma} ; \seltrans{\Delta} \vdash
  \config{D}$.
\end{theorem}

\begin{theorem}[Simulation ($\lact$ with selective receive in $\lact$)]\label{thm:seltrans-simulation}
  If $\Gamma ; \Delta \vdash \config{C}$ and $\config{C} \ceval \config{C'}$,
  then $\forall \config{D} \in \seltrans{\config{C}}$, there exists a
  $\config{D'}$ such that $\config{D} \ceval^+ \config{D'}$ and $\config{D'} \in
  \seltrans{\config{C'}}$.
\end{theorem}

\subparagraph{Remark.}
Originally we expected to need to add an analogous selective receive
construct to $\lch$ in order to be able to translate $\lact$ with
selective receive into $\lch$. We were surprised (in part due to the
complex reduction rule and the native runtime support in Erlang) when
we discovered that selective receive can be emulated in plain
$\lact$. Moreover, we were pleasantly surprised that types pose no
difficulties in the translation.

 \subsection{Choice}\label{sec:extensions-choice}

The calculus $\lch$ supports only blocking receive on a \emph{single}
channel. A more powerful mechanism is \emph{selective communication},
where a value is taken nondeterministically from \emph{two}
channels. An important use case is receiving a value when either
channel could be empty.

Here we have considered only the most basic form of selective choice
over two channels. More generally, it may be extended to arbitrary
regular data types~\cite{paykin:choose}.
As Concurrent ML~\cite{reppy:cml} embraces rendezvous-based
synchronous communication, it provides \emph{generalised selective
  communication} where a process can synchronise on a mixture of input
or output communication events.
Similarly, the join patterns of the join calculus~\cite{FournetG96}
provide a general abstraction for selective communication over
multiple channels.

As we are working in the asynchronous
setting where a $\lchwd{give}$ operation can reduce immediately, we consider
only input-guarded choice.
Input-guarded choice can be added straightforwardly to $\lch$, as
shown in Figure~\ref{fig:choose:defn}. Emulating such a construct
satisfactorily in $\lact$ is nontrivial, because messages must be
multiplexed through a local queue. One approach could be to use the
work of Chaudhuri~\cite{chaudhuri:concurrent-ml-haskell} which shows how to
implement generalised choice using synchronous message passing, but
implementing this in $\lch$ may be difficult due to the asynchrony of
$\lchwd{give}$. We leave a more thorough investigation to future work.

\begin{figure}[t]
{\footnotesize
\begin{mathpar}
  \inferrule
  { \Gamma \vdash V: \chan{A} \\ \Gamma \vdash W : \chan{B} }
  { \Gamma \vdash \lchchoose{V}{W} : A + B }
\end{mathpar}}
{\small
\[
\begin{array}{rcl}
 E[\lchchoose{a}{b}] \parallel \lchbuf{a}{W_1 \cdot
 \seq{V_1}} \parallel \lchbuf{b}{\seq{V_2}} & \ceval &
    E[\effreturn{(\inl{W_1})}] \parallel \lchbuf{a}{\seq{V_1}} \parallel
\lchbuf{b}{\seq{V_2}} \\
 E[\lchchoose{a}{b}] \parallel \lchbuf{a}{\seq{V_1}} \parallel \lchbuf{b}{W_2 \cdot \seq{V_2}} & \ceval &
   E[\effreturn{(\inr{W_2})}] \parallel \lchbuf{a}{\seq{V_1}} \parallel
\lchbuf{b}{\seq{V_2}}
\end{array}
\]}
\caption{Additional typing and evaluation rules for $\lch$ with choice}
\label{fig:choose:defn}
\end{figure}

\subsection{Behavioural types}\label{sec:extensions-behavioural-types}

Behavioural types allow the type of an object (e.g. a channel) to
evolve as a program executes.
A widely studied behavioural typing discipline is that of
\emph{session types}~\cite{honda:dyadic, honda:primitives}, which are channel types
sufficiently expressive to describe \emph{communication protocols} between
participants. For example, the session type for a channel which sends two
integers and receives their sum could be defined as 
  ${!}\textsf{Int}.{!}\textsf{Int}.{?}\textsf{Int}.\textsf{end}$.
Session
types are suited to channels, whereas current work on
session-typed actors concentrates on
runtime monitoring~\cite{neykova:actors}.

A natural question to ask is whether one can combine the benefits of
actors and of session types---indeed, this was one of our original
motivations for wanting to better understand the relationship between
actors and channels in the first place!
A session-typed channel may support both sending and receiving (at
different points in the protocol it encodes), but communication with
another process' mailbox is one-way.
We have studied several variants of $\lact$ with \emph{polarised}
session types~\cite{PfenningG15, LindleyM16} which capture such
one-way communication, but they seem too weak to simulate
session-typed channels.
In future, we would like to find an extension of $\lact$ with
behavioural types that admits a similar simulation result to the ones
in this paper.

 \section{Related work}\label{sec:related-work}

Our formulation of concurrent $\lambda$-calculi is inspired by
$\lambda(\textsf{fut})$~\cite{niehren:lambda-fut}, a concurrent
$\lambda$-calculus with threads, futures, reference cells, and an atomic
exchange construct. In the presence of lists, futures are sufficient to encode
asynchronous channels. In $\lch$, we concentrate on asynchronous channels to
better understand the correspondence with actors.  Channel-based concurrent
$\lambda$-calculi form the basis of functional languages with session
types~\cite{gay:linearsessions, lindley:semantics}.

Concurrent ML~\cite{reppy:cml} extends Standard ML with a rich set of
combinators for synchronous channels, which again can
emulate asynchronous channels. A core notion in Concurrent ML is
nondeterministically synchronising on multiple synchronous events, such as
sending or receiving messages; relating such a construct to an actor calculus
is nontrivial, and remains an open problem.
Hopac~\cite{hopac} is a channel-based concurrency library for F\#, based on
Concurrent ML. The Hopac documentation relates synchronous channels and
actors~\cite{hopac-actors}, implementing actor-style primitives using channels,
and channel-style primitives using actors. The implementation of channels using
actors uses mutable references to emulate the $\lchwd{take}$ function, whereas
our translation achieves this using message passing. Additionally, our
translation is formalised and we prove that the translations are type- and
semantics-preserving.

Links~\cite{cooper:links} provides actor-style concurrency, and the paper
describes a translation into $\lambda(\mkwd{fut})$. Our translation is
semantics-preserving and can be done without synchronisation.

The actor model was designed by Hewitt~\cite{hewitt:actors} and examined in the
context of distributed systems by Agha~\cite{agha:actors}.
Agha et al.~\cite{agha:foundation} describe a functional actor calculus
based on the $\lambda$-calculus augmented by three core constructs:
\textsf{send} sends a message; \textsf{letactor} creates a new
actor; and \textsf{become} changes an actor's behaviour.
The operational semantics is defined in terms of
a global actor mapping, a global multiset of
messages, a set of \emph{receptionists} (actors which are
externally visible to other configurations), and a set of external actor names.
Instead of \textsf{become}, we use an explicit $\lactrecv{}$
construct, which more closely
resembles Erlang (referred to by the authors as ``essentially an actor
language'').
Our concurrent semantics, more in the spirit of process calculi, encodes
visibility via
name restrictions and structural congruences.  The authors consider a
behavioural theory in terms of
operational and testing equivalences---something we have not
investigated. 

Scala has native support for actor-style
concurrency, implemented efficiently without explicit virtual machine
support~\cite{haller:scala-actors}.
The actor model inspires \emph{active objects}~\cite{lavender:active-objects}:
objects supporting asynchronous method calls which return responses using
futures. De Boer et al.~\cite{deboer:complete-future} describe a language for active objects
with cooperatively scheduled threads within each object. Core
ABS~\cite{johnsen:abs} is a specification language based on active objects.
Using futures for synchronisation sidesteps the type pollution problem inherent
in call-response patterns with actors, although our translations work in the
absence of synchronisation. By working in the functional setting, we obtain
more compact calculi.

 \section{Conclusion}\label{sec:conclusion}

Inspired by languages such as Go which take channels as core
constructs for communication, and languages such as Erlang which are
based on the actor model of concurrency, we have presented
translations back and forth between a concurrent $\lambda$-calculus
$\lch$ with channel-based communication constructs and a concurrent
$\lambda$-calculus $\lact$ with actor-based communication constructs.
We have proved that $\lact$ can simulate $\lch$ and vice-versa.

The translation from $\lact$ to $\lch$ is straightforward, whereas the
translation from $\lch$ to $\lact$ requires considerably more effort.
Returning to Figure~\ref{fig:mailboxes-as-pinned-channels}, this is unsurprising!

We have also shown how to extend $\lact{}$ with synchronisation, greatly
simplifying the translation from $\lch$ into $\lact$, and have shown how
Erlang-style selective receive can be emulated in $\lact$. Additionally, we have
discussed input-guarded choice in $\lch{}$, and how behavioural types may fit in
with $\lact{}$.

In future, we firstly plan to strengthen our
operational correspondence results by considering operational completeness.
Secondly, we plan to investigate how to emulate $\lch$ with input-guarded choice
in $\lact$.  Finally, we intend to use the lessons learnt from studying $\lch$
and $\lact$ to inform the design of an actor-inspired language with behavioural
types.

\section*{Acknowledgements}
This work was supported by EPSRC grants EP/L01503X/1 (University of
Edinburgh CDT in Pervasive Parallelism) and EP/K034413/1 (A Basis for
Concurrency and Distribution). Thanks to Philipp Haller, Daniel
Hillerstr\"om, Ian Stark, and the anonymous reviewers for detailed
comments.

{\small
\bibliographystyle{plainurl}
\bibliography{bibliography}
}

\newpage
\appendix
\section{Coalescing Transformation}~\label{appendix:coalescing}

Our translation from $\lch{}$ into $\lact{}$ relies on the assumption that all
channels have the same type, which is rarely the case in practice. Here, we
sketch a sample type-directed transformation which we call \emph{coalescing}, which
transforms an arbitrary $\lch{}$ program into an $\lch{}$ program which has only
one type of channel.

The transformation works by encapsulating each type of message in a variant
type, and ensuring that $\lchwd{give}$ and $\lchwd{take}$ use the correct
variant injections. Although the translation
necessarily loses type information, thus introducing partiality, we can
show that terms and configurations that are the result of the coalescing
transformation never reduce to an error.

We say that a type $A$ is a \emph{base carried type} in a configuration $\Gamma ;
\Delta \vdash \config{C}$ if there exists some subterm $\Gamma \vdash V :
\chan{A}$, where $A$ is \emph{not} of the form $\chan{B}$.

In order to perform the coalescing transformation, we require an environment $\sigma$
which maps each base carried type $A$ to a unique token $\ell$, which we use as an
injection into a variant type.

We write $\sigma \smile \Gamma ; \Delta \vdash
\config{C}$ if $\sigma$ contains a bijective mapping $A \mapsto \ell$ for each base
carried type in $\Gamma ; \Delta \vdash \config{C}$. We extend the relation
analogously to judgements on values and computation terms.

Next, we define the notion of a \emph{coalesced channel type}, which can be used
to ensure that all channels in the system have the same type.

\begin{definition}[Coalesced channel type]\hfill \\
Given a token environment $\sigma = A_0 \mapsto \ell_0, \ldots A_n \mapsto
\ell_n$, we define the \emph{coalesced channel type} $\cct{\sigma}$ as

\[
  \cct{\sigma} = \mu X . \variant{\ell_0 : \coalesceinner{A_0}, \ldots, \ell_n :
  \coalesceinner{A_n}, \ell_c : \chan{X}}
\]

where

\begin{minipage}{0.5\textwidth}
  \centering \[ \begin{array}{rcl}
  \coalesceinner{\one} &=& \one \\ \coalesceinner{A \to B} &=&
  \coalesceinner{A} \to \coalesceinner{B} \\ \coalesceinner{A \times B}
  &=& \coalesceinner{A} \times \coalesceinner{B} \\ \end{array} \]
\end{minipage}
  \begin{minipage}{0.5\textwidth} \[ \begin{array}{rcl} \coalesceinner{A + B}
      &=& \coalesceinner{A} + \coalesceinner{B} \\
      \coalesceinner{\listty{A}} &=& \listty{\coalesceinner{A}} \\
      \coalesceinner{\chan{A}} &=& \chan{X} \\
      \coalesceinner{\mu Y . A} &=& \mu Y . \coalesceinner{A}
  \end{array} \]
    \end{minipage}

which is the single channel type which can receive values of all possible types
sent in the system.
  \end{definition}

Note that the definition of a base carried type excludes the possibility of a
type of the form $\chan{A}$ appearing in $\sigma$.
To handle the case of sending channels, we require $\cct{\sigma}$ to be a
recursive type; a distinguished token $\ell_c$ denotes the variant case for
sending a channel over a channel.

Retrieving a token from the token environment $\sigma$ is defined by the
following inference rules. Note that $\chan{A}$ maps to the distinguished token
$\ell_c$.

\begin{mathpar}
  \inferrule
  { (A \mapsto \ell) \in \sigma }
  { \sigma(A) = \ell }

  \inferrule
  { }
  { \sigma(\chan{A}) = \ell_c }
\end{mathpar}

With the coalesced channel type defined, we can define a function mapping types
to coalesced types.

\begin{definition}[Type coalescing] \hfill \\ \begin{minipage}{0.5\textwidth}
  \centering \[ \begin{array}{rcl}
  \coalesce{\one}{\sigma} &=& \one \\ \coalesce{A \to B}{\sigma} &=&
  \coalesce{A}{\sigma} \to \coalesce{B}{\sigma} \\ \coalesce{A \times B}{\sigma}
  &=& \coalesce{A}{\sigma} \times \coalesce{B}{\sigma} \\ \end{array} \]
\end{minipage}
  \begin{minipage}{0.5\textwidth} \[ \begin{array}{rcl} \coalesce{A + B}{\sigma}
      &=& \coalesce{A}{\sigma} + \coalesce{B}{\sigma} \\
      \coalesce{\listty{A}}{\sigma} &=& \listty{\coalesce{A}{\sigma}} \\
      \coalesce{\chan{A}}{\sigma} &=& \chan{\cct{\sigma}} \\ 
      \coalesce{\mu X . A}{\sigma} &=& \mu X . \coalesce{A}{\sigma} \\
  \end{array} \]
    \end{minipage}
\end{definition}

We then extend $\coalesce{-}{\sigma}$ to typing environments $\Gamma$, taking
into account that we must annotate channel names.

\begin{definition}[Type coalescing on environments]
  \begin{enumerate}
    \item For $\Gamma$:
      \begin{enumerate}
        \item $\coalesce{\emptyset}{\sigma} = \emptyset$
        \item $\coalesce{x : A, \Gamma}{\sigma} = x : \coalesce{A}{\sigma},
          \coalesce{\Gamma}{\sigma}$.
        \item $\coalesce{a : \chan{A}, \Gamma}{\sigma} = a_A :
          \chan{\cct{\sigma}}, \coalesce{\Gamma}{\sigma}$.
      \end{enumerate}
    \item For $\Delta$: 
      \begin{enumerate}
        \item $\coalesce{\emptyset}{\sigma} = \emptyset$
        \item $\coalesce{a: A, \Delta}{\sigma} = a_A : \cct{\sigma}, \coalesce{\Delta}{\sigma}$
      \end{enumerate}
    \end{enumerate}
\end{definition}

\begin{figure}
~Coalescing of channel names \hfill
\framebox{$\jarg{\sigma} \Gamma \vdash V \leadsto V'$}
\begin{mathpar}
  \inferrule
    [Ref]
    { a : \chan{A} \in \Gamma }
    { \jarg{\sigma} \Gamma \vdash a : \chan{A} \leadsto a_A}
\end{mathpar}

~Coalescing of communication and concurrency primitives
\hfill
\framebox{$\jarg{\sigma} \Gamma \vdash M \leadsto M'$}
\begin{mathpar}
\inferrule
  [Give]
  {\jarg{\sigma} \Gamma \vdash V : A \leadsto V' \\\\
   \jarg{\sigma} \Gamma \vdash W : \chan{A} \leadsto W' \\ \sigma(A) = \ell }
  {\jarg{\sigma} \Gamma \vdash \lchgive{V}{W} : \one \leadsto 
  \lchgive{(\roll{\variant{\ell = V'}})}{W'} }

\inferrule
  [Take]
  {\jarg{\sigma} \Gamma \vdash V : \chan{A} \leadsto V' \\ \sigma(A) = \ell_j }
  {\jarg{\sigma} \Gamma \vdash \lchtake{V} : A \leadsto 
   { 
    \begin{aligned}[t]
    & \efflettwo{x}{\lchtake{V'}} \\
    & \efflettwo{y}{\unroll{x}} \\
    & \caseofone{y} \\
    & \quad \variant{\ell_0 = y} \ldots \variant{\ell_{j-1} = y} \mapsto \texttt{error} \\
    & \quad \variant{\ell_j = y} \mapsto y \\
    & \quad \variant{\ell_{j+1} = y} \ldots \variant{\ell_{n} = y} \mapsto \texttt{error}
    & \} \\
    \end{aligned}
   }
  }

\inferrule
  [NewCh]
  { }
  { \jarg{\sigma} \Gamma \vdash \lchnewch : \chan{A} \leadsto \lchnewch{}_A }

\inferrule
  [Fork]
  { \jarg{\sigma} \Gamma  \vdash M : \one \leadsto M' }
  { \jarg{\sigma} \Gamma \vdash \lchfork{M} : \one \leadsto \lchfork{M'} }
\end{mathpar}

~Coalescing of configurations
\hfill
\framebox{$\jarg{\sigma} \Gamma \vdash \config{C} \leadsto \config{C}'$} \\
\begin{mathpar}
\inferrule
  [Par]
  { \jarg{\sigma} \Gamma ; \Delta \vdash \config{C}_1 \leadsto \config{C}'_1 \\ 
    \jarg{\sigma} \Gamma ; \Delta \vdash \config{C}_2 \leadsto \config{C}'_2 }
  { \jarg{\sigma} \Gamma ; \Delta \vdash \config{C}_1 \parallel \config{C}_2 \leadsto \config{C}'_1 \parallel \config{C}'_2}
  
\inferrule
  [Chan]
  { \jarg{\sigma} \Gamma, a : \chan{A} ; \Delta, a : A \vdash \config{C} \leadsto \config{C}' }
  { \jarg{\sigma} \Gamma ; \Delta \vdash (\nu a) \config{C} \leadsto (\nu a_A)\config{C}' }
  
\inferrule
  [Term]
  {\jarg{\sigma} \Gamma \vdash M : A \leadsto M'}
  {\jarg{\sigma} \Gamma ; \cdot \vdash M \leadsto M'}
  
\inferrule
  [Buf]
  {(\jarg{A, \sigma} \Gamma \vdash V_i : A \leadsto V'_i)_i}
  {\jarg{\sigma} \Gamma ; a : A \vdash \lchbuf{a}{\seq{V}} \leadsto \lchbuf{a_A}{\seq{V'}} }
\end{mathpar}

~Coalescing of buffer values
\hfill \framebox{$\{A, \sigma\} \; \Gamma  \vdash V : A \leadsto V'$} \\
\begin{mathpar}
\inferrule
  { \jarg{\sigma} \Gamma  \vdash V : A \leadsto V' \\ \sigma{}(A) = \ell }
  { \jarg{A, \sigma} \Gamma  \vdash V : A \leadsto \variant{\ell = V} }
\end{mathpar}
\caption{Type-directed coalescing pass}\label{fig:lch-coalescing}
\end{figure}

Figure~\ref{fig:lch-coalescing} describes the coalescing pass from $\lch{}$ with
multiple channel types into $\lch{}$ with a single channel type. Judgements
are of the shape $\jarg{\sigma} \Gamma \vdash V : A \leadsto V'$ for values;
$\jarg{\sigma} \Gamma \vdash M : A \leadsto M'$ for computations; and $\jarg{\sigma} \Gamma \vdash
\config{C} \leadsto \config{C'}$ for configurations, where $\sigma$ is an
bijective mapping from types to tokens, and primed values are the results of
the coalescing pass. We omit the rules for values and functional terms, which
are homomorphisms.

Of particular note are the rules for $\lchwd{give}$ and $\lchwd{take}$. The
coalesced version of $\lchwd{give}$ ensures that the correct token is used to
inject into the variant type. The translation of $\lchwd{take}$ retrieves a value
from the channel, and pattern matches to retrieve a value of the correct type
from the variant. As we have less type information, we have to account for the
possibility that pattern matching fails by introducing an error term, which we
define at the top-level of the term:

\[
  \letin{\texttt{error}}{(\rec{f}{x}{f \app x})}{\ldots}
\]

The translation on configurations ensures that all existing values contained
within a buffer are wrappen in the appropriate variant injection.

The coalescing step necessarily loses typing information on channel types. To
aid us in stating an error-freedom result, we annotate channel names $a$ with
their original type; for example, a channel with name $a$ carrying values of
type $A$ would be translated as $a_A$. It is important to note that annotations
are irrelevant to reduction, i.e.:

\[
  E[\lchgive{a_A}{W}] \parallel \lchbuf{a_B}{\seq{V}} \ceval E[\effreturn{()}]
  \parallel \lchbuf{a_B}{\seq{V} \cdot W}
\]

%
%
%
%
%

As previously discussed, the coalescing pass means that channel types are less
specific, with the pass introducing partiality in the form of an error term,
\texttt{error}. However, since we began with a type-safe program in $\lch{}$, we
can show that programs that have been coalesced from well-typed $\lch{}$
configurations never reduce to an error term.

\begin{definition}[Error configuration] \hfill \\
  A configuration $\config{C}$ is an \emph{error configuration} if $\config{C}
  \equiv G[\texttt{error}]$ for some configuration context $G$.
\end{definition}

\begin{definition}[Error-free configuration] \hfill \\
  A configuration $\config{C}$ is \emph{error-free} if it is not an error
  configuration.
\end{definition}

We can straightforwardly see that the initial result of a coalescing pass is
error-free:

\begin{lemma}\label{lem:refinements-error-free}
  If $\sigma \smile \Gamma \vdash \config{C} \leadsto \config{C}'$, then $\config{C}'$ is
  error-free.
\end{lemma}
\begin{qedproof}
  By induction on the derivation of $\Gamma \vdash \config{C} \leadsto
  \config{C'}$.
\end{qedproof}

Next, we show that error-freedom is preserved under reduction. To do so,
we make essential use of the fact that the coalescing pass annotates each
channel with its original type.

\begin{lemma}[Error-freedom (coalesced $\lch{}$)] \hfill \\
  If $\Gamma ; \Delta \vdash \config{C} \leadsto \config{C}'_1$, and
  $\config{C}'_1 \ceval^* \config{C}'_2$, then $\config{C}'_2$ is error free.
\end{lemma}

\begin{qedproof}(Sketch.)
By preservation in $\lch{}$, we have that $\Gamma ; \Delta \vdash \config{C}'_2$,
and by Lemma~\ref{lem:refinements-error-free}, we can assume that
$\config{C}'_1$ is error-free.

We show that an error term can never arise. Suppose that $\config{C}'_2$ was an error
configuration, meaning that $\config{C}'_2 \equiv G[\texttt{error}]$ for some
configuration context $G$. As we have decreed that the \texttt{error} term does
not appear in user programs, we know that \texttt{error} must have arisen from
the refinement pass.
By observation of the refinement rules, we see that \texttt{error} is introduced
only in the refinement rule for \textsc{Take}.

Stepping backwards through the reduction sequence introduced by the
\textsc{Take} rule, we have that:

\[
\begin{aligned}[t]
  & \caseofone{\variant{\ell_k = W}} \\
  & \quad \variant{\ell_0 = y} \ldots \variant{\ell_{j-1} = y} \mapsto \texttt{error} \\
  & \quad \variant{\ell_j = y} \mapsto y \\
  & \quad \variant{\ell_{j+1} = y} \ldots \variant{\ell_{n} = y} \mapsto \texttt{error} \\
  & \}
\end{aligned}
\]

for some $k \ne j$.

Stepping back further, we have that:

\[
\begin{aligned}[t]
  & \efflettwo{x}{\lchtake{a_B}} \\
  & \efflettwo{y}{\unroll{(\roll{x})}} \\
  & \caseofone{y} \\
  & \quad \variant{\ell_0 = y} \ldots \variant{\ell_{j-1} = y} \mapsto \texttt{error} \\
  & \quad \variant{\ell_j = y} \mapsto y \\
  & \quad \variant{\ell_{j+1} = y} \ldots \variant{\ell_{n} = y} \mapsto
  \texttt{error} \\
  & \}
\end{aligned}
\]

Now, inspecting the premises of the refinement rule for \textsc{Take}, we have
that $\Gamma \vdash a_B : \chan{A} \leadsto V'$ and $\sigma(A) = \ell_j$.
Examining the refinement rule for \textsc{Ref}, we have that $\jarg{\sigma}
\Gamma \vdash a : \chan{A} \leadsto a_A$, thus we have that $B = A$.

However, we have that $\sigma(A) = \ell_j$ and $\sigma(A) = \ell_k$
but we know that $k \ne j$, thus leading to a contradiction since $\sigma$ is
bijective.
\end{qedproof}

Since annotations are irrelevant to reduction, it follows that $\config{C'}$ has
identical reduction behaviour with all annotations erased.

\newpage
\section{Selected full proofs}\label{appendix:proofs}

\subsection{$\lch{}$ Preservation}

\begin{fake}{Lemma~\ref{lem:lch-term-progress} (Progress ($\lch$ terms))} \hfill \\
Assume $\Gamma$ is either empty or only contains entries of the form $a_i :
\chan{A_i}$.

\noindent
If $\Gamma \vdash M : A$, then either:
\begin{enumerate}
\item $M = \effreturn{V}$ for some value $V$
\item $M$ can be written $E[M']$, where $M'$ is a communication or concurrency primitive (i.e.\ $\lchgive{V}{W}, \lchtake{V}, \lchfork{M}$, or $\lchnewch{}$)
\item There exists some $M'$ such that $M \teval M'$
\end{enumerate}
\end{fake}

\begin{proof}
By induction on the derivation of $\Gamma \vdash M : A$.
Cases \textsc{Return}, \textsc{Give}, \textsc{Take}, \textsc{NewCh}, and
\textsc{Fork} hold immediately due to 1) and 2).
For \textsc{App}, by inversion on the typing relation we have that $V = (\lambda
x . M)$, and thus can $\beta$-reduce.
For \textsc{EffLet}, we have that $\Gamma \vdash
\efflet{x}{M}{N} : B$, with $\Gamma \vdash M : A$ and $\Gamma, x : A \vdash N :
B$. We proceed using the inductive hypothesis. If $M =
\effreturn{V}$, we can take a step $\efflet{x}{\effreturn{V}}{N} \teval
N \{ V / x\}$. If $M$ is a communication term, then we note that
$\efflet{x}{E}{N}$ is an evaluation context and can write $E[\lchfork{N'}]$
etc.\ to satisfy 2). Otherwise, $M$ can take a step and we can take a step by the
lifting rule $E[M] \teval E[M']$.
\end{proof}

\begin{fake}{Lemma~\ref{lem:lch-equiv-typeability}}
  If $\Gamma ; \Delta \vdash \config{C}$ and $\config{C} \equiv \config{D}$,
  then $\Gamma ; \Delta \vdash \config{D}$.
\end{fake}
\begin{proof}
  By induction on the derivation of $\Gamma ; \Delta \vdash \config{C}$.

  \textbf{Case } $\config{C} \parallel \config{D} \equiv \config{D} \parallel \config{C}$.
  \hfill \\
  We assume that $\Gamma ; \Delta \vdash \config{C} \parallel \config{D}$,
  so we have that  $\Delta$ partitions as $\Delta_1, \Delta_2$ such that
  $\Gamma ; \Delta_1 \vdash \config{C}$ and $\Gamma, \Delta_2 \vdash
  \config{D}$. It follows immediately by
  the symmetry of \textsc{Par} that $\Gamma ; \Delta_2, \Delta_1 \vdash
  \config{D} \parallel \config{C}$ as required.

  \textbf{Case } $\config{C} \parallel (\config{D} \parallel \config{E})$
  \hfill \\
  We assume that $\Gamma ; \Delta \vdash \config{C} \parallel (\config{D} \parallel \config{E})$.
  We can show that $\Delta$ partitions as $\Delta_1, \Delta_2$ such that
  $\Gamma ; \Delta_1 \vdash \config{C}$ and
  $\Gamma ; \Delta_2 \vdash \config{D} \parallel \config{E}$.
  We can then show that $\Delta_2$ splits as $\Delta_3, \Delta_4$ such
  that $\Gamma ; \Delta_3 \vdash \config{D}$ and $\Gamma ; \Delta_4 \vdash \config{E}$.
  It is then possible to show that $\Gamma; \Delta_1, \Delta_3 \vdash \config{C} \parallel \config{D}$
  and $\Gamma ; \Delta_4 \vdash \config{E}$. Recomposing, we have that
  $\Gamma ; \Delta \vdash (\config{C} \parallel \config{D}) \parallel \config{E}$ as required.

  \textbf{Case } $\config{C} \parallel (\nu a) \config{D} \equiv
  (\nu a)(\config{C} \parallel \config{D})$ if $a \not\in \fv{(\config{C})}$
  \hfill \\
  We assume that $\Gamma ; \Delta \vdash \config{C} \parallel (\nu a)\config{D}$.
  From this, we know that $\Delta$ splits as $\Delta_1, \Delta_2$ such that
  $\Gamma ; \Delta_1 \vdash \config{C}$ and $\Gamma ; \Delta_2 \vdash (\nu a) \config{D}$.
  By \textsc{Chan}, we know that $\Gamma, a : \chan{A} ; \Delta_2, a : A \vdash \config{D}$.
  By weakening, we have that $\Gamma, a : \chan{A} \vdash \config{C}$;
  we can show that $\Gamma, a : \chan{A} ; \Delta_1, \Delta_2, a : A \vdash \config{C} \parallel \config{D}$.
  By \textsc{Chan}, we have that $\Gamma ; \Delta \vdash (\nu a) (\config{C}
\parallel \config{D})$ as required.

  \textbf{Case } $G[\config{C}] \equiv G[\config{D}]$ if $\config{C} \equiv \config{D}$.
  Immediate by the inductive hypothesis.
\end{proof}

\begin{lemma}[Replacement]\label{lem:lch-replacement}
  If $\Gamma \vdash E[M] : A$, $\Gamma \vdash M : B$, and $\Gamma \vdash N : B$,
  then $\Gamma \vdash E[N] : A$.
\end{lemma}
\begin{qedproof}
  By induction on the structure of $E$.
\end{qedproof}

\begin{fake}{Theorem~\ref{thm:lch-preservation} (Preservation ($\lch{}$ configurations)} \hfill \\
If $\Gamma ; \Delta \vdash \config{C}_1$ and $\config{C}_1 \ceval \config{C}_2$,
then $\Gamma ; \Delta \vdash \config{C}_2$.
\end{fake}

\begin{qedproof}
By induction on the derivation of $\config{C} \ceval \config{C'}$. We use
Lemma~\ref{lem:lch-replacement} implicitly throughout.

\noindent
\textbf{Case \textsc{Give}}\hfill \\
From the assumption $\Gamma ; \Delta \vdash \Ex[\lchgive{W}{a}] \parallel
\lchbuf{a}{\seq{V}}$, we have that $\Gamma ; \cdot \vdash \Ex[\lchgive{W}{a}]$
and $\Gamma ; a : A \vdash \lchbuf{a}{\seq{V}}$. Consequently, we know that
$\Delta = a : A$.

From this, we know that $\Gamma \vdash \lchgive{W}{a} : \one$ and thus $\Gamma
\vdash W : A$ and $\Gamma \vdash a : \chan{A}$.
We also have that $\Gamma ; a : A \vdash \lchbuf{a}{\seq{V}}$, thus $\Gamma
\vdash V_i : A$ for all $V_i \in \seq{V}$.
By \textsc{Unit} we can show $\Gamma ; \cdot \vdash \Ex[\effreturn{()}]$ and by \textsc{Buf}
we can show $\Gamma ; a : A \vdash \lchbuf{a}{\seq{V} \cdot W}$; recomposing, we
arrive at $\Gamma ; \Delta \vdash \Ex[\effreturn{()}] \parallel
\lchbuf{a}{\seq{V}{W}}$ as required.

\noindent
\textbf{Case \textsc{Take}}\hfill \\
From the assumption $\Gamma ; \Delta \vdash \Ex[\lchtake{a}] \parallel
\lchbuf{a}{W \cdot \seq{V}}$, we have that $\Gamma ; \cdot \vdash
\Ex[\lchtake{a}]$ and that $\Gamma ; a : A \vdash \lchbuf{a}{W \cdot \seq{V}}$.
Consequently, we know that $\Delta = a : A$.

From this, we know that
$\Gamma \vdash \lchtake{a} : A$, and thus $\Gamma \vdash a : \chan{A}$.
Similarly, we have that $\Gamma ; a : \chan{A} \vdash \lchbuf{a}{W \cdot \seq{V}}$, and thus
$\Gamma \vdash W : A$.

Consequently, we can show that $\Gamma ; \cdot \vdash E[\effreturn{W}]$ and
$\Gamma ; a : A \vdash \lchbuf{a}{\seq{V}}$; recomposing, we
arrive at $\Gamma ; \Delta \vdash E[\effreturn{W}] \parallel \lchbuf{a}{\seq{V}}$ as required.

\noindent
\textbf{Case \textsc{NewCh}}\hfill \\
By \textsc{Buf} we can type $\Gamma ; a : A \vdash \lchbuf{a}{\epsilon}$,
and since $\Gamma \vdash \lchnewch{} : \chan{A}$, it
is also possible to show $\Gamma, a : \chan{A} \vdash a : \chan{A}$,
thus $\Gamma, a : \chan{A} \vdash E[\effreturn{a}]$.

Recomposing by \textsc{Par} we
have $\Gamma, a : \chan{A} ; a : A \vdash E[\effreturn{a}] \parallel \lchbuf{a}{\epsilon}$,
and by \textsc{Chan} we have $\Gamma ; \cdot \vdash (\nu a) (E[\effreturn{a}] \parallel
\lchbuf{a}{\epsilon})$ as required.

\noindent
\textbf{Case \textsc{Fork}}\hfill \\
From the assumption $\Gamma ; \Delta \vdash E[\lchfork{M}]$, we have that
$\Delta = \emptyset$ and $\Gamma \vdash M : \one$.
By \textsc{Unit} we can show $\Gamma ; \cdot \vdash E[\effreturn{()}]$, and by \textsc{Term}
we can show $\Gamma ; \cdot \vdash M$.  Recomposing, we arrive at $\Gamma ;
\Delta \vdash E[\effreturn{()}] \parallel M$ as required.

\noindent
\textbf{Case \textsc{Lift} }
Immediate by the inductive hypothesis.

\noindent
\textbf{Case \textsc{LiftV} }
Immediate by Lemma~\ref{lem:lch-term-pres}.
\end{qedproof}

\newpage
\subsection{$\lact{}$ Preservation}

\begin{lemma}[Replacement]\label{lem:lact-replacement}
  If $\Gamma \mid C \vdash E[M] : A$, $\Gamma \mid C \vdash M : B$, and $\Gamma
  \mid C \vdash N : B$, then $\Gamma \mid C \vdash E[N] : B$.
\end{lemma}

\begin{qedproof}
  By induction on the structure of E.
\end{qedproof}

\begin{fake}{Theorem~\ref{thm:lact-preservation} (Preservation ($\lact{}$ configurations))} \hfill \\
If $\Gamma ; \Delta \vdash \config{C}_1$ and $\config{C}_1 \ceval \config{C}_2$,
then $\Gamma ; \Delta \vdash \config{C}_2$.
\end{fake}

\begin{qedproof}
By induction on the derivation of $\config{C}_1 \ceval \config{C}_2$, making
implicit use of Lemma~\ref{lem:lact-replacement}.

\noindent
\textbf{Case \textsc{Spawn}} \hfill \\
From the assumption that $\Gamma ; \Delta \vdash \actor{a}{E[\lactspawn{M}]}{\seq{V}}$, we have that
$\Delta = a : C$, that $\Gamma \mid C \vdash \lactspawn{M} : \pid{A}$ and
$\Gamma \mid A \vdash M : \one$.

We can show $\Gamma, b : \pid{A} \vdash b : \pid{A}$; and
therefore that $\Gamma, b : \pid{A} \vdash E[\effreturn{b}] : \one$. By \textsc{Actor},
it follows that $\Gamma, b : \pid{A} ; a : C \vdash \actor{a}{E[\effreturn{b}]}{\seq{V}}$.

By \textsc{Actor}, we can show $\Gamma, b : \pid{A} ; b : A \vdash \actor{b}{M}{\epsilon}$.

Finally, by \textsc{Pid} and \textsc{Par}, we have $\Gamma ; \Delta \vdash (\nu
b) (\actor{a}{E[\effreturn{b}]}{\seq{V}} \parallel \actor{b}{M}{\epsilon})$ as required.

\noindent
\textbf{Case \textsc{Send}} \hfill \\
From the assumption that $\Gamma ; \Delta \vdash
\actor{a}{E[\lactsend{V'}{b}]}{\seq{V}} \parallel \actor{b}{M}{\seq{W}}$, we
have that $\Gamma ; a : A \vdash \actor{a}{E[\lactsend{V'}{b}]}{\seq{V}}$
and
$\Gamma ; b : C \vdash \actor{b}{M}{\seq{W}}$. Consequently, we can write
$\Delta = a : A, b : C$.
From this, we know that $\Gamma \mid A \vdash \lactsend{V'}{b} : \one$, so
we can write $\Gamma = \Gamma', b : \pid{C}$, and $\Gamma
 \vdash V' : C$.
Additionally, we know that $\Gamma ; b : C \vdash \actor{b}{M}{\seq{W}}$ and
thus that $(\Gamma \vdash W_i : C)$ for each entry $W_i \in \seq{W}$.

As $\Gamma \vdash V' : C$, it follows that $\Gamma \vdash
\seq{W} \cdot V'$ and therefore that $\Gamma ; b : C \vdash \actor{b}{M}{\seq{W}
\cdot V}$.
We can also show that $\Gamma \mid C \vdash \effreturn{()} : \one$, and
therefore it follows that $\Gamma ; a : A \vdash \actor{a}{E[\effreturn{()}]}{\seq{V}}$.

Recomposing, we have that $\Gamma ; \Delta \vdash
\actor{a}{E[\effreturn{()}]}{\seq{W}} \parallel \actor{b}{M}{\seq{W} \cdot V'}$
as required.

\noindent
\textbf{Case \textsc{Receive}} \hfill \\
By \textsc{Actor}, we have that $\Gamma; \Delta \vdash
\actor{a}{E[\lactrecv{}]}{W \cdot \seq{V}}$. From this, we know that $\Gamma
\mid A \vdash E[\lactrecv{}] : \one$ (and thus $\Gamma \mid A \vdash
\lactrecv : A$) and $\Gamma \vdash W : A$.

Consequently, we can show $\Gamma \vdash E[\effreturn{W}] : \one$. By
\textsc{Actor}, we arrive at $\Gamma ; \Delta \vdash \actor{a}{E[\effreturn{W}]}{\seq{V}}$ as required.

\noindent
\textbf{Case \textsc{Self}} \hfill \\
By \textsc{Actor}, we have that $\Gamma ; \Delta \vdash
\actor{a}{E[\lactself{}]}{\seq{V}}$, and thus that $\Gamma \mid A \vdash
E[\lactself] : \one$ and $\Gamma \mid A \vdash \lactself{} : \pid{A}$. We also know
that $\Gamma = \Gamma', a : \pid{A}$, and that $\Delta = a : A$.

Trivially, we can show $\Gamma', a : \pid{A} \vdash a : \pid{A}$. Thus it follows that
$\Gamma', a : \pid{A} \vdash E[\effreturn{a}] : \one$ and thus it follows that $\Gamma;
\Delta \vdash \actor{a}{E[\effreturn{a}]}{\seq{V}}$ as required.

\noindent
\textbf{Case \textsc{Lift} } Immediate by the inductive hypothesis.

\noindent
\textbf{Case \textsc{LiftV} } Immediate by Lemma~\ref{lem:lact-term-pres}.
\end{qedproof}

Proofs for the progress results for $\lact$ follow the same structure as their
$\lch$ counterparts.

\newpage
\subsection{Translation: $\lact{}$ into $\lch{}$}

\begin{fake}{Lemma~\ref{lem:lch-lact-term-typing}}
\begin{enumerate}
\item If $\Gamma \vdash V : A$ in $\lact{}$, then $\translateactty{\Gamma}
  \vdash \translateactval{V} : \translateactty{A}$ in $\lch{}$.
\item If $\Gamma \mid B \vdash M : A$ in $\lact{}$, then
  $\translateactty{\Gamma}, \alpha : \chan{\translateactty{B}} \vdash
  \translateactterm{M}{\alpha} : \translateactty{A}$ in $\lch{}$.
\end{enumerate}
\end{fake}
\begin{qedproof}
By simultaneous induction on the derivations of $\Gamma \vdash V: A$ and
$\Gamma \mid B \vdash M : A$.

\noindent
\textbf{Premise 1} \hfill \\ 

\noindent
\textbf{Case $\Gamma \vdash \alpha : A$}
\hfill \\ By the definition of $\translateactval{-}$, we have that
$\translateactval{\alpha} =
\alpha$. By the definition of $\translateactty{\Gamma}$, we have that $\alpha : \translateactty{A}
\in \translateactty{\Gamma}$. Consequently, it follows that $\translateactty{\Gamma}
\vdash \alpha : \translateactty{A}$.

\noindent
\textbf{Case $\Gamma \vdash \lambda x . M : A \lactto{} B$} \hfill \\
From the assumption that $\Gamma \vdash \lambda x . M : A \lactto{C} B$, we
have that $\Gamma, x : A \mid C \vdash M : B$. By the inductive hypothesis
(premise 2), we have that $\translateactty{\Gamma}, x : \translateactty{A},
\textit{ch} : \chan{\translateactty{C}} \vdash \translateactterm{M}{\textit{ch}}
: \translateactty{B}$.

By two applications of \textsc{Abs}, we have $\translateactty{\Gamma} \vdash
\lambda x . \lambda \textit{ch} . \translateactterm{M}{\textit{ch}} :
\translateactty{A} \to \chan{\translateactty{C}} \to \translateactty{B}$ as
required.

\noindent
\textbf{Case $\Gamma \vdash () : \one$} Immediate.

\noindent
\textbf{Case $\Gamma \vdash (V, W) : (A \times B)$} \hfill \\ From the
assumption that $\Gamma \vdash (V, W) : (A, B)$ we have that $\Gamma
\vdash V : A$ and $\Gamma \vdash W : B$. By the inductive
hypothesis (premise 1) and \textsc{Pair}, we can show $\translateactty{\Gamma}
\vdash (\translateactval{V}, \translateactval{W}) : (\translateactty{A} \times
\translateactty{B})$ as required.

\noindent
\textbf{Premise 2} \hfill \\
\noindent
\textbf{Case $\Gamma \mid C \vdash V \app W : B$} \hfill \\
From the assumption that $\Gamma \mid C \vdash V \app W
: B$, we have that $\Gamma \vdash V : A \lactto{C} B$ and $\Gamma
\vdash W : B$. By the inductive hypothesis (premise 1), we have that
$\translateactty{\Gamma}  \vdash \translateactval{V} : \translateactty{A}
\to \chan{\translateactty{C}} \to \translateactty{B}$, and
$\translateactty{\Gamma} \vdash \translateactval{W} : \translateactty{B}$.

By extending the context $\Gamma$ with a $\textit{ch} : \chan{\translateactty{C}}$, we can
show that $\translateactty{\Gamma}, \textit{ch}: \chan{\translateactty{C}}
\vdash \translateactval{V} \app \translateactval{W} \app \textit{ch} : \translateactty{B}$ as required.

\noindent
\textbf{Case $\Gamma \mid C \vdash \efflet{x}{M}{N} : B$} \hfill \\
From the assumption that $\Gamma  \mid C \vdash \efflet{x}{M}{N} : B$, we
have that $\Gamma  \mid C \vdash M : A$ and that $\Gamma, x : A \mid C \vdash N : B$.

By the inductive hypothesis (premise 2), we have that  $\translateactty{\Gamma},
\textit{ch} :
\chan{\translateactty{C}} \vdash \translateactterm{M}{\textit{ch}} :
\translateactty{A}$ and $\translateactty{\Gamma}, x : \translateactty{A}, \textit{ch} :
\chan{\translateactty{C}} \vdash \translateactterm{N}{\textit{ch}} :
\translateactty{B}$.

By \textsc{EffLet}, it follows that $\translateactty{\Gamma},
\textit{ch} :
\chan{\translateactty{C}}  \vdash
\efflet{x}{\translateactterm{M}{ch}}{\translateactterm{N}{ch}} : \translateactty{B}$ as
required.

\noindent
\textbf{Case $\Gamma \mid C \vdash \effreturn{V} : A$} \hfill \\
From the
assumption that $\Gamma \mid C \vdash \effreturn{V} : A$, we have that
$\Gamma \vdash V : A$.

By the inductive hypothesis (premise 1), we have that $\translateactty{\Gamma}
\vdash \translateactval{V} : \translateactty{A}$.

By weakening (as we do not use the mailbox channel), we can show that
$\translateactty{\Gamma}, y : \chan{\translateactty{C}} \vdash
\effreturn{\translateactty{V}} : \translateactty{A}$ as required.

\noindent
\textbf{Case $\Gamma \mid C \vdash \lactsend{V}{W} : \one$} \hfill \\
From the assumption that $\Gamma \mid C \vdash \lactsend{V}{W} : \one$,
we have that $\Gamma \vdash V : A$ and $\Gamma \vdash W :
\chan{A}$.

By the inductive hypothesis (premise 1) we have that $\translateactty{\Gamma}
\vdash \translateactval{V} : \translateactty{A}$ and $\translateactty{\Gamma}
\vdash \translateactval{W} : \chan{\translateactty{A}}$.

By \textsc{Give}, we can show that $\translateactty{\Gamma} \vdash
\lchgive{\translateactval{V}}{\translateactval{W}} : \one$, and by weakening we have that
$\translateactty{\Gamma}, \textit{ch} : \chan{\translateactty{C}} \vdash
\lchgive{\translateactval{V}}{\translateactval{W}} : \one$ as required.

\noindent
\textbf{Case $\Gamma \mid C \vdash \lactrecv{} : C$} \hfill \\ Given a
$\textit{ch} : \chan{\translateactty{C}}$, we can show that
$\translateactty{\Gamma}, \textit{ch}:
\chan{\translateactty{C}} \vdash \textit{ch} : \chan{\translateactty{C}}$ and
therefore that $\translateactty{\Gamma}, \textit{ch}: \chan{\translateactty{C}}
\vdash \lchtake{\textit{ch}} : \translateactty{C}$ as required.

\noindent
\textbf{Case $\Gamma \mid C \vdash \lactspawn{M} : \pid{A}$} \hfill \\
From the assumption that $\Gamma \mid C \vdash \lactspawn{M} : \pid{A}$, we have
that $\Gamma \mid A \vdash M : A$.

By the inductive hypothesis (premise 2), we have that $\translateactty{\Gamma},
\textit{chMb} : \chan{\translateactty{A}} \vdash \translateactterm{M}{chMb} : \translateactty{A}$. By
\textsc{Fork} and \textsc{Return}, we can show that $\translateactty{\Gamma}, chMb :
\chan{\translateactty{A}} \vdash
\letunit{\lchfork{(\translateactterm{M}{chMb})}}{\effreturn{chMb}} : \chan{\translateactty{A}}$.

By \textsc{NewCh} and \textsc{EffLet}, we can show that
$\translateactty{\Gamma} \vdash
\efflet{chMb}{\lchnewch{}}{\letunit{\lchfork{(\translateactterm{M}{chMb})}}{\effreturn{chMb}}} :
\chan{\translateactty{A}}$.

Finally, by weakening, we have that
\[ \translateactty{\Gamma}, \textit{ch} : \chan{\translateactty{C}} \vdash
  \efflet{\textit{chMb}}{\lchnewch{}}{\letunit{(\lchfork{\translateactterm{M}{\textit{chMb}}})}{\effreturn{\textit{chMb}}}} :
\chan{\translateactty{A}} \]

as required.

\noindent
\textbf{Case $\Gamma \mid C \vdash \lactself{} : \pid{C}$} \hfill \\
Given a $ch : \chan{\translateactty{C}}$, we can show that $\translateactty{\Gamma}, ch :
\chan{\translateactty{C}} \vdash \effreturn{\textit{ch}} : \chan{\translateactty{C}}$ as
required.

\end{qedproof}

\begin{fake}{Theorem~\ref{thm:lact-lch-config-typing}}
  If $\Gamma ; \Delta \vdash \config{C}$, then $\translateactty{\Gamma} ;
\translateactty{\Delta} \vdash \translateactconfig{C}$.  \end{fake}

\begin{qedproof}
By induction on the derivation of $\Gamma ; \Delta \vdash \config{C}$.

\noindent
\textbf{Case Par} \hfill \\ From the assumption that $\Gamma ; \Delta \vdash
\config{C}_1 \parallel \config{C}_2$, we have that
$\Delta$ splits as $\Delta_1, \Delta_2$ such that
$\Gamma ; \Delta_1 \vdash
\config{C}_1$ and $\Gamma ; \Delta_2 \vdash \config{C}_2$. By the inductive
hypothesis, we have that $\translateactty{\Gamma} ; \translateactty{\Delta_1} \vdash
\translateactconfig{\config{C}_1}$ and $\translateactty{\Gamma} ; \translateactty{\Delta_2} \vdash
\translateactconfig{\config{C}_2}$. Recomposing by \textsc{Par}, we have that
$\translateactty{\Gamma} ; \translateactty{\Delta_1}, \translateactty{\Delta_2} \vdash \translateactconfig{\config{C}_1}
\parallel \translateactconfig{\config{C}_2}$ as required.

\noindent
\textbf{Case Pid} \hfill \\ From the assumption that $\Gamma ; \Delta \vdash (\nu
a) \config{C}$, we have that $\Gamma, a : \pid{A} ; \Delta, a : A \vdash \config{C}$.
By the inductive hypothesis, we have that $\translateactty{\Gamma},
a : \chan{\translateactty{A}} ;
\translateactty{\Delta}, a : \translateactty{A} \vdash \translateactconfig{\config{C}}$.
Recomposing by \textsc{Pid}, we have that $\translateactty{\Gamma} ; \translateactty{\Delta} \vdash (\nu
  a) \translateactconfig{\config{C}}$ as required.

\noindent
\textbf{Case Actor} \hfill \\ 
From the assumption that $\Gamma, a : \pid{A} ; a : A \vdash
\actor{a}{M}{\seq{V}}$, we have that $\Gamma, a : \pid{A} \mid A \vdash M : \one$. By
Lemma~\ref{lem:lact-lch-term-typing}, we have that $\translateactty{\Gamma}, a:
\chan{\translateactty{A}} \vdash \translateactterm{M}{a} : \one$.
It follows straightforwardly that $\translateactty{\Gamma}, a :
\chan{\translateactty{A}} ; \cdot
\vdash \translateactterm{M}{a}$.

We can also show that $\translateactty{\Gamma}, a : \chan{\translateactty{A}} ;
a : \translateactty{A} \vdash \lchbuf{a}{\translateactval{\seq{V}}}$ (where
$\translateactval{\seq{V}} = \translateactval{V_1} \cdot \ldots \cdot
\translateactval{V_n}$), by repeated applications of
Lemma~\ref{lem:lact-lch-term-typing}.

By \textsc{Term} and \textsc{Par}, we have that $\translateactty{\Gamma}, a :
\chan{\translateactty{A}} ; a : \translateactty{A} \vdash
\lchbuf{a}{\translateactval{\seq{V}}} \parallel \translateactterm{M}{a}$ as
required.

\end{qedproof}

\begin{fake}{Lemma~\ref{lem:lact-lch-struct}}
If $\Gamma ; \Delta \vdash \config{C}$ and $\config{C} \equiv \config{D}$, then
$\translateactconfig{\config{C}} \equiv \translateactconfig{\config{D}}$.
\end{fake}
\begin{proof}
  By induction on the derivation of $\Gamma ; \Delta \vdash \config{C}$,
  examining each possible equivalence. The result follows easily from
  the homomorphic nature of the translation on name restrictions and
  parallel composition.
\end{proof}

\begin{fake}{Theorem~\ref{thm:lact-lch-config-sim}} \hfill \\
  If $\Gamma \vdash \config{C}_1$ and $\config{C}_1 \ceval \config{C}_2$, then
  there exists some $\config{D}$ such that
  $\translateactconfig{\config{C}_1}\ceval^{*} \config{D}$, with $\config{D} \equiv
  \translateactconfig{\config{C}_2}$.
\end{fake}

\begin{qedproof}
By induction on the derivation of $\config{C} \ceval \config{C'}$.

\begin{proofcase}{Spawn}
Assumption & $\translateactconfig{\actor{a}{E[\lactspawn{M}]}{\seq{V}}}$ & (1) \\
Definition of $\translateactconfig{-}$ & $\lchbuf{a}{\translateactval{\seq{V}}} \parallel
  (\translateacttermzero{E}[\efflet{c}{\lchnewch}{\lchfork{(\translateactterm{M}{c})};
  \effreturn{c}}] \app a) $ & (2) \\
$\lchnewch{}$ reduction & $\lchbuf{a}{\translateactval{\seq{V}}} \parallel
      (\nu
        b)((\translateacttermzero{E}[\efflet{c}{\effreturn{b}}{\lchfork{(\translateactterm{M}{c})};
          \effreturn{c}}]
\app a) \parallel \lchbuf{b}{\epsilon}) $ & (3)\\
Let reduction & $\lchbuf{a}{\translateactval{\seq{V}}} \parallel
          (\nu b)((\translateacttermzero{E}[\lchfork{(\translateactterm{M}{b})};
          \effreturn{b}] \app a )  \parallel \lchbuf{b}{\epsilon}) $ & (4)\\
Fork reduction & $\lchbuf{a}{\translateactval{\seq{V}}} \parallel
                (\nu b)((\translateacttermzero{E}[\effreturn{()}; \effreturn{b}] \app a) \parallel
          (\translateactterm{M}{b}) \parallel \lchbuf{b}{\epsilon}) $ & (5)\\
Let reduction & $\lchbuf{a}{\translateactval{\seq{V}}} \parallel
            (\nu b)((\translateacttermzero{E}[\effreturn{b}] \app a) \parallel (\translateactterm{M}{b}) \parallel \lchbuf{b}{\epsilon}) $ & (6)\\
$\equiv$ & $(\nu b)(\lchbuf{a}{\translateactval{\seq{V}}} \parallel
          (\translateacttermzero{E}[\effreturn{b}] \app a) \parallel \lchbuf{b}{\epsilon} \parallel (\translateactterm{M}{b}) ) $ & (7)\\
              = & $\translateactconfig{(\nu b)(\actor{a}{E[\effreturn{b}]}{\seq{V}} \parallel \actor{b}{M}{\epsilon})}$ & (8) \\
\end{proofcase}

\begin{proofcase}{Self}
Assumption & $\translateactconfig{\actor{a}{E[\lactself{}]}{\seq{V}}}$ & (1) \\
Definition of $\translateactconfig{-}$ &
  $\lchbuf{a}{(\translateactval{\seq{V}})} \parallel
  (\translateacttermzero{E}[\effreturn{a}] \app a) $ & (2) \\
  = & $\translateactconfig{\actor{a}{E[\effreturn{a}]}{\seq{V}}} $ &  \\
\end{proofcase}

\begin{proofcase}{Send}
Assumption & $\translateactconfig{\actor{a}{E[\lactsend{V'}{b}]}{\seq{V}} \parallel \actor{b}{M}{\seq{W}} }$ & (1) \\
Definition of $\translateactconfig{-}$ & $\lchbuf{a}{\translateactval{\seq{V}}} \parallel 
  (\translateacttermzero{E}[\lchgive{\translateactval{V'}}{b}] \app a)
  \parallel \lchbuf{b}{\translateactval{\seq{W}}} \parallel
  (\translateactterm{M}{b})  $ & (2) \\
  Give reduction & $\lchbuf{a}{\translateactval{\seq{V}}} \parallel
  (\translateacttermzero{E}[\effreturn{()}] \app a)
  \parallel \lchbuf{b}{\translateactval{\seq{W}} \cdot
  \translateactval{V'}} \parallel (\translateactterm{M}{b})  $ & (3) \\
  = & $\translateactconfig{\actor{a}{E[\effreturn{()}]}{\seq{V}} \parallel \actor{b}{M}{\seq{W} \cdot V'}}$ & (4) \\
\end{proofcase}

\begin{proofcase}{Receive }
Assumption & $\translateactconfig{\actor{a}{E[\lactrecv{}]}{W \cdot \seq{V}}}$ & (1) \\
Definition of $\translateactconfig{-}$ &
  $ \lchbuf{a}{\translateactval{W} \cdot \translateactval{\seq{V}}} \parallel
  (\translateacttermzero{E}[\lchtake{a}] \app a)$ & (2) \\
Take reduction  & $ \lchbuf{a}{\translateactval{\seq{V}}} \parallel
  (\translateacttermzero{E}[\effreturn{\translateactval{W}}] \app a)$ & (3) \\
  = & $ \translateactconfig{\actor{a}{E[\effreturn{W}]}{\seq{V}}} $ & (4)\\
\end{proofcase}

Lift is immediate from the inductive hypothesis, and LiftV is immediate from
Lemma~\ref{lem:lact-lch-term-sim}.
\end{qedproof}

\newpage
\subsection{Translation: $\lch$ into $\lact{}$}

\begin{fake}{Lemma~\ref{lem:lch-lact-term-typing}}
  \begin{enumerate}
    \item If $\jarg{B} \Gamma \vdash V : A$, then
      $\translatechty{\Gamma} \vdash \translatechval{V} :
      \translatechty{A}$.
    \item If $\jarg{B} \Gamma  \vdash M : A$, then
      $\translatechty{\Gamma} \mid \translatechty{B} \vdash \translatechterm{M} : \translatechty{A}$.
  \end{enumerate}
\end{fake}
\begin{qedproof}
  By simultaneous induction on the derivations of $\jarg{B}  \Gamma \vdash V :
  A$ and $\jarg{B} \Gamma \vdash M : A$.

  \noindent
  \textbf{Premise 1} \hfill \\

  \noindent
  \textbf{Case \textsc{Var}} \hfill \\
  From the assumption that $\jarg{B} \Gamma \vdash \alpha : A$, we know that $\alpha :
  A \in \Gamma$. By the definition of $\translatechty{\Gamma}$, we have that
  $\alpha : \translatechty{A} \in \translatechty{\Gamma}$. Since
  $\translatechval{\alpha} = \alpha$, it
  follows that $\translatechty{\Gamma} \vdash \alpha :
  \translatechty{A}$ as required.

  \noindent
  \textbf{Case \textsc{Abs}} \hfill \\
  From the assumption that $\jarg{C} \Gamma \vdash \lambda x . M : A \to B$, we have
  that $\jarg{C} \Gamma, x : A \vdash M : B$. By the inductive hypothesis (Premise 2), we
  have that $\translatechty{\Gamma}, x : \translatechty{A} \mid
  \translatechty{C} \vdash
  \translatechterm{M} : \translatechty{A}$. By \textsc{Abs}, we can show
  that $\translatechty{\Gamma} \mid \translatechty{C} \vdash \lambda x . \translatechterm{M} :
  \translatechty{A} \lactto{\translatechty{C}} \translatechty{B}$ as
  required.

  \noindent
  \textbf{Case \textsc{Rec}} Similar to \textsc{Abs}.

  \noindent
  \textbf{Case \textsc{Unit}}  Immediate.

  \noindent
  \textbf{Case \textsc{Pair}} \hfill \\
  From the assumption that $\jarg{C} \Gamma \vdash (V, W) : A \times B$, we have
  that $\jarg{C} \Gamma \vdash V : A$ and that $\jarg{C} \Gamma \vdash W :
  B$.

  By the inductive hypothesis (premise 1), we have that
  $\translatechty{\Gamma} \vdash \translatechval{V} :
  \translatechty{A}$ and that
  $\translatechty{\Gamma} \vdash \translatechval{W} :
  \translatechty{B}$.

  It follows by \textsc{Pair} that $\translatechty{\Gamma} \vdash
  (\translatechval{V}, \translatechval{W}) : (\translatechty{A} \times
  \translatechty{B})$ as required.

  \noindent
  \textbf{Cases \textsc{Inl}, \textsc{Inr}}
  Similar to \textsc{Pair}.

  \noindent
  \textbf{Case \textsc{Roll}} Immediate.

  \noindent
  \textbf{Premise 2} \hfill \\

  \noindent
  \textbf{Case \textsc{App}} \hfill \\
  From the assumption that $\jarg{C} \Gamma \vdash V \app W : B$, we have that
  $\jarg{C} \Gamma \vdash V : A \to B$ and $\jarg{C} \Gamma \vdash W : A$. By
  the inductive hypothesis (premise 1), we have that
  $\translatechty{\Gamma} \vdash \translatechval{V} : \translatechty{A}
  \lactto{\translatechty{C}} \translatechty{B}$
  and that
  $\translatechty{\Gamma} \vdash \translatechval{W} : \translatechty{A}$.

  By \textsc{App}, it follows that $\translatechty{\Gamma} \mid \translatechty{C} \vdash
  \translatechval{V} \app \translatechval{W} : \translatechty{B}$ as
  required.

  \noindent
  \textbf{Case \textsc{Return}} \hfill \\
  From the assumption that $\jarg{C} \Gamma \vdash \effreturn{V} : A$, we have that
  $\jarg{C} \Gamma \vdash V : A$. By the inductive hypothesis we have that
  $\translatechty{\Gamma} \vdash \translatechval{V} : \translatechty{A}$
  and thus by \textsc{Return} we can show that $\translatechty{\Gamma} \mid
  \translatechty{C} \vdash \effreturn{\translatechval{V}} :
  \translatechty{A}$ as required.

  \textbf{Case \textsc{EffLet}} \hfill \\
  From the assumption that $\jarg{C} \Gamma \vdash \efflet{x}{M}{N} : B$, we
  have that $\jarg{C} \Gamma \vdash M : A$ and $\jarg{C} \Gamma, x : A \vdash N
  : B$.

  By the inductive hypothesis (premise 2), we have that
  $\translatechty{\Gamma} \mid \translatechty{C} \vdash
  \translatechterm{M} : \translatechty{A}$ and $\translatechty{\Gamma}, x
  : \translatechty{A} \mid \translatechty{C} \vdash \translatechterm{N} :
  \translatechty{B}$. Thus by \textsc{EffLet} it follows that
  $\translatechty{\Gamma} \mid \translatechty{C} \vdash
  \efflet{x}{\translatechterm{M}}{\translatechterm{N}} : \translatechty{B}$.

  \noindent
  \textbf{Case \textsc{LetPair}}
  Similar to \textsc{EffLet}.

  \noindent
  \textbf{Case \textsc{Case}} \hfill \\
  From the assumption that $\jarg{C} \Gamma \vdash \caseof{V}{\inl{x} \mapsto M;
  \inr{y} \mapsto N}:B$, we have that $\jarg{C} \Gamma \vdash V : A + A'$, that
  $\jarg{C} \Gamma, x : A \vdash M : B$, and that $\jarg{C} \Gamma, y : A'
  \vdash N : B$.

  By the inductive hypothesis (premise 1) we have that
  $\translatechty{\Gamma} \vdash \translatechval{V} : \translatechty{A}
  + \translatechty{A'}$, and by premise 2 we have that
  $\translatechty{\Gamma}, x : \translatechty{A} \mid \translatechty{C}
  \vdash \translatechterm{M} : \translatechty{B}$ and
  $\translatechty{\Gamma}, y : \translatechty{A'} \mid
  \translatechty{C} \vdash \translatechterm{M} : \translatechty{B}$.

  By \textsc{Case}, it follows that $\translatechty{\Gamma} \mid
  \translatechty{C} \vdash \caseof{\translatechval{V}}{\inl{x} \mapsto
  \translatechterm{M}; \inr{y} \mapsto \translatechterm{N}} :
  \translatechty{B}$.

  \noindent
  \textbf{Case \textsc{Unroll}} Immediate.


  \noindent
  \textbf{Case \textsc{Fork}} \hfill \\
  From the assumption that $\jarg{A} \Gamma \vdash \lchfork{M} : \one$, we have that
  $\jarg{A} \Gamma \vdash M : \one$. By the inductive hypothesis, we have that
  $\translatechty{\Gamma} \mid \translatechty{A} \vdash \translatechterm{M} : \one$.

  We can show that $\translatechty{\Gamma} \mid \translatechty{A} \vdash
  \lactspawn{\translatechterm{M}} : \pid{\one}$ and also show that
  $\translatechty{\Gamma}, x : \pid{\one} \vdash \effreturn{()} : \one$.

  Thus by \textsc{EffLet}, it follows that $\translatechty{\Gamma} \mid
  \translatechty{A} \vdash
  \efflet{x}{\lactspawn{\translatechterm{M}}}{\effreturn{()}} : \one$ as required.

  \noindent
  \textbf{Case \textsc{Give}} \hfill \\
  From the assumption that $\jarg{A} \Gamma \vdash \lchgive{V}{W} : \one$, we
  have that $\jarg{A} \Gamma \vdash V : A$ and $\jarg{A} \Gamma \vdash W : \chanzero$. By the
  inductive hypothesis, we have that $\translatechty{\Gamma} \vdash
  \translatechval{V} : \translatechty{A}$ and $\translatechty{\Gamma}
  \vdash \translatechval{W} : \pid{\translatechty{A} + \pid{\translatechty{A}}}$. We can
  show that $\translatechty{\Gamma} \vdash \inl{\translatechval{V}} :
  \translatechty{A} + \pid{\translatechty{A}}$, and thus it follows that
  $\translatechty{\Gamma} \mid \translatechty{A} \vdash \lactsend{\inl{V}}{\translatechval{W}} :
  \one$ as required.

  \noindent
  \textbf{Case \textsc{Take}} \hfill \\
  From the assumption that $\jarg{A} \Gamma \vdash \lchtake{V} : A$, we have
  that $\jarg{A} \Gamma \vdash V : \chanzero$.

  By the inductive hypothesis (premise 1), we have that
  $\translatechty{\Gamma} \vdash \translatechval{V} :
  \pid{\translatechty{A} + \pid{\translatechty{A}}}$.

  We can show that:

  \begin{itemize}
    \item $\translatechty{\Gamma} \mid \translatechty{A} \vdash \lactself{}
      : \pid{\translatechty{A}}$
    \item $\translatechty{\Gamma} \mid \translatechty{A} \vdash \inr{\lactself{}}
      : \translatechty{A} + \pid{\translatechty{A}}$
    \item $\translatechty{\Gamma}, \textit{\textit{selfPid}} :
      \pid{\translatechty{A}} \mid \translatechty{A} \vdash
      \lactsend{\inr{\textit{\textit{selfPid}}}}{\translatechval{V}} : \one$
    \item $\translatechty{\Gamma}, \textit{\textit{selfPid}} :
      \pid{\translatechty{A}}, z : \one \mid \translatechty{A} \vdash
  \lactrecv{} : \translatechty{A}$
  \end{itemize}

  Thus by two applications of \textsc{EffLet} (noting that we desugar $M ; N$
  into $\efflet{z}{M}{N}$, where z is fresh), we arrive at:

  \[
    \translatechty{\Gamma} \mid \translatechty{A} \vdash
    \efflet{\textit{\textit{selfPid}}}{\lactself{}}{\lactsend{(\inr{\textit{\textit{selfPid}}})}{\translatechval{V}};
    \lactrecv{}} : \translatechty{A}
  \]

  as required.

  \noindent
  \textbf{Case \textsc{NewCh}} \hfill \\
  We have that $\jarg{A} \Gamma ; \Delta \vdash \lchnewch{} : \chanzero$.

  Our goal is to show that $\translatechty{\Gamma} \mid \translatechty{A} \vdash
  \lactspawn{\metadef{body} \app (\emptylist, \emptylist) } :
  \pid{\translatechty{A} + \pid{\translatechty{A}}}$.
  To do so amounts to showing that $\translatechty{\Gamma} \mid
  \translatechty{A} + \pid{\translatechty{A}} \vdash
  \metadef{body} \app (\emptylist, \emptylist) : \one$.

  We sketch the proof as follows. Firstly, by the typing of $\lactrecv{}$,
  \textit{recvVal} must have type $\translatechty{A} +
  \pid{\translatechty{A}}$. By inspection of both case branches and
  \textsc{LetPair}, we have that \textit{state} must have type
  $\listty{\translatechty{A}} \times \listty{\pid{\translatechty{A}}}$.

  We expect \metadef{drain} to have type
  \[
    \listty{\translatechty{A}} \times \listty{\pid{\translatechty{A}}}
  \lactto{\translatechty{A}}
  \listty{\translatechty{A}} \times \listty{\pid{\translatechty{A}}}
  \]
  since we use the returned value as a recursive call to $g$, which must have
  the same type of \textit{state}. Inspecting the case split in
  \metadef{drain}, we have that the empty list case for values returns the
  input state, which of course has the same type. The same can be said for the
  empty list case of the \textit{readers} case split.

  In the case where both the values and readers lists are non-empty, we have
  that $v$ has type $\translatechty{A}$ and \textit{pid} has type
  $\pid{\translatechty{A}}$. Additionally, we have that $vs$ has type
  $\listty{\translatechty{A}}$ and \textit{pids} has type
  $\listty{\pid{\translatechty{A}}}$. 
  We can show via send that $\lactsend{v}{\textit{pid}}$ has type $\one$.
  After desugaring $M ; N$ into an \textsc{EffLet}, we can show that the returned
  value is of the correct type.
\end{qedproof}
\begin{fake}{Theorem~\ref{thm:lch-lact-config-typing}}
If $\jarg{A} \Gamma ; \Delta \vdash \config{C}$ with $\Gamma \compat \Delta$,
then $\translatechty{\Gamma} ; \translatechty{\Delta} \vdash
\translatechconfig{\config{C}}$.
\end{fake}

\begin{qedproof}
  By induction on the derivation of $\jarg{A} \Gamma ; \Delta \vdash
  \config{C}$.

  \noindent
  \textbf{Case \textsc{Par}} \hfill \\
  From the assumption that $\jarg{A} \Gamma ; \Delta \vdash \config{C}_1 \parallel
  \config{C}_2$, we have that $\Delta$ splits as $\Delta_1, \Delta_2$ such that
  $\jarg{A} \Gamma ; \Delta_1 \vdash \config{C}_1$ and $\jarg{A} \Gamma ; \Delta_2
  \vdash \config{C}_2$.

  By the inductive hypothesis, we have that $\translatechty{\Gamma} ;
  \translatechty{\Delta_1} \vdash \translatechconfig{\config{C}_1}$ and
  $\translatechty{\Gamma} ; \translatechty{\Delta_2} \vdash
  \translatechconfig{\config{C}_2}$. By the definition of
  $\translatechty{-}$ on linear configuration environments, it follows that
  $\translatechty{\Delta_1}, \translatechty{\Delta_2} =
  \translatechty{\Delta}$. Consequently, by \textsc{Par}, we can show that
  $\translatechty{\Gamma} ; \translatechty{\Delta} \vdash
  \translatechconfig{C}_1 \parallel \translatechconfig{C}_2$ as required.

  \noindent
  \textbf{Case \textsc{Chan}} \hfill \\
  From the assumption that $\jarg{A} \Gamma ; \Delta \vdash (\nu a) \config{C}$,
  we have that $\jarg{A} \Gamma, a : \chan{A} ; \Delta, a : A \vdash
  \config{C}$. By the inductive hypothesis, we have that
  $\translatechty{\Gamma}, a : \pid{\translatechty{A} +
  \pid{\translatechty{A}}} ; \translatechty{\Delta}, a : \translatechty{A} +
  \pid{\translatechty{A}} \vdash \translatechconfig{C}$.

  By \textsc{Pid}, it follows that $\translatechty{\Gamma} ;
  \translatechty{\Delta} \vdash (\nu a) \translatechconfig{\config{C}}$ as
  required.

  \noindent
  \textbf{Case \textsc{Term}} \hfill \\
  From the assumption that $\jarg{A} \Gamma ; \cdot \vdash M$, we have that
  $\jarg{A} \Gamma \vdash M : \one$.

  By Lemma~\ref{lem:lch-lact-term-typing}, we have that
  $\translatechty{\Gamma} \mid \translatechty{A} \vdash
  \translatechterm{M} : \one$. By weakening, we can show that
  $\translatechty{\Gamma}, a : \pid{\translatechty{A}} \mid
  \translatechty{A} \vdash
  \translatechterm{M} : \one$.

  It follows that, by \textsc{Actor}, we can construct a configuration of the
  form $\translatechty{\Gamma}, a : \pid{\translatechty{A}} ; a :
  \translatechty{A} \vdash \actor{a}{\translatechterm{M}}{\epsilon}$, and by
  \textsc{Pid}, we arrive at

  \[
    \translatechty{\Gamma} ; \cdot \vdash (\nu a)
    (\actor{a}{\translatechterm{M}}{\epsilon})
  \]

  as required.

  \noindent
  \textbf{Case \textsc{Buf}} \hfill \\
  We assume that $\jarg{A} \Gamma ; a : A \vdash \lchbuf{a}{\seq{V}}$, and since
  $\Gamma \compat \Delta$, we have that we can write $\Gamma = \Gamma', a :
  \chan{A}$.

  From the assumption that $\jarg{A} \Gamma', a : \chan{A} ; a : A \vdash
  \lchbuf{a}{\seq{V}}$, we have that $\jarg{A} \Gamma', a : \chan{A} \vdash V_i : A$ for
  each $V_i \in \seq{V}$.

  By repeated application of Lemma~\ref{lem:lch-lact-term-typing}, we have that
  $\jarg{A} \translatechty{\Gamma'}, a : \pid{\translatechty{A} +
  \pid{\translatechty{A}}} \vdash \translatechval{V_i} :
  \translatechty{A}$
  for each $V_i \in \seq{V}$, and by \textsc{ListCons} and
  \textsc{EmptyList} we can construct a list $\translatechval{V_1} :: \ldots ::
  \translatechval{V_n} :: \emptylist$ with type $\listty{\translatechty{A}}$.

  Relying on our previous analysis of the typing of \metadef{body}, we have
  that $\translatechterm{\metadef{body}}$ has type:

  \[
    (\listty{\translatechty{A}} \times \listty{\pid{\translatechty{A}}})
    \lactto{\translatechty{A} + \pid{\translatechty{A}}} \one
  \]

  Thus it follows that

  \begin{align*}
    \translatechty{\Gamma'}, a : \pid{\translatechty{A} +
    \pid{\translatechty{A}}} \mid \translatechty{A} +
    \pid{\translatechty{A}} \vdash \\ \metadef{body} \app (
    \translatechval{\seq{V}}, \emptylist{}) : \one
  \end{align*}

  By \textsc{Actor}, we can see that

    \begin{align*}
    \translatechty{\Gamma'}, a : \pid{\translatechty{A} +
    \pid{\translatechty{A}}} ; a :  \translatechty{A} +
    \pid{\translatechty{A}} \vdash \\
    \actor{a}{\metadef{body} \app (\translatechval{\seq{V}},
    \emptylist{})}{\epsilon}
    \end{align*}

  as required.

\end{qedproof}

\begin{fake}{Theorem~\ref{thm:lch-lact-config-sim}} \hfill \\
  If $\jarg{A} \Gamma; \Delta \vdash \config{C}_1$, and $\config{C}_1 \ceval
  \config{C}_2$, then there exists some $\config{D}$ such that
  $\translatechconfig{\config{C}_1} \ceval^* \config{D}$ with $\config{D} \equiv
\translatechconfig{\config{C}_2}$.
\end{fake}

\begin{qedproof} \hfill \\
  By induction on the derivation of $\config{C}_1 \ceval \config{C}_2$.

  \noindent
\textbf{Case } \textsc{Give} \hfill \\

\[
\begin{array}{rl}
\text{Assumption} & E[\lchgive{W}{a}] \parallel \lchbuf{a}{\seq{V}} \\
\text{Definition of } \translatechconfig{-} & 
  (\nu b) (
    \actor{b}{\translatechterm{E}[\lactsend{(\termconst{inl} \; \translatechval{W})}{a}]}{\epsilon}) \parallel
    \actor{a}{\metadef{body} \; (\listterm{\translatechval{\seq{V}}}, \emptylist) }{\epsilon} \\
\equiv & 
  (\nu b) (
    \actor{b}{\translatechterm{E}[\lactsend{(\termconst{inl} \; \translatechval{W})}{a}]}{\epsilon} \parallel
    \actor{a}{\metadef{body} \; (\listterm{\translatechval{\seq{V}}}, \emptylist) }{\epsilon}) \\
\ceval \text{ (Send) } &
  (\nu b) (
    \actor{b}{\translatechterm{E}[\effreturn{()}]}{\epsilon} \parallel
    \actor{a}{\metadef{body} \; (\listterm{\translatechval{\seq{V}}}, \emptylist) }{\termconst{inl} \; \translatechval{W}}) \\
\\
\end{array}
\]

Now, let $G[-] = (\nu b) \actor{b}{\translatechterm{E}[\effreturn{()}]}{\epsilon} \parallel [-] )$.

We now have

\[
G[\actor{a}{\metadef{body} \; (\listterm{\translatechval{\seq{V}}}, \emptylist)}{(\termconst{inl} \; \translatechval{W})}]
\]

which we can expand to

\[
G[
\actor{a}{
  \begin{array}[t]{rl}
         & (\termconst{rec} \; g (\textit{state}) \; . \\
         & \quad \termconst{let} \; recvVal \Leftarrow \lactrecv \; \termconst{in} \\
         & \quad \termconst{let} \; (vals, readers) = state \app \termconst{in} \\
         & \qquad \caseofone{\textit{recvVal}} \\
         & \qquad \quad \inl{v} \mapsto
        {
        \begin{aligned}[t]
          & \efflettwo{\textit{newVals}}{\listappend{vals}{\listterm{v}}} \\
          & \termconst{let} \; state' \Leftarrow \metadef{drain} \app (\textit{newVals}, readers)
            \app \termconst{in} \\
          & \metadef{body} \app (state') \\
        \end{aligned}
        } \\
         & \qquad \quad \inr{\textit{pid}} \mapsto
        {
        \begin{aligned}[t]
          & \efflettwo{\textit{newReaders}}{\listappend{readers}{\listterm{pid}}} \\
          & \termconst{let} \; state' \Leftarrow \metadef{drain} \app (vals, \textit{newReaders}) \; \termconst{in} \\
          & \metadef{body} \app (\textit{state'})\}) \; (\translatechval{\seq{V}}, \emptylist) \\
        \end{aligned}
        }
\end{array}
}{\termconst{inl} \; \translatechval{W}}
]
\]

Applying the arguments to the recursive function $g$; performing the $\lactrecv$, and the $\termconst{let}$ and $\mkwd{case}$ reductions, we have:

\[
G[
\actor{a}{
        \begin{aligned}[t]
          & \efflettwo{\textit{newVals}}{\listappend{\translatechval{\seq{V}}}{\listterm{\translatechval{W}}}} \\
          & \termconst{let} \; state' \Leftarrow \metadef{drain} \app (\textit{newVals}, \emptylist{})
            \app \termconst{in} \\
          & \metadef{body} \app (state') \\
        \end{aligned}
}{\epsilon}
]
\]

Next, we reduce the append operation, and note that since we pass a state
without pending readers into $\metadef{drain}$, that the argument is
returned unchanged:

\[
G[
\actor{a}{
          \termconst{let} \; state' \Leftarrow
          \effreturn{(\listcons{\translatechval{\seq{V}}}{\translatechval{W} :: \emptylist})} \app \termconst{in} \app
          \metadef{body} \app state'
}{\epsilon}
]
\]

Next, we apply the let-reduction and expand the evaluation context:

\[
  (\nu b) (\actor{b}{\translatechconfig{E}[\effreturn{()}]}{\epsilon} \parallel \actor{a}{
            \metadef{body} \; (\listcons{\translatechval{\seq{V}}}{\translatechval{W} :: \emptylist})}{\epsilon})
\]

which is structurally congruent to

\[
  (\nu b) (\actor{b}{\translatechconfig{E}[\effreturn{()}]}{\epsilon}) \parallel \actor{a}{
            \metadef{body} \; (\listcons{\translatechval{\seq{V}}}{\translatechval{W} :: \emptylist})}{\epsilon}
\]

which is equal to

\[
\translatechconfig{E[\effreturn{()}] \parallel \lchbuf{a}{\translatechval{\seq{V} \cdot W}} }
\]

as required.

  \noindent
\textbf{Case } \textsc{Take} \hfill \\

\[
\begin{array}{rl}
\text{Assumption} & E[\lchtake{a}] \parallel \lchbuf{a}{W \cdot \seq{V}} \\
\text{Definition of } \translatechconfig{-} & (\nu b) (\actor{b}{\translatechterm{E}[
 { 
   \begin{aligned}[t]
    & \efflettwo{\textit{selfPid}}{\lactself} \\
    & \lactsend{(\termconst{inr} \; \textit{selfPid})}{a}; \\
    & \lactrecv{}]
   \end{aligned}
 }
}{\epsilon}) \parallel \actor{b}{\metadef{body} \; (
  \listcons{\translatechval{W}}{\translatechval{\seq{V}}}, \emptylist{}) }{\epsilon} \\
\equiv{} & (\nu b) (\actor{b}{\translatechterm{E}[
 { 
   \begin{aligned}[t]
    & \efflettwo{\textit{selfPid}}{\lactself} \\
    & \lactsend{(\termconst{inr} \; \textit{selfPid})}{a}; \\
    & \lactrecv{}]
   \end{aligned}
 }
}{\epsilon} \parallel \actor{b}{\metadef{body} \; (
  \listcons{\translatechval{W}}{\translatechval{\seq{V}}}, \emptylist{}) }{\epsilon}) \\
\ceval (\text{Self}) & (\nu b) (\actor{b}{\translatechterm{E}[
 { 
   \begin{aligned}[t]
    & \efflettwo{\textit{selfPid}}{\effreturn{b}} \\
    & \lactsend{(\termconst{inr} \; \textit{selfPid})}{a}; \\
    & \lactrecv{}]
   \end{aligned}
 }
}{\epsilon} \parallel \actor{b}{\metadef{body} \; (
  \listcons{\translatechval{W}}{\translatechval{\seq{V}}}, \emptylist{}) }{\epsilon}) \\
  \teval (\text{Let})& (\nu b) (\actor{b}{\translatechterm{E}[
 { 
   \begin{aligned}[t]
    & \lactsend{(\termconst{inr} \; b)}{a}; \\
    & \lactrecv{}]
   \end{aligned}
 }
}{\epsilon} \parallel \actor{b}{\metadef{body} \; (
  \listcons{\translatechval{W}}{\translatechval{\seq{V}}}, \emptylist{}) }{\epsilon}) \\
\ceval (\text{Send}) & (\nu b) (\actor{b}{\translatechterm{E}[\effreturn{()}; \lactrecv{}]}{\epsilon} \parallel 
  \actor{b}{\metadef{body} \; (
  \listcons{\translatechval{W}}{\translatechval{\seq{V}}}, \emptylist{}) }{(\termconst{inr} \; b)}) \\
\teval (\text{Let}) & (\nu b) (\actor{b}{\translatechterm{E}[\lactrecv{}]}{\epsilon} \parallel 
  \actor{b}{\metadef{body} \; (
  \listcons{\translatechval{W}}{\translatechval{\seq{V}}}, \emptylist{}) }{(\termconst{inr} \; b)}) \\
\end{array}
\]

Now, let $G[-] = (\nu b)(\actor{b}{\translatechterm{E}[\lactrecv{}]}{\epsilon} \parallel [-])$.

Expanding, we begin with:

\[
G[\actor{a}{
  {
    \begin{array}[t]{rl}
         & (\termconst{rec} \; g (\textit{state}) \; . \\
         & \quad \termconst{let} \; recvVal \Leftarrow \lactrecv \; \termconst{in} \\
         & \quad \termconst{let} \; (vals, readers) = state \app \termconst{in} \\
         & \qquad \caseofone{\textit{recvVal}} \\
         & \qquad \quad \inl{v} \mapsto 
         {
         \begin{aligned}[t]
           & \efflettwo{\textit{newVals}}{\listappend{vals}{\listterm{v}}} \\
           & \termconst{let} \; state' \Leftarrow \metadef{drain} \app (\textit{newVals}, readers)
             \app \termconst{in} \\
           & g \app (state') \\
         \end{aligned}
         } \\
         & \qquad \quad \inr{\textit{pid}} \mapsto
         {
         \begin{aligned}[t]
           & \efflettwo{\textit{newReaders}}{\listappend{readers}{\listterm{pid}}} \\
           & \termconst{let} \; state' \Leftarrow \metadef{drain} \app (vals, \textit{newReaders}) \; \termconst{in} \\
           & g \app (\textit{state'})\}) \; (\listcons{\translatechval{W}}{\translatechval{\seq{V}}}, \emptylist )\\
         \end{aligned}
         } \\ 
    \end{array}
  }
}{(\termconst{inr} \; b)}]
\]

Reducing the recursive function, receiving from the mailbox, splitting the pair, and then taking the second branch on the case statement, we have:

\[
G[\actor{a}{
  \begin{array}[t]{rl}
    & \efflettwo{\textit{newReaders}}{\listappend{\emptylist{}}{\listterm{b}}} \\
    & \efflettwo{state'}{\metadef{drain} \; (vals, \textit{newReaders})} \\
    & \metadef{body} \; state'
  \end{array}
}{\epsilon}]
\]

Reducing the list append operation, expanding $\metadef{drain}$, and re-expanding $G$, we have:

\[
(\nu b) ( \actor{b}{E[\lactrecv]}{\epsilon} \parallel 
  \actor{a}{
   { \begin{aligned}[t] 
     & \termconst{let} \; state' \Leftarrow ( \lambda x . \\
        & \quad \termconst{let} \; (vals, readers) \; = x \; \termconst{in} \\
       & \quad \termconst{case} \; vals \; \{ \\
        & \qquad \emptylist \mapsto \effreturn{(vals, readers)} \\
        & \qquad \listcons{v}{vs} \mapsto \\
       & \qquad \quad \termconst{case} \; readers \; \{ \\
        & \qquad \quad \quad \quad \emptylist{} \mapsto \effreturn{(vals, readers)}\\
        & \qquad \quad \quad \quad \listcons{pid}{pids} \mapsto 
           \begin{aligned}[t]
             & \lactsend{v}{pid} ; \\
             & \effreturn{(vs, pids)} \} \} ) \; (\listcons{\translatechval{W}}{\translatechval{\seq{V}}}, \listterm{b}) \; \termconst{in} 
           \end{aligned} \\ 
        & \metadef{body} \; state'
     \end{aligned}
   }
  }{\epsilon}
\]

Next, we reduce the function application, pair split, and the case statements:

\[
(\nu b) ( \actor{b}{E[\lactrecv]}{\epsilon} \parallel
  \actor{a}{
   { \begin{aligned}[t]
     & \termconst{let} \; state' \Leftarrow
           \begin{aligned}[t]
             & \lactsend{\translatechval{W}}{b} ; \\
             & \effreturn{(\translatechval{\seq{V}}, \emptylist{})})
           \end{aligned} \\ 
        & \metadef{body} \; state'
     \end{aligned}
   }
  }{\epsilon}
\]

We next perform the send operation, and thus we have:

\[
(\nu b) (\actor{b}{E[\lactrecv{}]}{\translatechval{W}} \parallel \actor{a}{\metadef{body} \; ({(\translatechval{\seq{V}}, \emptylist{})})}{\epsilon}
\]

Finally, we perform the $\lactrecv$ and apply a structural congruence to arrive at
\[
(\nu b) (\actor{b}{E[\translatechval{W}]}{\epsilon}) \parallel \actor{a}{\metadef{body} \; ({(\translatechval{\seq{V}}, \emptylist{})})}{\epsilon}
\]

which is equal to
\[
\translatechconfig{E[\effreturn{W}] \parallel \lchbuf{a}{\seq{V}}}
\]
as required.

  \noindent
\textbf{Case } \textsc{NewCh} \hfill \\

\[
\begin{array}{rl}
\text{Assumption} & E[\lchnewch{}] \\
\text{Definition of } \translatechconfig{-} & (\nu a) (\actor{a}{
  \translatechterm{E}[\lactwd{spawn} \; (
    {\metadef{body} \; (\emptylist{}, \emptylist{})} )]}{\epsilon}) \\
\ceval & (\nu a) (\nu b) (\actor{a}{\translatechterm{E}[\effreturn{b}]}{\epsilon} \parallel
\actor{b}{
         \metadef{body} \; (\emptylist{}, \emptylist{})
    }{\epsilon}) \\
\equiv & (\nu b) (\nu a)
    (\actor{a}{\translatechterm{E}[\effreturn{b}]}{\epsilon}) \parallel
\actor{b}{
         \metadef{body} \; (\emptylist{}, \emptylist{})
    }{\epsilon}) \\
= & \translatechconfig{(\nu b)(E[\effreturn{b}] \parallel \lchbuf{b}{\epsilon})} \\
\end{array}
\]

as required.

  \noindent
\textbf{Case } \textbf{Fork} \hfill \\

\[
\begin{array}{rl}
\text{Assumption} & E[\lchfork{M}] \\
\text{Definition of } \translatechconfig{-} & 
  (\nu a) (\actor{a}{\translatechterm{E}[\efflet{x}{\lactspawn{\translatechterm{M}}}{\effreturn{()}}]}{\epsilon}) \\
\ceval & (\nu a) (\nu b) (\actor{a}{\translatechterm{E}[\efflet{x}{\effreturn{b}}{\effreturn{()}}]}{\epsilon} 
  \parallel \actor{a}{\translatechterm{M}}{\epsilon}) \\
\teval & (\nu a) (\nu b) (\actor{a}{\translatechterm{E}[\effreturn{()}]}{\epsilon} 
  \parallel \actor{b}{\translatechterm{M}}{\epsilon}) \\
\equiv & (\nu a) (\actor{a}{\translatechterm{E}[\effreturn{()}]}{\epsilon}) \parallel 
  (\nu b) (\actor{b}{\translatechterm{M}}{\epsilon}) \\
= & \translatechconfig{E[\effreturn{()}] \parallel M}
\end{array}
\]

as required.

  \noindent
\textbf{Case } \textbf{Lift} Immediate by the inductive hypothesis.

  \noindent
\textbf{Case } \textbf{LiftV}
Immediate by Lemma~\ref{lem:lch-lact-term-sim}.

\end{qedproof}

\newpage
\subsection{Correctness of Selective Receive Translation}
\subsubsection{Translation preserves typing}

\begin{fake}{Theorem~\ref{thm:seltrans-term-typing} (Type Soundness (Translation from $\lact$ with selective receive into $\lact$))}
    If $\Gamma ; \Delta \vdash \config{C}$, then it is the case that for all $\config{D} \in \seltrans{\config{C}}$, then $\seltrans{\Gamma} ; \seltrans{\Delta} \vdash \seltrans{\config{D}}$.
\end{fake}

\begin{proof}

    By induction on the typing derivation of $\Gamma ; \Delta \vdash \config{C}$.

\noindent
    \textbf{Case} $\Gamma ; \Delta \vdash \config{C}_1 \parallel \config{C}_2$ \hfill \\
    From the assumption that
    \[
    \Gamma \vdash \Delta \vdash \config{C}_1 \parallel \config{C}_2
    \]
  we have that $\Delta$ splits as $\Delta_1, \Delta_2$ and we have that $\Gamma ; \Delta_1 \vdash \config{C}_1$ and $\Gamma ; \Delta_2 \vdash \config{C}_2$.

    By the inductive hypothesis, we have that all $\config{C}_1' \in \seltrans{\config{C}_1}$ are well-typed under $\seltrans{\Gamma} ; \seltrans{\Delta_1}$, and similarly for $\config{C'}_2$ under $\seltrans{\Gamma}$ and $\seltrans{\Delta_2}$.

    Thus we can conclude (since $\seltrans{\Delta_1}$ and $\seltrans{\Delta_2}$ remain disjoint) that
    \[
      \seltrans{\Gamma} ; \seltrans{\Delta} \vdash \config{C}_1' \parallel \config{C}_2'
    \]
    for all $\config{C}_1' \in \seltrans{\config{C}_1}$ and $\config{C}_2' \in \seltrans{\config{C}_2}$.

\noindent
    \textbf{Case} $\Gamma ; \Delta \vdash (\nu a) \config{C}$ \hfill \\
    From the assumption that $\Gamma ; \Delta \vdash (\nu a) \config{C}$, we have that
    \[
      \Gamma, a : \pid{\variant{\ell_i : A_i}_i} ; a : \variant{\ell_i : A_i}_i \vdash \config{C}
    \]
    By the inductive hypothesis, we have that for all $\config{D} \in \seltrans{\config{C}}$, it is the case that 
    \[
      \seltrans{\Gamma}, a : \pid{\variant{\ell_i : \seltrans{A_i}}_i} ; \seltrans{\Delta}, a : \variant{\ell_i : \seltrans{A_i}}_i \vdash \config{D}
    \]
    It follows by \textsc{Pid} that
    for all $\config{D} \in \seltrans{\config{C}}$, that $\seltrans{\Gamma} ; \seltrans{\Delta} \vdash (\nu a) \config{D}$ as required.

\noindent
    \textbf{Case} $\Gamma, a : \pid{\variant{\ell_i : A_i}_i} ; a :
    \variant{\ell_i : A_i}_i \vdash \actor{a}{M}{\seq{V}}$ \hfill \\

    From the
    assumption
    \[
      \Gamma, a : \pid{\variant{\ell_i : A_i}_i} ; a :
    \variant{\ell_i : A_i}_i \vdash \actor{a}{M}{\seq{V}}
    \]
    we have that
    \[
    \Gamma, a : \pid{\variant{\ell_i : A_i}_i} \mid \variant{\ell_i : A_i}_i
    \vdash M : \one
  \]
    and that
    \[
      \Gamma, a : \pid{\variant{\ell_i : A_i}_i} \vdash V_k : \variant{\ell_i : A_i}_i
    \]
    for all $V_k \in \seq{V}$.

    By Lemma~\ref{lem:seltrans-term-typing}, we have that
    \[
      \seltrans{\Gamma}, a : \pid{\variant{\ell_i : \seltrans{A_i}}_i}, \textit{mb} : \listty{\variant{\ell_i : \seltrans{A_i}_i}} \vdash \seltransterm{M}{\textit{mb}} : \one
    \]

    Again by
    Lemma~\ref{lem:seltrans-term-typing}, we have that
    \[
      \seltrans{\Gamma}, a :
    \pid{\variant{\ell_i : \seltrans{A_i}}_i} \vdash \seltrans{V_i}
    \]
    for each $V_i \in \seq{V}$.

    We need to show that regardless of the position of the partition between the save queue and mailbox,
    the resulting configuration is well typed.

    \underline{Subcase} $\config{D} = \actor{a}{\seltransterm{M}{\emptylist{}}}{\seltrans{\seq{V}}}$,

    By the substitution lemma it follows that
    \[
      \seltrans{\Gamma}, a : \pid{\variant{\ell_i : \seltrans{A_i}}_i} \mid \listty{\variant{\ell_i :
    \seltrans{A_i}_i}} \vdash \seltransterm{M}{\emptylist} : \one
    \]

    By
    \textsc{Actor}, we can show $\seltrans{\Gamma} ; a :
    \pid{\variant{\ell_i : \seltrans{A_i}}_i} ; a : \variant{\ell_i : \seltrans{A_i}}_i \vdash
    \actor{a}{\seltransterm{M}{\emptylist}}{\seq{V}}$ as required.

    For the remainder of the cases, we establish our result by induction on
    $i$.

    \textit{Base case} $i=1$:

    We can show that
    \[
      \seltrans{\Gamma}, a : \pid{\variant{\ell_i : \seltrans{A_i}}_i} \vdash \seltrans{V_1} : \variant{\ell_i : \seltrans{A_i}_i}
    \]
    and thus that
    \[
      \seltrans{\Gamma}, a : \pid{\variant{\ell_i : \seltrans{A_i}}_i} \vdash \listterm{\seltrans{V_1}}  : \listty{\variant{\ell_i : \seltrans{A_i}}_i}
    \]
    Similarly, we have that
    \[
      (\Gamma, a : \pid{\variant{\ell_i : \seltrans{A_i}}_i} \vdash \seltrans{V_j} : \variant{\ell_i : \seltrans{A_i}}_i)_{j \in 2 .. n}
    \]
    Let us refer to this sequence as $\seq{V'}$. Thus we can show:
    \[
      \seltrans{\Gamma}, a : \pid{\variant{\ell_i : \seltrans{A_i}}_i} ; a : \variant{\ell_i : \seltrans{A_i}_i} \vdash \actor{a}{\seltransterm{M}{\listterm{\seltrans{V_1}}}}{\seq{V'}}
    \]
    as required.

    \textit{Inductive case} $i = j + 1$:

    By the inductive hypothesis, we have that
    \[
      \seltrans{\Gamma}, a : \variant{\ell_i : \seltrans{A_i}}_i ; a : \variant{\ell_i : \seltrans{A_i}}_i \vdash \actor{a}{\seltransterm{M}{\seq{W_j}}}{\seq{W'_j}}
    \]
    where:
    \begin{itemize}
        \item $\seq{W_j} = \seltrans{V_1} :: \ldots :: \seltrans{V_j} :: \emptylist$
        \item $\seq{W'_j} = \seltrans{V_{j+1}} \cdot \ldots \cdot \seltrans{V_n}$
    \end{itemize}

    By \textsc{Actor}, we have that
    \[
      \seltrans{\Gamma} ; a : \variant{\ell_i : \seltrans{A_i}}_i \mid \variant{\ell_i : \seltrans{A_i}}_i \vdash \seltransterm{M}{\seq{W_j}} : \one
    \]
    and thus that
    \[
      \seltrans{\Gamma}, a : \pid{\variant{\ell_i : \seltrans{A_i}}_i} \vdash \seq{W_j} : \listty{\variant{\ell_i : \seltrans{A_i}}_i}
    \]
    We also have that
    \[
      (\seltrans{\Gamma}, a : \pid{\variant{\ell_i : \seltrans{A_i}}_i} \vdash \seltrans{V_k} : \variant{\ell_i : \seltrans{A_i}_i})_{k \in j + 1 .. n}
    \]
    Thus we can show that
    \[
      \seltrans{\Gamma}, a : \pid{\variant{\ell_i : \seltrans{A_i}}_i} \vdash \seltrans{V_1} :: \ldots :: \seltrans{V_j} :: \seltrans{V_{j + 1}} ::
    \emptylist : \listty{\variant{\ell_i : \seltrans{A_i}}_i}
    \]
    It also remains the case that
    \[
      (\seltrans{\Gamma}, a : \pid{\variant{\ell_i : \seltrans{A_i}_i}} \vdash \seltrans{V_k})_{k \in j + 2 .. n}
    \]

    Therefore we can show that
    \[
      \seltrans{\Gamma}, a : \pid{\variant{\ell_i :
    \seltrans{A_i}}_i} ; a : \variant{\ell_i : \seltrans{A_i}}_i \vdash
    \actor{a}{\seltransterm{M}{\seltrans{V_1} :: \ldots :: \seltrans{V_j} ::
    \seltrans{V_{j + 1}} :: \emptylist}}{\seltrans{V_{j + 2}} \cdot \ldots \cdot
    \seltrans{V_n}}
    \]
    completing the induction.
    \end{proof}

    \newpage

    \subsubsection{Simulation of selective receive}

    \begin{lemma}{Simulation ($\lact$ with selective receive in $\lact$---terms)}\label{lem:selrecv-term-sim}
      If $\Gamma \vdash M : A$ and $M \teval M'$, then given some $\alpha$, it is the case that $\seltrans{M}{\alpha} \teval^+ \seltrans{M'}{\alpha}$.
    \end{lemma}
    \begin{proof}
      A straightforward induction on $M \teval M'$.
    \end{proof}

    \begin{lemma}\label{lem:seltrans-pure-reduction}
      If $\Gamma \vdash_\textsf{P} M : A$ and $M \teval^* \effreturn{V}$,
      then $\seltransterm{M}{\textit{mb}} \teval^* \effreturn{(\seltrans{V}, \textit{mb})}$ for any \textit{mb}.
    \end{lemma}
    \begin{proof}
      By the definition of $\teval^{*}$, we have that there exists some reduction sequence
      $M \teval M_1 \teval \cdots \teval M_n \teval \effreturn{V}$.

      The result follows by repeated application of Lemma~\ref{lem:selrecv-term-sim}.
    \end{proof}

    \begin{lemma}\label{lem:branches-matchesany}
      Suppose $V = \variant{\ell = V'}$.

      If $\Gamma \vdash V : A$ and $\neg (\matchesany{(\seq{c}, V)})$, then
      \[
        \caseof{\variant{\ell = \seltrans{V'}}}{\metadef{branches}(\seq{c}, \textit{mb}, \textit{default})}
        \teval^+ \textit{default} \app \variant{\ell = \seltrans{V'}}.
      \]
    \end{lemma}
    \begin{proof} (Sketch.)
      For it to be the case that $\neg (\matchesany{(\seq{c}, V)})$, we must have that for all $c_i \in \seq{c}$ where
      $c_i = \variant{\ell_i = x_i} \when M_i \mapsto N_i$:

      \begin{itemize}
        \item $\ell_i \ne \ell$, or
        \item $\ell_i = \ell$, but $M_i \{ V' / x_i \} \teval^{+} \ffalse$.
      \end{itemize}

      Recall that
      \[
        \metadef{branches}(\seq{c}, \textit{mb}, \textit{default}) \defeq
        \metadef{patBranches}(\seq{c}, \textit{mb}, \textit{default})\cdot
        \metadef{defaultBranches}(\seq{c}, \textit{mb}, \textit{default})
      \]

      If $\ell \ne \ell_i$ for any pattern in $\seq{c}$, then there will be no corresponding branch in \metadef{patBranches},
      but there will be a branch $\variant{\ell = x} \mapsto \textit{default} \app (\ell = x)$ in \metadef{defaultBranches},
      reducing to the required result $\textit{default} \app \variant{\ell = \seltrans{V}}$.

      On the other hand, suppose we have a set of clauses $\seq{c'} \subseteq \seq{c}$ such that for each
      $(\variant{\ell' = x_j} \when M_j \mapsto N_j)_j \in \seq{c'}$,
      it is the case that $\ell' = \ell$. Then, it must be the case that for all $M_j$, that $M_j \{ V' / x_j \} \teval^+ \ffalse$.

      The result can now be established by induction on the length of $\seq{c'}$, by the definition of \metadef{ifPats}. For the
      base case, we have $\textit{default} \app \variant{\ell = V'}$ immediately, as required. For the inductive cases, we note
      that by Lemma~\ref{lem:seltrans-pure-reduction}, each guard $\seltransterm{M_j}{\textit{mb}} \teval^* \effreturn{(\ffalse, \textit{mb})}$,
      and thus that each \metadef{ifPats} clause will proceed inductively (taking the ``else'' branch) until the base case is reached.
    \end{proof}

    \begin{lemma}\label{lem:patbranches-match}
      Suppose:
      \begin{itemize}
        \item We have a consistent set of patterns $\seq{c} = \{ (\variant{\ell_a = x_a \when M_a}) \mapsto N_a \}_a$
        \item $\exists k, l. \forall i. i < k . \neg(\matchesany(\seq{c}, V_i)) \wedge \matches(c_l, V_k) \wedge \forall j . j < l \Rightarrow \neg(\matches(c_j, V_k))$.
        \item $\Gamma \vdash V_k : A$
        \item $V_k = \variant{\ell = V'_k}$
      \end{itemize}

      Then:
      \[
        \caseof{\variant{\ell_k = V'_k}}{ \metadef{branches}(\seq{c}, \textit{mb}, \textit{default}}) \teval^{+} (\seltransterm{N_l}{mb}) \{ \seltrans{ V'_k} / x_l \}
      \]

    \end{lemma}
    \begin{proof}(Sketch.)
      Recall that
      \[
        \metadef{branches}(\seq{c}, \textit{mb}, \textit{default}) \defeq
        \metadef{patBranches}(\seq{c}, \textit{mb}, \textit{default})\cdot
        \metadef{defaultBranches}(\seq{c}, \textit{mb}, \textit{default})
      \]

      If $\matches(c_l, V_k)$, then it must be the case that $c_l$ is of the form $\{ \variant{\ell_k =x_l} \when M_l \mapsto N_l\}$,
      with $M_l \{ V'_k / x_l \} \teval^{*} \ttrue$.

      By the definition of \metadef{patBranches}, we know that there exists an ordered list $\seq{c_\ell}$ where $\seq{c_\ell} \subseteq \seq{c}$,
      and $\metadef{label}(c) = \ell_k$; therefore we have that $c_l \in \seq{c_\ell}$.

      Again, we proceed by induction on the length of $\seq{c_\ell}$.
      We know that $\forall j . j < l \Rightarrow \neg(\matches(c_j, V_k))$, and therefore that $M_j \teval^{*} \effreturn{\ffalse}$.
      Since $\seq{c_\ell}$ is ordered, by Lemma~\ref{lem:seltrans-pure-reduction}, we have that $\seltransterm{M}{mb} \{ \seltrans{V_k} / x_j \} \teval^{*} \effreturn{(\ffalse, \textit{mb})}$, and
      proceed inductively.
      For $\metadef{ifPats}(mb, \ell, y, c_l \cdot \textit{pats}, \textit{default})$, it will be the case (again by Lemma~\ref{lem:seltrans-pure-reduction}) that
      $\seltrans{M_l}{\textit{mb}} \teval \effreturn{(\ttrue, \textit{mb})}$.
      Reducing, we take the $\ttrue$ branch of the if-statement, and arrive at $(\seltrans{N_l}{\textit{mb}})\{\seltrans{V_k} / x_l\}$ as required.

    \end{proof}

    \begin{fake}{Theorem~\ref{thm:seltrans-simulation} (Simulation ($\lact$ with selective receive in $\lact$))}
      If $\Gamma ; \Delta \vdash \config{C}$ and $\config{C} \ceval \config{C'}$,
      then $\forall \config{D} \in \seltrans{\config{C}}$, there exists a
      $\config{D'}$ such that $\config{D} \ceval^+ \config{D'}$.
    \end{fake}

    \begin{proof}
      By induction on the derivation of $\config{C} \ceval \config{C'}$.

      The configuration of the save queue and mailbox are irrelevant to all
      reductions except \textbf{Receive}, so we will show the case where the
      save queue is empty; the other cases follow exactly the same pattern.

      \noindent
      \textbf{Case Spawn} \hfill \\
      We assume that $\Gamma ; \Delta \vdash \actor{a}{E[\lactspawn{M}]}{\seq{V}}$.

      \[
        \begin{array}{lll}
          \text{Assumption} & \actor{a}{E[\lactspawn{M}]}{\seq{V}} & (1) \\
          \text{Def.\ } \seltrans{-} &
            \actor{a}{\seltrans{E}[
               \efflet{\textit{spawnRes}}
                      {\lactspawn{(\seltransterm{M}{\emptylist})}}
                      {\effreturn{(\textit{spawnRes}, \emptylist)}}]}{\seq{V}} & (2) \\
          \ceval &
            (\nu b)
               \actor{a}{\seltrans{E}[
               \efflet{\textit{spawnRes}}
                      {\effreturn{b}}
                      {\effreturn{(\textit{spawnRes}, \emptylist)}}]}{\seq{V}} \\
               & \qquad \parallel \actor{b}{\seltrans{M}{\emptylist}}{\epsilon}  & (3) \\
          \teval &
            (\nu b)
               \actor{a}{\seltrans{E}[\effreturn{(\textit{b}, \emptylist)]}}{\seq{V}}
                 \parallel \actor{b}{\seltrans{M}{\emptylist}}{\epsilon} & (4) \\
               = &
            (\nu b)
               \actor{a}{E[\effreturn{b}]}{\seq{V}} \parallel \actor{b}{M}{\epsilon} & (5) \\
        \end{array}
      \]

      \noindent
      \textbf{Case Send} \hfill \\
      We assume that $\Gamma ; \Delta \vdash \actor{a}{E[\lactsend{V'}{b}]}{\seq{V}} \parallel \actor{b}{M}{\seq{W}}$.

      \[
        \begin{array}{lll}
          \text{Assumption} & \actor{a}{E[\lactsend{V'}{b}]}{\seq{V}} \parallel \actor{b}{M}{\seq{W}} & (1) \\
          \text{Def.\ } \seltrans{-} & \actor{a}{\seltrans{E}[\efflet{x}{\lactsend{\seltrans{V'}}{b}}{(x, \emptylist)}]}{\seltrans{\seq{V}}} \parallel
            \actor{b}{\seltransterm{M}{\emptylist}}{\seltrans{\seq{W}}} & (2) \\
          \ceval & \actor{a}{\seltrans{E}[\efflet{x}{\effreturn{()}}{(x, \emptylist)}]}
          {\seltrans{\seq{V}}} \parallel \actor{b}{\seltransterm{M}{\emptylist}}{\seltrans{\seq{W}} \cdot \seltrans{V'}} & (3) \\
          \teval & \actor{a}{\seltrans{E}[\effreturn{((), \emptylist)}]}
          {\seltrans{\seq{V}}} \parallel \actor{b}{\seltransterm{M}{\emptylist}}{\seltrans{\seq{W}} \cdot \seltrans{V'}} & (4) \\
          = & \actor{a}{E[\effreturn{()}]}{\seq{V}} \parallel \actor{b}{M}{\seq{W} \cdot V'}

        \end{array}
      \]

      \noindent
      \textbf{Case SendSelf} \hfill \\
      We assume that $\Gamma ; \Delta \vdash \actor{a}{E[\lactsend{V}{a}]}{\seq{V}}$.

      \[
        \begin{array}{lll}
          \text{Assumption} & \actor{a}{E[\lactsend{W}{a}]}{\seq{V}} & (1) \\
          \text{Def.\ } \seltrans{-} & \actor{a}{\seltrans{E}[\efflet{x}{\lactsend{\seltrans{W}}{a}}{\effreturn{(x, \emptylist)}}}{\seltrans{\seq{V}}} & (2) \\
          \ceval & \actor{a}{\seltrans{E}[\efflet{x}{\effreturn{()}}{\effreturn(x, \emptylist)}]}{\seltrans{\seq{V}} \cdot \seltrans{W}} & (3) \\
          \teval & \actor{a}{\seltrans{E}[\effreturn{((), \emptylist)}]}{\seltrans{\seq{V}} \cdot \seltrans{W}} & (4) \\
          = & \actor{a}{E[\effreturn{()}]}{\seq{V} \cdot W} & (5) \\
        \end{array}
      \]

      \noindent
      \textbf{Case Self} \hfill \\
      We assume that $\Gamma ; \Delta \vdash \actor{a}{E[\lactself]}{\seq{V}}$.

      \[
        \begin{array}{lll}
          \text{Assumption} & \actor{a}{E[\lactself]}{\seq{V}} & (1) \\
          \text{Def.\ } \seltrans{-} & \actor{a}{(\seltrans{E}[\efflet{\textit{selfPid}}{\lactself}{\effreturn{(\textit{selfPid}, \emptylist)}}]) }{\seltrans{\seq{V}}} & (2) \\
          \ceval & \actor{a}{(\seltrans{E}[\efflet{\textit{selfPid}}{\effreturn{a}}{\effreturn{(\textit{selfPid}, \emptylist)}}]) }{\seltrans{\seq{V}}} & (3) \\
          \teval & \actor{a}{(\seltrans{E}[\effreturn{(a, \emptylist)}]) }{\seltrans{\seq{V}}} & (4) \\
          = & \actor{a}{E[\effreturn{a}]}{\seq{V}} & (5) \\
        \end{array}
      \]

      \noindent
      \textbf{Case Lift} \hfill \\
      Immediate by the inductive hypothesis.

      \noindent
      \textbf{Case LiftM} \hfill \\
      Immediate by Lemma~\ref{lem:selrecv-term-sim}.

      \noindent
      \textbf{Case Receive} \hfill \\
      The interesting case is the translation of selective receive. We sketch the proof as follows.

      We assume that $\Gamma ; \Delta \vdash \actor{a}{E[\selrecv{\seq{c}}]}{\seq{V}}$, where:
      \begin{itemize}
        \item $\seq{V} = V_1 \cdot \ldots \cdot V_n$
        \item $\seq{c} = \{ \variant{\ell_i = x_i} \when M_i \teval^{*} N_i \}_i$
      \end{itemize}

      By the assumptions of the reduction rule, we also have that there exists
      some $V_k$ and $c_l$ such that:
      \begin{itemize}
        \item For all $i < k$, $V_i$ does not match any pattern in $\seq{c}$.
        \item $V_k$ matches $c_l$.
        \item For all $j < l$, it is the case that $V_k$ does not match $c_k$.
      \end{itemize}

      Thus intuitively, $V_k$ is the first value in the mailbox to match any of the patterns.

      By the definition of $\seltrans{-}$, we have several possible
      configurations of the save queue and the mailbox to consider. More
      specifically:

      \begin{itemize}
        \item $\actor{a}{\metadef{find}(\seq{c}, \emptylist)}{\seltrans{\seq{V}}}$
        \item $\{ \actor{a}{\metadef{find}(\seq{c}, \seq{W_i^1})}{\seq{W_i^2}} \mid 1..n \}$,
          where $\seq{W_i^1} = \seltrans{V_1} :: \ldots :: \seltrans{V_i} :: \emptylist$, and
          $\seq{W_i^2} = \seltrans{V_{i + 1}} \cdot \ldots \cdot \seltrans{V_n}$,
      \end{itemize}

      \underline{Subcase 1}
      We have $\actor{a}{\metadef{find}(\seq{c}, \emptylist)}{\seltrans{\seq{V_1}} \cdot \ldots \cdot V_k \cdot \ldots \cdot \seq{V_2}}$.
      Expanding \metadef{find}, we have:

      \[
        \actor{a}{E[
          \begin{aligned}[t]
          & (\rectwo{\textit{findLoop}}{\textit{ms}} \\
          & \quad \letintwo{(mb_1, mb_2)}{ms} \\
          & \quad \caseofone{mb_2} \\
          & \qquad \emptylist{} \mapsto \transloop{\seq{c}}{\textit{mb}_1} \\
          & \qquad \listcons{x}{mb_2'} \mapsto \\
          & \qquad \quad \efflettwo{mb'}{mb_1 \doubleplus mb_2'} \\
          & \qquad \quad \caseofone{x} \metadef{branches}(\seq{c}, \textit{mb}', \\
          & \qquad \qquad
          \begin{aligned}[t]
            & \lambda y . (\efflettwo{\textit{mb}_1'}{\textit{mb}_1 \doubleplus \listterm{y}} \\
            & \qquad \textit{findLoop} \app (\textit{mb}_1', \textit{mb}_2')) ) \})
            \app (\emptylist{}, \emptylist{})
          \end{aligned}
        \end{aligned}
      ]}{\seltrans{\seq{V_1}} \cdot \ldots \cdot \seltrans{V_k} \cdot \ldots \cdot \seltrans{\seq{V_2}}}
      \]

      After reducing the application of the recursive function, the pair deconstruction, and the case split, we have:

      \[
        \actor{a}{E[\transloop{\seq{c}}{\emptylist{}}]}
          {\seltrans{\seq{V_1}} \cdot \ldots \cdot \seltrans{V_k} \cdot \ldots \cdot \seltrans{\seq{V_2}}}
      \]

      Expanding \metadef{loop} and \metadef{branches}, we have:

      \[
        \actor{a}{E[
          \begin{aligned}[t]
            & \quad (\rectwo{\textit{recvLoop}}{\textit{mb}} \\
            & \qquad \efflettwo{x}{\lactrecv} \\
            & \qquad \caseofone{x} \\
            & \begin{aligned}[t]
              & \qquad \quad \metadef{patBranches}(\seq{c}, \textit{mb}, \textit{default}) \: \cdot \\
              & \qquad \quad \metadef{defaultBranches}(\seq{c}, \textit{mb}, \textit{default}) \\
        & \qquad \}) \app \emptylist{} \\
            \end{aligned}
          \end{aligned}
        ]}{\seltrans{\seq{V_1}} \cdot \ldots \cdot \seltrans{V_k} \cdot \seltrans{\seq{V_2}}}
      \]

      where $\textit{default} = \lambda y . \efflet{\textit{mb'}}{\textit{mb} \doubleplus \listterm{y}}{\textit{recvLoop} \app \textit{mb}'}$.

      Recall that $\seltrans{\seq{V_1}} = \seltrans{V_1} \cdot \ldots \cdot \seltrans{V_{k - 1}}$,
      and we have that $\neg(\matchesany{(\seq{c}, V_i)}$ for all $i < k$. By Lemma~\ref{lem:branches-matchesany}, we have
      that
      \[
        \caseof{V_i}{\metadef{branches}(\seq{c}, \textit{mb})} \teval^* \textit{default} \app V_i
      \]
      In this case, \textit{default} appends the current value onto the end of the save queue, and calls \textit{recvLoop}
      recursively. Thus, proceeding inductively, we can show that we arrive at
      \[
        \actor{a}{E[\metadef{loop}(\seq{c}, \seltrans{V_1}:: \ldots :: \seltrans{V_{k-1}} :: \emptylist)]}
          {\seltrans{V_k} \cdot \seltrans{\seq{V_2}}}
      \]

      Expanding, we see

      \[
        \actor{a}{E[
          \begin{aligned}[t]
            & \quad (\rectwo{\textit{recvLoop}}{\textit{mb}} \\
            & \qquad \efflettwo{x}{\lactrecv} \\
            & \qquad \caseofone{x} \\
            & \begin{aligned}[t]
              & \qquad \quad \metadef{patBranches}(\seq{c}, \textit{mb}, \textit{default}) \: \cdot \\
              & \qquad \quad \metadef{defaultBranches}(\seq{c}, \textit{mb}, \textit{default}) \\
        & \qquad \}) \app (\seltrans{V_1}:: \ldots :: \seltrans{V_{k-1}} :: \emptylist) \\
            \end{aligned}
          \end{aligned}
        ]}{\seltrans{V_k} \cdot \seltrans{\seq{V_2}}}
      \]

      Reducing the function application and the receive from the mailbox:

      \[
        \actor{a}{E[\caseof{\seltrans{V_k}}{\metadef{branches}(\seq{c},(\seltrans{V_1}:: \ldots :: \seltrans{V_{k-1}} :: \emptylist), \textit{default})}]}{\seltrans{\seq{V_2}}}
      \]

      Now, by Lemma~\ref{lem:patbranches-match}, we have that
      \[
        \actor{a}{E[(\seltransterm{N_l}{\seltrans{V_1}:: \ldots :: \seltrans{V_{k-1}}}) \{ \seltrans{V'_k} / x_l \}]}{\seltrans{\seq{V_2}}}
      \]

      where $\seltrans{V_k} = \variant{\ell_k = \seltrans{V'_k}}$.

      It follows that this is in $\seltrans{\actor{a}{E[N_l \{ \seltrans{V'_k} / x_l\}]}{\seltrans{\seq{V_1}} \cdot \seltrans{\seq{V_2}}}}$, concluding this subcase.

      \underline{Subcase 2}:
      In this subcase, we examine the case where the save queue is initially non-empty.
      Pick some arbitrary $j \in 1 \ldots n$. Let:

      \begin{itemize}
        \item $\seltrans{\seq{W_1}} = \seltrans{V_1} :: \ldots :: \seltrans{V_j} :: \emptylist$
        \item $\seltrans{\seq{W_2}} = \seltrans{V_{j+1}} \cdot \ldots \cdot \seltrans{V_n}$
      \end{itemize}

      We have $\actor{a}{E[\metadef{find}(\seq{c}, \listterm{\seltrans{\seq{W_1}}})]}{\seltrans{\seq{W_2}}}$.
      We now have two possible scenarios: $k \le j$, or $k > j$.

      Expanding \metadef{find}, we have:

      \[
        \actor{a}{E[
          \begin{aligned}[t]
          & (\rectwo{\textit{findLoop}}{\textit{ms}} \\
          & \quad \letintwo{(mb_1, mb_2)}{ms} \\
          & \quad \caseofone{mb_2} \\
          & \qquad \emptylist{} \mapsto \transloop{\seq{c}}{\textit{mb}_1} \\
          & \qquad \listcons{x}{mb_2'} \mapsto \\
          & \qquad \quad \efflettwo{mb'}{mb_1 \doubleplus mb_2'} \\
          & \qquad \quad \caseofone{x} \metadef{branches}(\seq{c}, \textit{mb}', \\
          & \qquad \qquad
          \begin{aligned}[t]
            & \lambda y . (\efflettwo{\textit{mb}_1'}{\textit{mb}_1 \doubleplus \listterm{y}} \\
            & \qquad \textit{findLoop} \app (\textit{mb}_1', \textit{mb}_2')) ) \})
            \app (\emptylist{}, \seltrans{\seq{W_1}})
          \end{aligned}
        \end{aligned}
      ]}{\seltrans{\seq{W_2}}}
      \]

      Define \metadef{fLoop} as follows:
      \[
        \metadef{fLoop}(\textit{mb}_1, \textit{mb}_2) \defeq
        \begin{aligned}[t]
          & (\rectwo{\textit{findLoop}}{\textit{ms}} \\
          & \quad \letintwo{(mb_1, mb_2)}{ms} \\
          & \quad \caseofone{mb_2} \\
          & \qquad \emptylist{} \mapsto \transloop{\seq{c}}{\textit{mb}_1} \\
          & \qquad \listcons{x}{mb_2'} \mapsto \\
          & \qquad \quad \efflettwo{mb'}{mb_1 \doubleplus mb_2'} \\
          & \qquad \quad \caseofone{x} \metadef{branches}(\seq{c}, \textit{mb}', \\
          & \qquad \qquad
          \begin{aligned}[t]
            & \lambda y . (\efflettwo{\textit{mb}_1'}{\textit{mb}_1 \doubleplus \listterm{y}} \\
            & \qquad \metadef{fLoop} \app (\textit{mb}_1', \textit{mb}_2')) ) \})
            \app (\textit{mb}_1, \textit{mb}_2)
          \end{aligned}
        \end{aligned}
      \]

      After reducing the application of the recursive function and the pair deconstruction we have:

      \underline{Subsubcase} $k \le j$:

      Our reasoning is similar to in the first subcase.
      We appeal to repeated applications of Lemma~\ref{lem:branches-matchesany} and arrive at

      \[
        \actor{a}{\seltrans{E}[\metadef{fLoop}(\seltrans{\seq{W'_1}}, \seltrans{V_k} :: \seltrans{\seq{W''_1}})]}{\seq{\seltrans{W_2}}}
      \]

      where $\seltrans{\seq{W'_1}} = \seltrans{V_1} :: \ldots :: \seltrans{V_{k - 1}}$ and
      $\seltrans{\seq{W''_2}} = \seltrans{V_{k + 1}} :: \ldots :: \seltrans{V_{j}}$.

      Expanding, we have
      \[
        \actor{a}{E[
          \begin{aligned}[t]
          & (\rectwo{\textit{findLoop}}{\textit{ms}} \\
          & \quad \letintwo{(mb_1, mb_2)}{ms} \\
          & \quad \caseofone{mb_2} \\
          & \qquad \emptylist{} \mapsto \transloop{\seq{c}}{\textit{mb}_1} \\
          & \qquad \listcons{x}{mb_2'} \mapsto \\
          & \qquad \quad \efflettwo{mb'}{mb_1 \doubleplus mb_2'} \\
          & \qquad \quad \caseofone{x} \metadef{branches}(\seq{c}, \textit{mb}', \\
          & \qquad \qquad
          \begin{aligned}[t]
            & \lambda y . (\efflettwo{\textit{mb}_1'}{\textit{mb}_1 \doubleplus \listterm{y}} \\
            & \qquad \textit{findLoop} \app (\textit{mb}_1', \textit{mb}_2')) ) \})
            \app (\seltrans{\seq{W'_1}}, \seltrans{V_k} :: \seltrans{\seq{W''_1}})
          \end{aligned}
        \end{aligned}
      ]}{\seltrans{\seq{W_2}}}
      \]

      Reducing the recursive function application, pair split, and case expression, we have:
      \[
        \actor{a}{E[\caseof{\seltrans{V_k}}{\metadef{branches}(\seq{c}, \listappend{\seltrans{\seq{W'_1}}}{\seltrans{\seq{W''_1}}}, \textit{default}}]}{\seltrans{\seq{W_2}}}
      \]

      where
      \[
        \textit{default} =
        \begin{aligned}[t]
            & \lambda y . (\efflettwo{\textit{mb}_1'}{\seltrans{\seq{W'_1}} \doubleplus \listterm{y}} \\
            & \qquad \metadef{fLoop} \app (\textit{mb}_1', \seltrans{\seq{W''_1}}) ) \})
          \end{aligned}
      \]

      Now, by Lemma~\ref{lem:patbranches-match}, we have

      \[
        \actor{a}{\seltrans{E}[(\seltransterm{N_l}{(\listappend{\seltrans{\seq{W'_1}}}{\seltrans{\seq{W''_1}}})) \{ \seltrans{V'_k} / x_l \}}]}{\seltrans{\seq{W_2}}}
      \]

      which is in
      \[
        \seltrans{\actor{a}{E[N_l \{ V'_k / x_l \}]}{\seq{V_1} \cdot \seq{V_2}}}
      \]

      as required.

      \underline{Subsubcase} $k > j$:

      As before, we appeal to repeated applications of Lemma~\ref{lem:branches-matchesany}, this time
      arriving at:

      \[
        \actor{a}{\seltrans{E}[\metadef{fLoop}(\seltrans{\seq{W_1}}, \emptylist)]}{\seq{\seltrans{W_2}}}
      \]

      Expanding, we have:

      \[
        \actor{a}{E[
          \begin{aligned}[t]
          & (\rectwo{\textit{findLoop}}{\textit{ms}} \\
          & \quad \letintwo{(mb_1, mb_2)}{ms} \\
          & \quad \caseofone{mb_2} \\
          & \qquad \emptylist{} \mapsto \transloop{\seq{c}}{\textit{mb}_1} \\
          & \qquad \listcons{x}{mb_2'} \mapsto \\
          & \qquad \quad \efflettwo{mb'}{mb_1 \doubleplus mb_2'} \\
          & \qquad \quad \caseofone{x} \metadef{branches}(\seq{c}, \textit{mb}', \\
          & \qquad \qquad
          \begin{aligned}[t]
            & \lambda y . (\efflettwo{\textit{mb}_1'}{\textit{mb}_1 \doubleplus \listterm{y}} \\
            & \qquad \textit{findLoop} \app (\textit{mb}_1', \textit{mb}_2')) ) \})
            \app (\seltrans{\seq{W_1}}, \emptylist)
          \end{aligned}
        \end{aligned}
      ]}{\seltrans{\seq{W_2}}}
      \]

      Evaluating the application of the recursive function, the pair split, and
      the case statement, we have:

      \[
        \actor{a}{E[\transloop{\seq{c}}{\seltrans{\seq{W_1}]}]}}{\seltrans{\seq{W_2}}}
      \]

      Expanding $\metadef{loop}$, we have:

      \[
        \actor{a}{\seltrans{E}[
           \begin{aligned}[t]
                & \quad (\rectwo{\textit{recvLoop}}{\textit{mb}} \\
                & \qquad \efflettwo{x}{\lactrecv} \\
                & \qquad \caseofone{x} \\
                & \begin{aligned}[t]
                  & \qquad \quad \metadef{patBranches}(\seq{c}, \textit{mb}, \textit{default}) \: \cdot \\
                  & \qquad \quad \metadef{defaultBranches}(\seq{c}, \textit{mb}, \textit{default}) \\
         & \qquad \}) \app \seltrans{\seq{W_1}} \\
                \end{aligned}
              \end{aligned}
            ]}{\seltrans{\seq{W_2}}}
      \]

      with
      \[
        \textit{default} = \lambda y . \efflet{\textit{mb}'}{\listappend{\textit{mb}}{\listterm{y}}}{\textit{recvLoop} \app \textit{mb}'}
      \]

      Again by repeated applications of Lemma~\ref{lem:branches-matchesany}, we arrive at

      \[
        \actor{a}{\seltrans{E}[\transloop{\seq{c}}{\listappend{\seltrans{\seq{W_1}}}{\seltrans{\seq{W'_2}}}}]}{\seltrans{V_k} \cdot \seltrans{\seq{W''_2}}}
      \]

      where $\seltrans{\seq{W'_2}} = \seltrans{V_{j + 1}} :: \ldots :: \seltrans{V_{k - 1}}$,
      and $\seltrans{\seq{W''_2}} = \seltrans{V_{k + 1}} \cdot \ldots \cdot \seltrans{V_n}$.

      Now, expanding again, we have:

      \[
         \actor{a}{\seltrans{E}[
           \begin{aligned}[t]
                & \quad (\rectwo{\textit{recvLoop}}{\textit{mb}} \\
                & \qquad \efflettwo{x}{\lactrecv} \\
                & \qquad \caseofone{x} \\
                & \begin{aligned}[t]
                  & \qquad \quad \metadef{patBranches}(\seq{c}, \textit{mb}, \textit{default}) \: \cdot \\
                  & \qquad \quad \metadef{defaultBranches}(\seq{c}, \textit{mb}, \textit{default}) \\
         & \qquad \}) \app (\listappend{\seltrans{\seq{W_1}}}{\seltrans{\seq{W'_2}}}) \\
                \end{aligned}
              \end{aligned}
            ]}{\seltrans{V_k} \cdot \seltrans{\seq{W''_2}}}
      \]

      reducing the recursive function, receive, and case expression, we have:

      \[
    \actor{a}{\seltrans{E}[
            \caseof{\seltrans{V_k}}{\metadef{branches}(\seq{c}, (\listappend{\seltrans{\seq{W_1}}}{\seltrans{\seq{W'_2}}}), \textit{default} )}
            ]}{\seltrans{\seq{W''_2}}}
      \]

      Now, by Lemma~\ref{lem:patbranches-match}, we have

      \[
        \actor{a}{\seltrans{E}[(\seltransterm{N_l}{(\listappend{\seltrans{\seq{W_1}}}{\seltrans{\seq{W'_1}}})) \{ \seltrans{V'_k} / x_l \}}]}{\seltrans{\seq{W''_2}}}
      \]

      which is in
      \[
        \seltrans{\actor{a}{E[N_l \{ V'_k / x_l \}]}{\seq{V_1} \cdot \seq{V_2}}}
      \]

      as required.
\end{proof}

\end{document}